\def\preprint{1}
\def\shortver{0}
\def\llncs{0}
\def\focs{0}
\def\proceeding{0}

\def\draft{0} 
\def\anonymous{0}
\def\widemargin{0}
\def\toc{0} 
\def\qcrypt{0}
\ifnum\preprint=1 
    \ifnum\draft=1
        \documentclass[11pt,draft,letterpaper]{article}
        \pdfoutput=1
    \else
        \documentclass[11pt,final,letterpaper]{article}
        \pdfoutput=1
    \fi
    \usepackage{lmodern}
\fi
\ifnum\llncs=1
    \ifnum\draft=1
        \documentclass[runningheads,draft]{llncs}
        \renewcommand{\widemargin}{1}
    \else
        \documentclass[runningheads]{llncs}
        \renewcommand{\widemargin}{1}
    \fi
\fi
\usepackage{authblk}

\input{./tex/base.sty}

\ifnum \preprint=1
    \ifnum \focs=0
    \else 
        \usepackage[margin=1in]{geometry}
    \fi
\fi

\title{
The power of a single Haar random state: constructing and separating quantum pseudorandomness
}
\ifnum\llncs=1
    \titlerunning{The power of a single Haar random state}
\fi
\date{\today}
\usepackage{orcidlink}

\begin{document}
\ifnum\anonymous=0
    \ifnum\preprint=1
        \author[1]{Boyang Chen\, \orcidlink{0009-0000-1043-0977}}
        \author[2]{Andrea Coladangelo}
        \author[3]{Or Sattath}
        \affil[1]{Institute for Interdisciplinary Information Sciences, Tsinghua University}
        \affil[2]{Paul G. Allen School of Computer Science \& Engineering, University of Washington}
        \affil[3]{Computer Science Department, Ben-Gurion University of the Negev}
        \date{}
    \fi

    \ifnum\llncs=1

        \author{Boyang Chen\inst{1}\orcidlink{0009-0000-1043-0977}
            \and         Andrea Coladangelo\inst{2}\orcidlink{0000-0002-6773-2711}
            \and Or Sattath\inst{3}\orcidlink{0000-0001-7567-3822} 
        }
        \authorrunning{Chen et al.}
        %
        \institute{Department of Computer Science and Technology, Tsinghua University, Beijing, China \email{by-chen24@mails.tsinghua.edu.cn} \and
        Paul G. Allen School of Computer Science \& Engineering, University of Washington, Seattle, USA \email{coladan@cs.washington.edu} \and
        Department of Computer Science, Ben Gurion University of the Negev, Beersheba, Israel 
        \email{sattath@bgu.ac.il}
        }
    \fi
\else
    \author{}\institute{}
\fi

\maketitle
\begin{abstract} 

In this work, we focus on the following question: what are the cryptographic implications of having access to an oracle that provides a \emph{single} Haar random quantum state? We find that the study of such a model sheds light on several aspects of the notion of quantum pseudorandomness.

Pseudorandom states ($\prs$) are a family of states for which it is hard to distinguish between polynomially many copies of either a state sampled uniformly from the family or a Haar random state. A weaker notion, called single-copy pseudorandom states ($\oprs$), satisfies this property with respect to a single copy. We obtain the following results:
\begin{itemize}
\item First, we show, perhaps surprisingly, that $\oprs$ (as well as bit-commitments) exist relative to an oracle that provides a \emph{single} Haar random state. 
\item Second, we build on this result to show the existence of an isometry oracle relative to which $\oprs$ exist, but $\prs$ do not.
\end{itemize}
Taken together, our contributions yield one of the first black-box separations between central notions of quantum pseudorandomness, and introduce a new framework to study black-box separations between various inherently quantum primitives.\ifnum\shortver=0\footnote{\anote{Added this footnote.}We point out that 
an earlier version of this paper claimed an oracle separation of $\oprs$ and $\prs$ relative to a unitary oracle. However, the ``lifting'' of our isometry oracle to a unitary oracle contained a mistake, pointed out to us by Mark Zhandry. While the separation can still be lifted to be relative to a unitary oracle, this is done via different techniques in~\cite{BMM+24} (and via the more recently proposed unifying framework of~\cite{GZ25}). Our lifting techniques are only sufficient to yield a separation relative to a ``parametrized'' unitary oracle, i.e.\ one where the unitary oracle depends on the security parameter.}\fi

 \ifnum\llncs=1
 \fi
\end{abstract}
\ifnum\toc=1
    \ifnum\llncs=0
        \newpage
        \tableofcontents 
        \newpage
    \fi
\fi

\tableofcontents

\section{Introduction}
\label{sec:introduction}
It is well known that computational assumptions are necessary for almost all modern classical and quantum cryptographic tasks. The minimal assumption that is useful for classical cryptography is the existence of one-way functions ($\owf$). 
This assumption is known to be equivalent to the existence of many other cryptographic applications, such as pseudorandom number generators, pseudorandom functions, digital signatures, symmetric-key encryption, and commitments (see, e.g.,~\cite{Gol01,Gol04}).

The quantum setting presents a drastically different picture: a variety of quantum primitives are known that are sufficient to build cryptography, but are \emph{potentially weaker} than one-way functions. Recently, Tomoyuki Morimae coined the term \emph{Microcrypt}, as an addition to Impagliazzo's five worlds~\cite{Imp95}, to refer to such quantum primitives (and their cryptographic applications)\footnote{As far as we know, Morimae introduced the term in a talk \url{https://www.youtube.com/live/PKfYJlKD3z8?feature=share&t=1048}, though he did not provide a precise definition, so our definition might be slightly different than his original intention.}.
One of the tenants of Microcrypt are \emph{pseudorandom states} ($\prs$), first introduced by Ji, Liu, and Song~\cite{JLS18}. This is a family of efficiently generatable quantum states $\{\ket{\phi_k}\}_{k\in \{0,1\}^n}$ such that it is computationally hard to distinguish between polynomially many copies of (a) $\ket{\phi_k}$ sampled uniformly from the family, and (b) a uniformly (Haar) random quantum state. Ji, Liu, and Song also provided a black-box construction of $\prs$ from a $\owf$.
Subsequent to \cite{JLS18}, many other tenants of Microcrypt have been introduced, such as pseudorandom function-like states ($\prfs$)~\cite{AGQY22}, efficiently samplable statistically far-but-computationally-indistinguishable pairs of (mixed) quantum states ($\efi$ pairs)~\cite{Yan22,BCQ23}, one-way state generators~\cite{MY22b}, and pseudorandom states with proof of destruction~\cite{BBSS23}. 

Many cryptographic applications are known based on Microcrypt assumptions. By now, variants of all of the main Minicrypt\footnote{Minicrypt primitives are those that are equivalent to one-way functions. The term was introduced by Impagliazzo \cite{Imp95}.} primitives have been shown to be in Microcrypt, including symmetric-key encryption, commitments (recently, also commitments to quantum states~\cite{GJMZ23}), PRGs, PRFs, garbled circuits, message authentication codes, and digital signatures. Perhaps more surprisingly, Microcrypt also contains some tasks in Cryptomania, namely, secure multi-party computation~\cite{MY22b,BCKM21,GLSV21} and public-key encryption with quantum public keys~\cite{BGH+23}. The key factor contributing to the surprise is Impagliazzo and Rudich's separation between one-way functions (Minicrypt) and public-key encryption\footnote{Note that this classical separation does not apply for public key encryption with \emph{quantum} public keys.} and oblivious transfer (Cryptomania)~\cite{IR89}. The new constructions circumvent classical impossibilities because they involve quantum states, e.g.\ commitments and multiparty computation rely on quantum communication, and encryption schemes have quantum ciphertexts.  

The evidence that these quantum primitives are \emph{weaker} than Minicrypt comes from Kretschmer's quantum oracle separation of $\prs$ and $\owf$s \cite{Kre21}. The separating oracle consists of a family $\{\mathcal{U}_n\}_{n \in \mathbb{N}}$, where $\mathcal{U}_n$ is a list of exponentially many Haar random $n$-qubit unitaries $\{U_k\}_{k \in\{0,1\}^n}$. Relative to this oracle, there is a simple construction of a $\prs$: for $k \in \{0,1\}^n$, let $\ket{\phi_k}:=U_k\ket{0^n}$. 
Note that, if we just consider the action of the unitaries $U_k$ on the standard basis states, i.e.\ the set of states $U_k \ket{x}$ for $x \in \{0,1\}^n$, then, for each $n$, Kretschmer's oracle can be viewed as providing $2^{2n}$ ``essentially Haar random'' states\footnote{The states are Haar random subject to the constraint that they should be pairwise orthogonal (for each fixed $k$).}.
 In another work, Bouland, Fefferman and Vazirani~\cite{BFV19} show\footnote{Modulo a technical gap in their proof~\cite[p.\ 19]{BFV19}: "We expect the same result would apply \dots but we do not prove this fact." } a $\prs$ construction relative to a family $\{\mathcal{U}_n\}_{n \in \mathbb{N}}$, where $\mathcal{U}_n = (U, U^{-1})$ for a Haar random $n$-qubit $U$. By considering the action of $U$ on the standard basis states, this oracle can be viewed as providing $2^{n}$ essentially Haar random states. This raises a natural question. What can be done with much fewer Haar random states? We look at the most extreme case and ask: 
\begin{center}
\emph{What are the cryptographic implications of having oracle access to a \emph{single} Haar random state?}\footnote{Or, more precisely, one $n$-qubit Haar random state for each value of $n$ (which is accessed by providing the input $1^n$).}
\end{center}
We put forward the \chrs\ (\CHRS) model, where all parties (including the adversary) have access to an arbitrary \emph{polynomial number of copies} of a \emph{single} Haar random state. We find that this model sheds light on several aspects of quantum pseudorandomness.
First of all, is quantum pseudorandomness possible in this model? In the classical setting, having access to a fixed (random) string, which can be used both by the algorithm and the adversary, is not enough to construct pseudorandomness (e.g., pseudorandom generators). In the quantum setting, one may naturally expect that, similarly, a single Haar random state is not enough to construct \emph{quantum} pseudorandomness.

The $\prs$ variant that is most relevant for this work is \emph{single-copy} pseudorandom states ($\oprs$), introduced by Morimae and Yamakawa~\cite{MY22a}. They differ from (multi-copy) pseudorandom states ($\prs$) in two important ways (see \cref{def:oprs} for a formal definition):
\begin{enumerate}
    \item The adversary needs to distinguish between a single copy of the pseudorandom state and a single copy of a Haar random state.
    \item The construction has to be ``stretching'': the number of output qubits has to be greater than the key size (for this to be a non-trivial object).
\end{enumerate}

\subsection{Our results}
Our first result is that, perhaps surprisingly, single-copy pseudorandom states exist in this model:

\begin{theorem}[Informal]
\label{thm:main2}
$\oprs$ exist in the \CHRS\ model.
\end{theorem}

The $\oprs$ is statistically secure as long as the number of copies of the Haar random state that the adversary receives is polynomial. This result is shown in \cref{sec:oprs_in_chrs_model}. One of the main technical ingredients that we introduce to prove \cref{thm:main2} is a certain ``stretching'' result for quantum pseudorandomness in the CHRS model (\cref{thm:ammplify-random-informal} in the technical overview, and \cref{thm:ammplify-random} in the main text), which may find application elsewhere.

As a result, we show that the statistical $\oprs$ above can be used to achieve a surprisingly strong form of bit-commitment:
\begin{theorem}[Informal]
    In the \CHRS\ model, a non-interactive quantum bit-commitment exists that is statistically hiding and binding.
\end{theorem}

The hiding property holds against a computationally unbounded adversary that receives any polynomial number of copies of the Haar random state. In contrast, the binding property holds against a computationally unbounded adversary with an unbounded number of copies. Such a statistically binding and hiding commitment cannot exist in the standard model~\cite{LC97,May97}. The proof of the theorem follows the approach of Morimae and Yamakawa~\cite{MY22a} to construct commitments from a $\oprs$. The subtlety is that the construction of ~\cite{MY22a} utilizes the \textit{inverse} of the generator of the $\oprs$, something that is in general infeasible in the CHRS model. We settle the issue by showing a weak equivalence between the CHRS oracle and a corresponding unitary oracle, which is self-inverse (see \cref{sec:tech-state-to-unitary} for a technical overview). Thanks to Theorem 14 in \cite{Qia23}, the commitment scheme that we obtain in the CHRS model can be compiled into an $\epsilon$-simulation secure one, using an adaption of the compiler from \cite{BCKM21}. This version of commitment is sufficient to build secure multiparty computation via the construction in \cite{BCKM21}.

Even though plenty of relations involving Microcrypt primitives are known, the only \emph{black-box separations} involving Microcrypt are the following: Kretschmer~\cite{Kre21} separated post-quantum $\owf$ from $\prs$, via a quantum oracle. Ananth, Qian and Yuen~\cite{AQY22} observed that this separation also separates $\owf$ from $\prfs$. Kretschmer et al.~\cite{KQST23} separated $\owf$ from $\oprs$ via a \emph{classical} oracle. However, when we zoom in on Microcrypt, almost nothing is known about whether different Microcrypt primitives are equivalent to each other, or whether there is a hierarchy.
The only known non-trivial\footnote{\cite{BS20b} (see also \cite[p.3]{ALY23}) show that PRS with very short output ($c \cdot \log (\mylambda)$ for $c \ll 1$, where $\mylambda$ is the length of the key) exist \emph{unconditionally}. Hence, they are trivially black-box separated from all of the other Microcrypt primitives which require computational assumptions.} separation is between short output and long output $\prs$ (with the former being potentially stronger). This separation is an immediate consequence of the works of Barhoush et al. \cite{barhoush2023pseudo} (which gives a construction of quantum digital signatures from PRS with short output) and Coladangelo and Mutreja \cite{coladangelo2024blackbox} (which shows an oracle separation between quantum digital signatures and PRS with long output), and was also shown in a concurrent work of Bouaziz--Ermann and Muguruza \cite{bouaziz2024quantum}.

In this work, building on our \cref{thm:main2}, we show a second black-box separation \emph{within} Microcrypt:
\begin{theorem}[Informal]
\label{thm:main}
There is an isometry oracle relative to which $\oprs$ exist, but $\prs$ with output length at least $\log n + 10$ (where $n$ is the seed length) do not. Additionally, there exists a ``parametrized'' \footnote{A ``parametrized'' oracle is a family of oracles $\{O_\lambda\}$. Existence relative to $\{O_\lambda\}$ means that, for a security parameter $\lambda$, both the construction and the adversary are only allowed to query $O_\lambda$. An oracle of this kind does not rule out the most general kind of black-box construction (which can make use of an arbitrary unitary implementation of primitive $A$, and its inverse, in order to build primitive $B$), but only rules out black-box constructions of primitive $B$ that, for a fixed security parameter $n$, only make use of a unitary implementation of $A$ for the same fixed security parameter $n$. We clarify that, while our unitary oracle separation is ``parametrized'', our isometry oracle separation is not.} unitary oracle relative to which $\oprs$ exist, but $\prs$ with output length at least $\omega(\log n)$ do not. 
\end{theorem}
This yields one of the first black-box separations between central notions of quantum pseudorandomness. The separation relative to the isometry oracle is essentially tight in terms of output length, since PRS with very short output ($c \cdot \log (\mylambda)$ for $c \ll 1$) exist \emph{unconditionally} \cite{BS20b}. We show this result in \ifproceedingelse{the full version~\cite{arxivfullversion}}{\cref{sec:oralce-sep}}. The upgrade to a ``parametrized'' unitary oracle is inspired by techniques by Ji, Liu, and Song~\cite{JLS18} and Zhandry~\cite{zhandry2024space}, with some differences.\footnote{As mentioned earlier, a previous version of this paper claimed to lift the isometry oracle to a standard unitary oracle (rather than a ``parametrized'' one). However, the proof of this contained a mistake, pointed out to us by Mark Zhandry. A separation relative to a standard unitary oracle (in the full parameter regime) can still be obtained via different techniques, as in ~\cite{BMM+24} or \cite{GZ25}.} 

    Taken together, our contributions introduce a new framework that seems very well-suited to study black-box separations between various inherently quantum primitives, particularly between ``single-copy'' and ``multi-copy'' primitives. Our framework has already been fruitful, and has been utilized in the works of Bostanci, Chen, and Nehoran~\cite{BCN24}, and Behera et al.~\cite{BMM+24}\footnote{We clarify that, while \cite{BCN24} and \cite{BMM+24} are subsequent to the original version of our paper (which introduces the CHRS model, and proves the first black-box separation of $\oprs$ and $\prs$), our isometry-to-unitary oracle upgrade appears in a later version of our paper, which is concurrent to \cite{BCN24} and \cite{BMM+24}.}.

Finally, for the reader's benefit, we include in \ifproceedingelse{the full version~\cite{arxivfullversion}}{Section \ref{sec:black-box-reductions}} a formal discussion of various notions of black-box oracle separations and their implications in terms of the impossibility of black-box constructions. 

\paragraph{Related work.}
In this work, we introduce the \chrs\ (\CHRS) model, in which both the generation algorithm and the adversary have access to polynomially many copies of a Haar random state over $n$ qubits. 
There are two related models. The first, which our work is a particular case of, was called the \emph{quantum auxiliary input} model (where the quantum state is sometimes referred to as the quantum advice) by~\cite{morimae2023unconditionally}, in which the parties are provided with polynomially many copies of a quantum state, which need not be efficiently generatable\footnote{We prefer not to use the term ``quantum auxiliary input'' since in most other works we are aware of (see~\cite{DGKPV10} and references therein), a quantum auxiliary input typically represents a setting in which the adversary may have information that may depend on the honest parties' inputs, and in particular, the secret key. In contrast, in our setting and that of \cite{morimae2023unconditionally}, the ``auxiliary'' state is fixed, independently of any honest parties' input.}. Chailloux, Kerenidis, and Rosgen~\cite{CKR16} showed that quantum commitments with quantum auxiliary input exist under a \emph{computational assumption}. They provide two schemes, where either the hiding or binding properties are computational. 
Morimae, Nehoran, and Yamakawa~\cite{morimae2023unconditionally} and Qian~\cite{Qia23} recently proved, \emph{unconditionally}, the existence of a computationally hiding and statistically binding commitment in the quantum auxiliary input model. This improves on the result of~\cite{CKR16}, in the sense that the computational assumption is removed.

The second related model is the \emph{common reference quantum state} (CRQS) model, in which the quantum state needs to be efficiently generatable. Note that, in the classical setting, the common reference string represents a model with a trusted setup. In this model, \cite{morimae2023unconditionally} show a statistically hiding and binding commitment with similar properties to ours. The difference is in the order of quantifiers of the hiding property: in our work, the scheme is hiding against an adversary that is allowed to have any polynomial number of copies of the quantum (Haar-random) state; in their construction (see~\cite[Theorem 1.4]{morimae2023unconditionally}), they first pick a polynomial $t(n)$ and show a construction which is hiding against adversaries which receive $t(n)$ copies of the CRQS\footnote{Even though this was not formally claimed in~\cite{morimae2023unconditionally}, we believe that the construction mentioned in the previous paragraph, with (inefficiently generatable) auxiliary quantum inputs, satisfies the same statistical security guarantees as ours.}. Of course, the main disadvantage of our work is that a Haar random state cannot be efficiently generated, whereas the state they use is efficiently generatable. However, note that if one is satisfied with security against some fixed polynomial $t(n)$ of copies, the Haar random state can be replaced efficiently by a quantum $t(n)$-design. 

We emphasize the features that differentiate our work:
\begin{itemize}
\item[(i)] Our common random state is \emph{structure-less}: it is a Haar random state.
\item[(ii)] We show how to achieve quantum pseudorandomness in this model. The related works construct commitments directly, but their constructions do not have any implications with regard to quantum pseudorandomness. We find it quite surprising that a Haar random state alone can yield quantum pseudorandomness. It is also thanks to this connection that we are able to separate different flavors of quantum pseudorandomness, namely $\oprs$ and $\prs$.
\end{itemize}

Finally, in the past few years, many results regarding Microcrypt have been discovered---at this point, too many to cover in detail. 
\ifnum \shortver=0
    A diagram showing the different Microcrypt primitives, their relations, applications, and separations are depicted in \cref{fig:MircoCrypt_diagram} on \cpageref{fig:MircoCrypt_diagram}. 
\begin{figure} 
    \centering
    
    \centerline{
    \includegraphics[width=1.2\textwidth
    ]{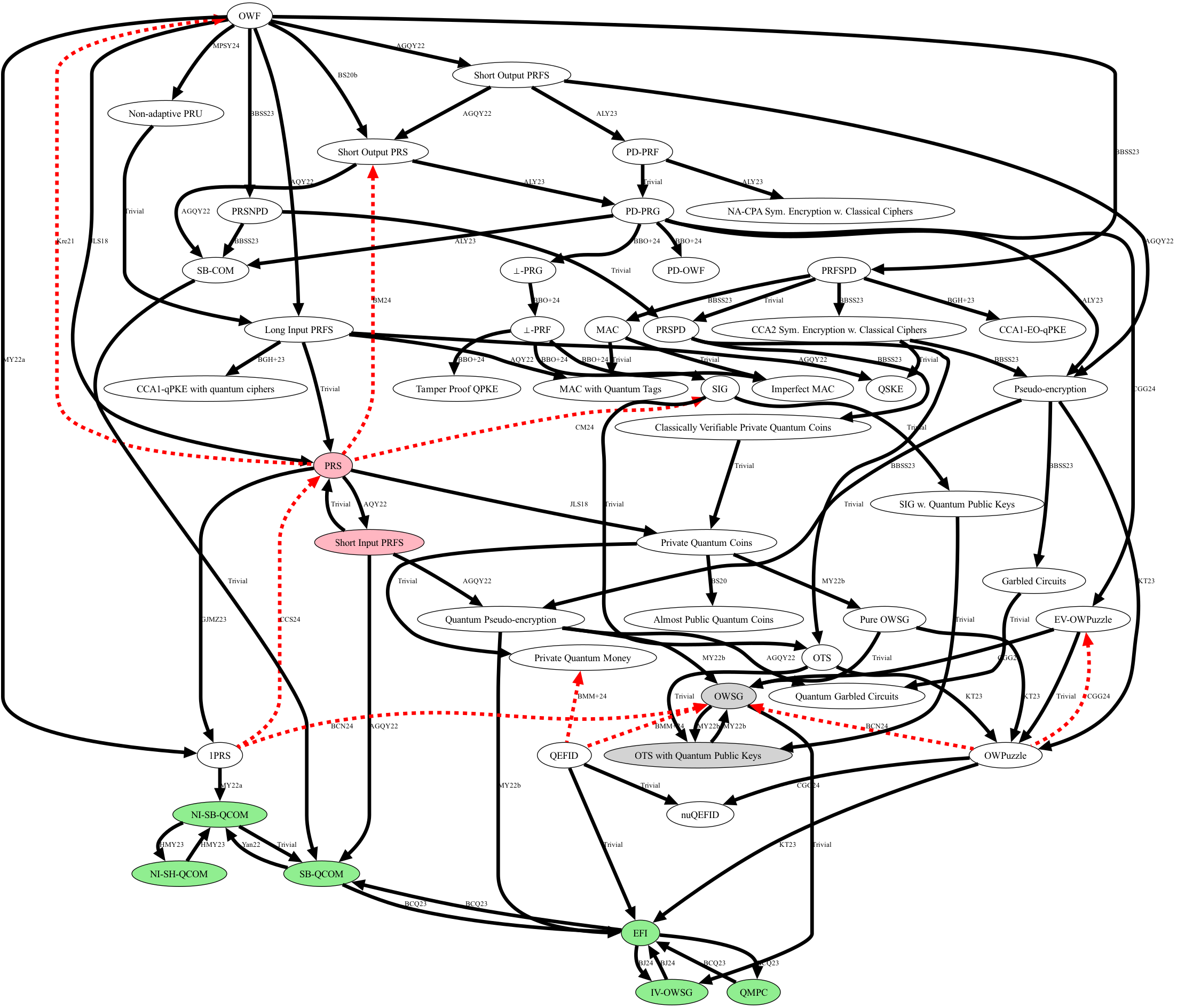}}
    \vspace{-1pt}
    \caption{Diagram of the known relations and applications in Microcrypt, as of September 2024. Regular arrows indicate implications, and dotted arrows indicate black-box separations. Nodes that share a color are equivalent. An interactive version of this diagram is available at \url{https://sattath.github.io/microcrypt-zoo/}, with additional features, such as ``mouseover a node'' reveals additional details, and ``mouseover an edge'' shows a clickable source for that relation. \ifnum\anonymous=0
        The website is updated periodically, therefore, the online version may differ from the one above as new results are published.
    \fi}
    \label{fig:MircoCrypt_diagram}
\end{figure}
\else
A diagram showing the different Microcrypt primitives, their relations, applications, and separations is available at \url{https://sattath.github.io/qcrypto-graph/}.
\fi
\paragraph{Concurrent work.} 
We point out the independent and concurrent work of Ananth, Gulati, and Lin~\cite{ananth2024note}, which appeared shortly after the first version of our paper, and was subsequently expanded in \cite{ananth2024cryptography}. We refer to the two works collectively as AGL. We briefly discuss how our work and AGL relate to each other. In short, AGL has stronger feasibility results, while our work has arguably stronger negative results. 

AGL improves upon our 1PRS construction, by presenting a strictly simpler 1PRS construction that achieves arbitrary stretch, with a simpler elementary analysis. AGL also provides a construction of PRS that are secure against adversaries that receive a \emph{fixed} (slightly less than linear) number of copies of the PRS state.

Our work gives an oracle separation between $\oprs$ and $\prs$ in the CHRS model, whereas AGL only separates $\oprs$ from $\prs$ that are limited to using one copy of the common Haar state (and thus it is a bit unclear what the implication of the latter is in terms of impossibility of black-box constructions). The more recent version of AGL includes a construction of $O({n^{0.99}})$-copy secure pseudorandom function-like states ($\textsf{PRFS}$) and an impossibility result for certain primitives beyond Microcrypt (like interactive key-agreement and commitments) in the CHRS model. Before our present work, all of the mentioned separations treated the CHRS oracle as an \emph{isometry}\footnote{One can view an input-less oracle that provides a state as an isometry.}. This was slightly unsatisfactory for the following reason: a separation of primitive A from primitive B relative to an isometry oracle only rules out black-box constructions of B from A that use ``isometry'' implementations of the procedures from A (i.e.\ when running an implementation of a procedure from A, the construction of B is not allowed to set the auxiliary qubits to anything but all zeros, and it is not allowed to use the inverse of the algorithms of A -- we refer the reader to \ifproceedingelse{the full version~\cite{arxivfullversion}}{Section \ref{sec:black-box-reductions}} for a formal discussion of this point). 

We also point out the work by Bostanci, Chen, and Nehoran~\cite{BCN24}, and by Behera et al.~\cite{BMM+24} (subsequent to the first version of our paper, but concurrent to the second), who also contain techniques to lift separations in this framework from isometry to unitary oracles. In particular, the lifting result from \cite{BMM+24} is stronger than ours, as it applies to the full parameter regime of $\prs$ output length, and, more importantly, lifts an isometry oracle to a standard unitary oracle (rather than a ``parametrized'' one). The results of \cite{BCN24} and \cite{BMM+24} also leverage our framework and extend our results to separate $\oprs$ and one-way state generators (a ``multi-copy'' notion of quantum ``one-wayness'' introduced in \cite{MY22a,MY22b}). Additionally, both of these works also study the recently introduced notion of ``one-way puzzles''~\cite{khurana2024commitments}, and separate its efficient and inefficient verifier variants. All of our other contributions (introducing the CHRS model itself, and showing that it is useful for separating notions of quantum pseudorandomness) are unique to our paper and \cite{ananth2024cryptography}. 

\ifnum \shortver=0
\paragraph{Open problems.}
This work opens up several directions for further research.
\begin{itemize}
\item Our separation result (Theorem~\ref{thm:main}) holds relative to a quantum oracle. Can it be shown relative to a \emph{classical} oracle? We note that Krethschmer et al.~\cite{KQST23} show a classical oracle relative to which $\oprs$ and commitments exist, but one-way functions do not. 


\item There are examples of primitives that we know can be constructed from $\prs$, but are \emph{not} known to be implied by $\oprs$. The main examples are one-time digital signatures with quantum public keys~\cite{MY22a}, private quantum coins~\cite{JLS18}, and quantum pseudo-encryption~\cite{AQY22}. Currently, we do not have a separation between those applications%
\footnote{or even ones which are based on stronger Microcrypt assumptions, such as the existence of long input $\prfs$, which can be used to construct message authentication codes with quantum tags~\cite{AQY22}, quantum symmetric key encryption~\cite{AQY22}, and public key encryption with quantum ciphers and quantum public keys~\cite{BGH+23}.}
and $\oprs$.  Understanding whether any of these applications are separated from $\oprs$ would be interesting. 

\end{itemize}
\fi

\ifnum\anonymous=0
\subsubsection*{Acknowledgments}
AC and OS thank NTT research and Mark Zhandry for organizing the quantum money workshop, as well as the participants of the workshop, where this research was initiated. BC thanks Xingjian Li for helpful discussions. The authors also thank Mark Zhandry for pointing out an error in our first proof upgrading our oracle separation from isometry to unitary.

BC acknowledges supported by National Key Research and Development Program of China (Grant No.\ 2023YFA1009403) and National Natural Science Foundation of China (Grant
No.\ 12347104).

This research was supported by the Israel Science Foundation (grant No. 2527/24).

            \BeforeBeginEnvironment{wrapfigure}{\setlength{\intextsep}{0pt}}
            \begin{wrapfigure}{r}{100px}
                \includegraphics[width=100px]{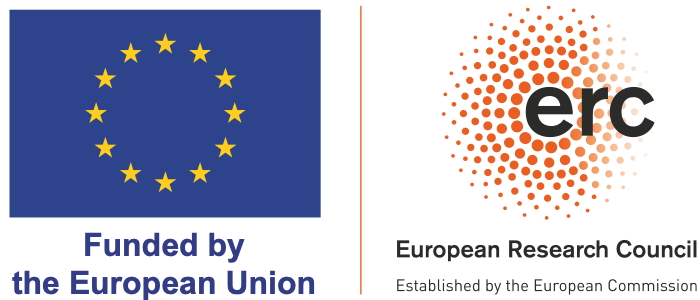}
            \end{wrapfigure}
            OS was funded by the European Union (ERC-2022-COG, ACQUA, 101087742). Views and opinions expressed are however those of the author(s) only and do not necessarily reflect those of the European Union or the European Research Council Executive Agency. Neither the European Union nor the granting authority can be held responsible for them.
\fi

\section{Technical Overview}
This section is organized as follows. In \cref{sec:1prs-techoverview}, we describe the construction of a $\oprs$ in the \CHRS\ model, and we give a high-level overview of the proof of security. We view this as the main technical contribution of our work. We also describe how to construct in the CHRS model following the approach in~\cite{MY22a}, with slight modification to deal with the inverse issue. Finally, in \cref{sec:separation-techoverview}, we describe an oracle separation between $\oprs$ and $\prs$. 
We consider the CHRS model augmented with quantum oracle access to a QPSPACE machine, and we describe a generic attack on any $\prs$ construction in this model. Since $\oprs$ still exist in this model, this yields an oracle separation between the two. 


\subsection{Construction of \texorpdfstring{$\oprs$}{1PRS} in the \CHRS\ model}

\label{sec:1prs-techoverview}
\paragraph{$\oprs$ definition.} Recall that, informally, a $\oprs$ is a QPT algorithm that takes as input a seed $k \in \{0,1\}^n$ (where $n$ is a security parameter) and outputs a state of some length $m>n$. We denote by $\ket{\phi_k}$ the output state on seed $k$. Then, security requires that a \emph{single} copy of the $\oprs$ state be computationally indistinguishable from a \emph{single} maximally mixed state of the same dimension, i.e.\ 
$$
\mathbb{E}_k \, \ket{\phi_k} \bra{\phi_k} \approx_c \frac{\mathds{1}}{2^m}
$$
(where $\approx_c$ denotes computational indistinguishability).

Note that this requirement is only non-trivial when $m>n$ (otherwise, one can simply output the seed itself). 
Equivalently, one can think of the problem of constructing a $\oprs$ as the problem of finding a family $\{U_k\}_{k \in \{0,1\}^n}$ of efficiently computable unitaries such that 
$$
\mathbb{E}_k \, U_k \ket{0}\bra{0} U_k^{\dagger} \approx_c \frac{\mathds{1}}{2^m} \,.
$$
This problem becomes trivial if the family  $\{U_k\}$ is large enough. In particular, if $m = n$, a classical one-time pad, i.e.\ taking $U_k = X^k$ already suffices. One way to achieve the above with $m>n$ is, of course, to use a classical PRG, but this is of course already equivalent to assuming OWFs.

\paragraph{Working in the CHRS model.} We will instead describe how to construct a $\oprs$ in the CHRS model, i.e.\ when polynomially many copies of a single Haar random state are available to the construction and to the adversary. Our construction uses a single copy of the state $\ket{\psi}$, but security holds even when $r = \poly(n)$ copies of $\ket{\psi}$ are available to the adversary. 

We restrict ourselves to considering constructions of the following form: the $\oprs$ family $\{\ket{\phi_k}\}$ is such that $\ket{\phi_k} = U_k \ket{\psi}$. Let $m$ be the number of qubits of $\ket{\psi}$. Thus, the problem reduces to finding a family $\{U_k\}_{k \in \{0,1\}^n}$, for $m>n$, such that\footnote{Technically, as pointed out in an earlier footnote, parties in the CHRS model (including the adversary) have access to copies of one $m$-qubit Haar random state \emph{for each $m$}. However, it is clear that this is immaterial to the proof, since, for a given output length $m$, we are restricting our attention to constructions (i.e.\ choices of $U_k$) that only act on the $m$-qubit Haar state, and ignore the others.}
\begin{equation}
\label{eq:1}
\mathbb{E}_{\ket{\psi} \leftarrow \mu_
{2^m}} \,\mathbb{E}_{k \in \{0,1\}^n} (U_k \ket{\psi}\bra{\psi} U_k^{\dagger}) \otimes (\ket{\psi}\bra{\psi})^{\otimes r} \approx_c \mathbb{E}_{\ket{\psi} \leftarrow \mu_{2^m}} \frac{\mathds{1}}{2^m} \otimes (\ket{\psi}\bra{\psi})^{\otimes r} \,.
\end{equation}
In fact, we will describe a construction that achieves \emph{statistical} (rather than just computational) indistinguishability, assuming $r$ is polynomial in $n$. As anticipated, the crux of the problem is to achieve the above with $m > n$. 

\paragraph{Construction of $\oprs$ in the CHRS model.}
For the reader's convenience (to help remember what the parameters refer to), going forward we have
\begin{itemize}
\item $k$: $\oprs$ seed.
\item $n = |k|$.
\item $m$: number of qubits of the output $\oprs$ state (this is also the number of qubits of the Haar random state $\ket{\psi}$).
\end{itemize}
Our construction of a $\oprs$ in the CHRS model is simple (although it is unclear a priori why it would work). We take the family of $m$-qubit unitaries $\{U_k\}$ to be a Quantum One-Time Pad (QOTP) on slightly less than half of the qubits, say $0.45m$. A bit more precisely, $k$ is a string of length $n \in [0.9m,m)$, which we can parse as $k = (a, b)$, where $a,b \in \{0,1\}^{n/2}$. Then, $U_k = (X^a Z^b) \otimes I$, i.e.\ $U_k$ applies $X^a Z^b$ to the first $n/2$ qubits of the $m$-qubit state it acts on. We now explain the intuition behind the construction.

\paragraph{First key idea: a quantum one-time pad on \emph{exactly} half of the qubits.} 

Notice, just for the sake of argument, that if we allowed ourselves to have $n = 2m$ (even though this violates the ``length extending'' requirement of $m >n$ by a large margin), then there would be a trivial choice of $U_k$ that works: simply pick $\{U_k\}$ to be a QOTP on \emph{all} of the qubits. Then, the $\oprs$ security property of Equation \eqref{eq:1} would be satisfied. Unfortunately, the full QOTP is very far from our goal: to comply with the length-extending requirement, a QOTP must be applied to \emph{strictly less than half} of the qubits. 

Let us simplify our life slightly for the moment: if we allow a QOTP on \emph{exactly half} of the qubits, i.e.\ $n = m$ (which still does not satisfy the requirement of $m>n$), is Equation \eqref{eq:1} satisfied? It turns out that the answer is yes (although the reason may be unclear at first). We provide an informal explanation.

The starting point is a recent result by Harrow \cite{harrow2023approximate}. This says that the state obtained by applying a Haar random unitary to one-half of a maximally entangled state is statistically indistinguishable from Haar random. Crucially, this guarantee also holds for multiple copies (in the appropriate parameter regime). A bit more precisely, Harrow proves the following. For $d \in \mathbb{N}$, let $\ket{\Phi_d} = \frac{1}{\sqrt{d}} \sum_{i=0}^{d-1} \ket{ii}$, and for a unitary $U$ acting on the left register, let $\ket{\phi_U} = (U \otimes I) \ket{\Phi_d}$. 
For a pure state $\ket{\psi}$, we denote by $\psi$ its density matrix.
\begin{lemma}[Harrow \cite{harrow2023approximate}, informal]\label{lem:harrow-tech-overview}
Let $r, d \in \mathbb{N}$. Then,
    \begin{equation*}
    \norm{\E_{\ket{\psi} \gets \mu_{d^2}} [\psi^{\otimes r}] - \E_{U \gets SU(d)} [\phi_U^{\otimes r}]} \leq \frac{r^2}{d}
    \end{equation*}
\end{lemma}
In the case of a single copy ($r=1$), the following is some intuition as to why the result holds. Consider a Haar random state and any partition of its qubits into two registers $\mathsf{A}$ and $\mathsf{B}$. Then, with very high probability, a Haar random state has Schmidt coefficients close to uniform. This is somewhat intuitive (although it requires some work to prove). This implies that the following mixed state is close to a Haar random state:
$$ \mathbb{E}_{U, U' \leftarrow SU(d)} (U \otimes U') \Phi_d  (U \otimes U')^{\dagger}\,,$$
(the latter is a maximally entangled state to which independent Haar random unitary changes of basis are applied to each side).
However, notice that
\ifnum\widemargin=0
\begin{align*}
    \mathbb{E}_{U, U' \leftarrow SU(d)} (U \otimes U') \Phi_d (U \otimes U')^{\dagger} &=  \mathbb{E}_{U, U' \leftarrow SU(d)} (U \cdot U'^T  \otimes I ) \Phi_d (U \cdot U'^T \otimes I )^{\dagger}\\
    &=  \mathbb{E}_{U \leftarrow SU(d)} (U \otimes I) \Phi_d (U \otimes I)^{\dagger} = \mathbb{E}_{U \leftarrow SU(d)} \phi_U \, , 
\end{align*}
\else
\begin{align*}
    &\mathbb{E}_{U, U' \leftarrow SU(d)} (U \otimes U') \Phi_d (U \otimes U')^{\dagger} \\&=  \mathbb{E}_{U, U' \leftarrow SU(d)} (U \cdot U'^T  \otimes I ) \Phi_d (U \cdot U'^T \otimes I )^{\dagger}\\
    &=  \mathbb{E}_{U \leftarrow SU(d)} (U \otimes I) \Phi_d (U \otimes I)^{\dagger} = \mathbb{E}_{U \leftarrow SU(d)} \phi_U \, , 
\end{align*}
\fi
where the first equality follows from the ``Ricochet'' property of the maximally entangled state, and the second by the unitary invariance of the Haar measure. Thus,
$$ \E_{\ket{\psi} \gets \mu_{d^2}} [\psi] \approx \E_{U \gets SU(d)} [\phi_U] \,. $$
The general result for $r>1$ copies is much more involved, and we refer the reader to \cite{harrow2023approximate}.

So, how does Harrow's result help the analysis? The $r$-copy result says that
$$ \E_{\ket{\psi} \gets \mu_{d^2}} [\psi^{\otimes r}] \approx \E_{U \gets SU(d)} [\phi_U^{\otimes r}] \,.$$
Let $m = n$ be even, and take $d =2^{m/2}$, so that $\ket{\psi}$ is an $m$-qubit state, and $\ket{\Phi_d} = \frac{1}{\sqrt{2^{m/2}}} \sum_{i=0}^{2^{m/2}-1} \ket{ii}$, i.e.\ a maximally entangled state on $m$ qubits. Let $\mathcal{P}_{m/2}$ denote the Pauli group on $m/2$ qubits. Applying a QOTP to the first $m/2$ qubits (i.e.\ \emph{exactly half}) of the \emph{first} out of the $r$ copies, we get:
\if\widemargin=0
\begin{align}
\E_{P \gets \mathcal{P}_{m/2}}\E_{\ket{\psi} \gets \mu_{d^2}} \left[(P\otimes I) \psi (P^{\dagger}\otimes I) \otimes \psi^{\otimes (r-1)}\right] 
&\approx \E_{P \gets \mathcal{P}_{m/2}} \E_{U \gets SU(d)} \left[(P U \otimes I) \Phi_d (U^{\dagger}P \otimes I)^{\dagger}\otimes  \phi_U^{\otimes (r-1)}\right] \label{eq:3}\\
&= \E_{P \gets \mathcal{P}_{m/2}} \E_{U \gets SU(d)} \frac{1}{2^{m/2}} \sum_{i,j} PU\ket{i}\bra{j}U^\dagger P^\dagger \ot \ket{i} \bra{j} \ot \phi_U^{\ot r-1} \nonumber \\
&= \E_{P \gets \mathcal{P}_{m/2}} \E_{U \gets SU(d)} \frac{1}{2^{m/2}} \sum_{i,j} PU\ket{i}\bra{j}U^\dagger P^\dagger \ot \ket{i} \bra{j} \ot \phi_U^{\ot r-1} \nonumber \\
&=  \frac{\mathds{1}}{2^{m}} \ot \E_{U \gets SU(d)} \phi_U^{\ot r-1} \,, \label{eq:4}
\end{align}
\else
\begin{align}
&\E_{P \gets \mathcal{P}_{m/2}}\E_{\ket{\psi} \gets \mu_{d^2}} \left[(P\otimes I) \psi (P^{\dagger}\otimes I) \otimes \psi^{\otimes (r-1)}\right] 
\\&\approx \E_{P \gets \mathcal{P}_{m/2}} \E_{U \gets SU(d)} \left[(P U \otimes I) \Phi_d (U^{\dagger}P \otimes I)^{\dagger}\otimes  \phi_U^{\otimes (r-1)}\right] \label{eq:3}\\
&= \E_{P \gets \mathcal{P}_{m/2}} \E_{U \gets SU(d)} \frac{1}{2^{m/2}} \sum_{i,j} PU\ket{i}\bra{j}U^\dagger P^\dagger \ot \ket{i} \bra{j} \ot \phi_U^{\ot r-1} \nonumber \\
&=  \frac{\mathds{1}}{2^{m}} \ot \E_{U \gets SU(d)} \phi_U^{\ot r-1} \,, \label{eq:4}
\end{align}
\fi
where the last line follows by the Pauli Twirl (Lemma \ref{lem:quat-otp}). Recall that the ``closeness'' in the approximation of Equation \eqref{eq:3} is $\frac{r^2}{2^{m/2}}$ (from Lemma \ref{lem:harrow-tech-overview}). We emphasize the crucial step in the last equality: thanks to the maximal entanglement between the two halves of the first register, the QOTP on the first half actually causes \emph{both} halves to become maximally mixed.

It follows that, given $r=\poly(m)$ copies of an $m$-qubit Haar random state, applying a QOTP on the first $m/2$ qubits of the first copy is enough to make the first copy maximally mixed, even given the other $r-1$ copies. This gets us closer to our goal, but we are not there yet: we are still using an $m$-bit seed to obtain an $m$-qubit state.

\paragraph{Second key idea: quantum one-time pad on \emph{slightly less} than half of the qubits.}

If a QOTP on slightly less than half of the qubits were sufficient, this would solve our problem. We show that this is indeed the case!

The key technical ingredient in our proof can be viewed as a sort of ``stretching'' result, which may be useful elsewhere. Consider an $m$-qubit common Haar random state. Very informally, the ``stretching'' result says the following: if there is a way to obtain ``$m-1$ qubits of single-copy pseudorandomness'' from $n$ bits of classical randomness (where $n$ should be thought of as being linear in $m$), then one can also obtain ``$m$ qubits of single-copy pseudorandomness'' from $n$ bits of classical randomness, with a slight loss in statistical distance (i.e.\ it is possible to get one extra qubit of pseudorandomness!).
The loss is small enough that the stretching can be applied repeatedly to get up to $m$ qubits of pseudorandomness from $c \cdot n$ bits of classical randomness, for some $0.9<c<1$, while keeping the statistical loss exponentially small in $m$.

Crucially, this stretching result also applies to our base result of Equation \eqref{eq:4} (where $n = m$). More precisely, we have the following.
\begin{theorem}[Informal]\label{thm:ammplify-random-informal}
    Let $r, n,m \in \mathbb{N}$. Let $\{U_k\}_{k \in \{0,1\}^n}$ be a set of $(m-1)$-qubit unitaries. Then,
    \begin{align*}
        &\norm{\E_{k} \E_{\ket{\psi}} (\I\ot U_k) \psi (\I\ot U_k^\dagger) \ot \psi^{\ot r-1} -  \frac{\I}{2^{m}} \ot  \E_{\ket{\psi}}\psi^{\ot r-1} } \\&\leq 5\norm{\E_k \E_{\ket{\psi'}} U_k \psi' U_k^\dagger \ot \psi'^{\ot r-1 } - \frac{\I}{2^{m-1}} \ot  \E_{\ket{\psi'}}  \psi'^{\ot r-1}} + O\left(\frac{r\sqrt{m}}{2^{m/2}}\right)\,,
    \end{align*}
    where $\ket{\psi}$ and  $\ket{\psi'}$ are Haar random states on $m$ and $(m-1)$ qubits, respectively.
\end{theorem}

In words, this says that if $\{U_k\}_{k \in \{0,1\}^n}$ generates a (single-copy) $(m-1)$-qubit pseudorandom state when applied to an $(m-1)$-qubit Haar random state, then applying $U_k$ to the last $m-1$ qubits of an $m$-qubit Haar random state (and ignoring the first qubit) also suffices to achieve the same, up to a small statistical loss. 

Applying Theorem \ref{thm:ammplify-random-informal} $l$ times, gives:
\begin{corollary}[Informal]
\label{cor:ammplify-random-informal}
    Let $\ell <m$. Let $\{U_k\}_{k \in \{0,1\}^n}$ be a set of $(m-\ell)$-qubit unitaries. Then,
    \begin{align*}
        &\norm{\E_{k} \E_{\ket{\psi}} (\I\ot U_k) \psi (\I\ot U_k^\dagger) \ot \psi^{\ot r-1} -  \frac{\I}{2^{m}} \ot  \E_{\ket{\psi}}\psi^{\ot r-1} } \\&\leq 5^\ell\norm{\E_k \E_{\ket{\psi'}} U_k \psi' U_k^\dagger \ot \psi'^{\ot r-1 } - \frac{\I}{2^{m-\ell}} \ot  \E_{\ket{\psi'}}  \psi'^{\ot r-1}} + O\left(\frac{r\sqrt{m}\,5^\ell}{2^{(m-\ell)/2}}\right)\,,
    \end{align*}
    where $\ket{\psi}$ and  $\ket{\psi'}$ are Haar random states on $m$ and $(m-\ell)$ qubits, respectively.
\end{corollary}

At first, the reader might be slightly worried about the exponential blow-up of the RHS in terms of $\ell$. However, this is counteracted by the trace distance term, which, for the base case, is exponentially small in the number of qubits. Thus, there is actually a regime of $\ell$ linear in $m$ for which the upper bound is exponentially small in $m$. In more detail, we apply Corollary \ref{cor:ammplify-random-informal} to our base result of Equation \eqref{eq:4} (replacing $m$ with $m-\ell$ there). Let $L_{m-\ell}$ be the statistical closeness (in trace distance) between the two sides of Equation \eqref{eq:4}. Then we have the following: applying a QOTP to $\frac{m-\ell}{2}$ qubits of an $m$-qubit Haar random state suffices to yield a (single-copy) pseudorandom state, with a statistical loss of $L_{m-\ell} \cdot 5^\ell  + O\left(\frac{r\sqrt{m}5^\ell}{2^{(m-\ell)/2}}\right)$. Recall from earlier that $L_{m-\ell} = O\left(\frac{r^2}{2^{(m-\ell)/2}}\right)$, and so the total statistical loss is $O\left(\frac{r^2}{2^{(m-\ell)/2}} \cdot 5^\ell\right)  + O\left(\frac{r\sqrt{m}5^\ell}{2^{(m-\ell)/2}}\right)$.

Notice crucially that, when $\ell$ is too large, the factor of $5^\ell$ dominates $L_{m-\ell}$! However, when $\ell = 0.1 m$, the loss is $
O\left(\frac{(r^2+r\sqrt{m})5^{0.1m}}{2^{0.45m}}\right)$, which is still exponentially small in $m$. Thus, interestingly, our construction works as long as the QOTP is applied on $0.45m$ qubits (a constant fraction less than half), but it does not seem to work for much smaller constant fractions\footnote{We are unsure whether this regime is tight or not. Settling this is an interesting open question.}.

The high-level intuition for the result is that a typical Haar random state on $m$ qubits is ``close'' to being maximally entangled across the $(1, m-1)$ bipartition (i.e.\ the bipartition that considers the first qubit as the ``left'' register, and the remaining $m-1$ qubits as the ``right'' register). More concretely, the mixed state obtained by sampling a Haar random $m$-qubit state is close (in trace distance) to the state obtained by sampling two Haar random $(m-1)$-qubit states $\ket{\psi_1}$ and $\ket{\psi_2}$, and outputting $ \ket{\psi'}= \frac{1}{\sqrt{2}} \ket{0}\ket{\psi_1} + \frac{1}{\sqrt{2}} \ket{1}\ket{\psi_2}$, i.e.
$$
\E_{\psi} [\psi] \approx \E_{\psi_0, \psi_1} [\psi'] \,.
$$
Note that in the state $\ket{\psi'}$ the two coefficients are exactly $\frac{1}{\sqrt{2}}$ (while, for a Haar random $m$-qubit state, each coefficient would instead come from a \emph{distribution} which concentrates at $\frac{1}{\sqrt{2}}$). This observation also holds for $r>1$ copies of $\psi$ and $\psi'$, respectively, at the cost of a factor of $r$ loss in trace distance. 

How does this help? The crucial point is that if $\{U_k\}$ is a family of ``twirling'' unitaries, i.e.\ a family of unitaries such that the channel $\mathbb{E}_k \, U_k (\cdot ) U_k^{\dagger}$ maps the ``right'' register to the maximally mixed state (when also taking into account the averaging over $\psi'$), then, similarly as in the calculation of \cref{eq:4}, the ``left'' register also becomes maximally mixed (due to the fact that the two registers were originally maximally entangled). We refer the reader to \cref{ssec:ampli-random} for more details.

\begin{remark}
    The reader may wonder whether constructing a $\oprs$ can be achieved more easily or with better parameters by leveraging, for example, the following result from Dickinson and Nayak \cite{dickinson2006approximate}. This says that $n + 2\log{\frac{1}{\epsilon}} + 4$ bits of key length are sufficient to encrypt an $n$-qubit state so that it is $\epsilon$-close (in trace distance) to the maximally mixed state (rather than $2n$ bits for $n$ qubits using the standard QOTP). While the result seems potentially very useful, it does not seem to help: crucially, when we invoke the Pauli twirl property in Equation~\eqref{eq:4}, we rely on the fact that it makes the cross terms vanish \emph{perfectly}. If cross terms vanished only \emph{approximately}, the double sum over $i,j$ would cause the error to blow up (given the tradeoff between key length and precision).
\end{remark}

\paragraph{Commitment in the CHRS model} As a direct corollary, we can construct an unconditional quantum bit commitment protocol in the CHRS model. We first recall Morimae and Yamakawa's scheme~\cite{MY22a}. To commit to the bit $b\in \{0,1\}$, the sender generates
\begin{equation*}
    \ket{\psi_b}:=\frac{1}{\sqrt{2^{2m+n}}} \sum_{x,z \in \{0,1\}^m} \sum_{k \in \{0,1\}^n } \ket{x,z,k} \otimes P_{x,z}^b \ket{\phi_k},
\end{equation*}
where $\{ \ket{\phi_k}\}_k$ is the $\oprs$ family, with key-size $n$ and outputs size $m$, and $P_{x,z}:=\bigotimes_{j=1}^{m} X_j^{x_j}Z_j^{z_j}$. 
To commit, only the right register is sent to the receiver. (The hiding property can be seen easily: note that if $b=1$, the state is maximally mixed by the properties of the quantum one-time pad, and if $b=0$, the state is a random $\oprs$ state; these two cases are  indistinguishable, by the $\oprs$ property.) 
To reveal, the committer sends the rest of the state and the bit $b$. The receiver applies $V_b^\dagger$, where $V_b\ket{0 \ldots 0}=\ket{\psi_b}$, measures all the qubits, and accepts if and only if the outcome is $0\ldots 0$. 
As mentioned, the problem is that Applying $V_b^\dagger$ requires the inverse transformation of the one generating the $\oprs$ state and cannot be done in a black-box manner.

\bnote{New paragraph on commitment from unitary simulation oracle}  Recall that our $\oprs$ takes the form $\ket{\phi_k} = (X^aZ^b \otimes I) \ket{\psi}$. Thus, to invert the generation algorithm of our $\oprs$, we need to map $\ket{\psi}$ to $\ket{0}$. In~\cref{sec:tech-state-to-unitary}, we show that our separation between $\oprs$ and $\prs$ also holds relative to a self-inverse unitary oracle. Therefore, applying the inverse transformation can be done efficiently\footnote{In a prior version of this work, a construction in the CHRS model was shown by adapting the commitment construction of~\cite{morimae2023unconditionally}, which does not use the inverse.}. 

\subsection{Oracle separation between \texorpdfstring{$\prs$}{PRS} and \texorpdfstring{$\oprs$}{1PRS}}
\label{sec:separation-techoverview}
We now describe an oracle relative to which $\oprs$ exist, but $\prs$ do not. We consider the CHRS model augmented with quantum oracle access to a unitary QPSPACE machine\footnote{As mentioned previously, the CHRS oracle, which provides copies of the Haar random state, can be thought of as implementing an \emph{isometry}. This is spelled out in \cref{sec:1prs_CHRS}. On the other hand, the QPSPACE machine takes as input a state $\ket{\alpha}$, and the description of a unitary circuit $C$ computable in ``polynomial space'', and returns $C\ket{\psi}$. For a precise definition, we refer the reader to \ifproceedingelse{the full version}{the start of \cref{sec:oralce-sep}}.}. Going forward, we refer to the former as the ``CHRS oracle'' and to the latter as the ``QPSPACE oracle''. We refer the reader to \ifproceedingelse{the full version}{the start of \cref{sec:oralce-sep}} for a precise definition of the QPSPACE oracle.

The existence of $\oprs$ in this model follows immediately from the fact that our construction in the CHRS model achieves \emph{statistical}, rather than computational, security when the adversary has polynomially many copies of the common Haar random state. Thus, the QPSPACE oracle (which is independent of the sampled Haar random state), does not help the adversary.

On the other hand, we show that a $\prs$ does not exist in this model. We describe an explicit attack on any $\prs$ construction.
\paragraph{Breaking $\prs$ security via the ``Quantum OR Lemma''.} Notice that, in this model, since the CHRS oracle is \emph{input-less}, we can assume, without loss of generality, that any algorithm that uses the CHRS oracle makes all of its calls to it at the start, i.e.\ the algorithm first obtains all of the copies of $\ket{\psi}$ that it needs, and then proceeds without making any additional call to the CHRS oracle. Thus, any $\prs$ construction takes the following form\footnote{Again, technically, the construction could make use of states $\ket{\psi_m}$ for different values of $m$ (at most polynomially different values). This does not affect the argument very much, and, for simplicity, in this technical overview, we consider constructions that use only copies of $\ket{\psi_m}$ for a single $m$.}. Let $\ket{\psi}$ be the common Haar random state. Then, the family of pseudorandom states is $\{\ket{\phi_k}\}_{k \in \{0,1\}^n}$, with 
$$\ket{\phi_k} = \gen_k \parens{\ket{\psi}^{\otimes r} \otimes \ket{0^t}} \,,$$
for some $r$ and $t$ polynomial in $n$, and $\gen_k$ a unitary that is efficiently computable given access to the QPSPACE oracle.

\vspace{2mm}
The problem of breaking the PRS is then the following: given polynomially many copies of $\ket{\tilde{\phi}}$, where either (i) $\ket{\tilde{\phi}} = \ket{\phi_k}$ for some $k$, or (ii)  $\ket{\tilde{\phi}}$ is Haar random (independent of $\ket{\psi}$), decide which is the case. Notice that this problem can be recast as follows, for some appropriate projections $\{\Lambda_k \}_{k \in \{0,1\}^n}$, and some constants $a,b$ with $b-a >0$.

Given $\ket{\tilde{\phi}}$ as above, and $r$ copies of $\ket{\psi}$, determine whether 
\begin{itemize}
\if\widemargin=0
\item[(i)] There exists $k \in \{0,1\}^n$ such that $\tr\left[\Lambda_k  \left(\ket{\tilde{\phi}}\bra{\tilde{\phi}}\otimes (\ket{\psi}\bra{\psi})^{\otimes r} \otimes (\ket{0}\bra{0})^{\otimes t}\right)\right] > b$, or
\else
\item[(i)] There exists $k \in \{0,1\}^n$ such that \[\tr\left[\Lambda_k  \left(\ket{\tilde{\phi}}\bra{\tilde{\phi}}\otimes (\ket{\psi}\bra{\psi})^{\otimes r} \otimes (\ket{0}\bra{0})^{\otimes t}\right)\right] > b\text{ , or}\] 
\fi
\item[(ii)] For all $k \in \{0,1\}^n$, $\tr\left[\Lambda_k  \left(\ket{\tilde{\phi}}\bra{\tilde{\phi}}\otimes (\ket{\psi}\bra{\psi})^{\otimes r} \otimes (\ket{0}\bra{0})^{\otimes t}\right)\right]  <a$.
\end{itemize}
What are the projections $\Lambda_k$? For clarity, let's denote the registers in $\ket{\tilde{\phi}}\bra{\tilde{\phi}}\otimes (\ket{\psi}\bra{\psi})^{\otimes r} \otimes (\ket{0}\bra{0})^{\otimes t}$ as $ \ket{\tilde{\phi}}\bra{\tilde{\phi}}_{\mathsf{A}}\otimes (\ket{\psi}\bra{\psi})^{\otimes r}_{\mathsf{B}}\otimes (\ket{0}\bra{0})^{\otimes t}_{\mathsf{C}}$. Then, in words, $\Lambda_k$ applies $\gen_k$ to registers $\mathsf{BC}$, followed by a ``swap test'' between $\mathsf{A}$ and $\mathsf{BC}$ (projecting onto the ``accept'' outcome of the swap test). Formally,
$$\Lambda_k = (I_{\mathsf{A}} \otimes \gen_{k,\mathsf{BC}}) \, \Pi_{sym}^2 \,(I_{\mathsf{A}} \otimes \gen_{k,\mathsf{BC}}) \,,$$
where $\Pi_{sym}^2$ is the projection onto the symmetric subspace over $\mathsf{A}$ and $\mathsf{BC}$.

Importantly, the latter problem takes a form that is \emph{almost} amenable to the ``quantum OR lemma''~\cite{harrow2017sequential}. The version of the ``quantum OR lemma'' that is relevant here informally says that there is an algorithm that requires only a \emph{single} copy of $\ket{\tilde{\phi}} \ket{\psi}^{\otimes r} \ket{0}^{\otimes t}$ such that:
\begin{itemize}
\item in case (i), outputs 0 with probability at least $b^2/7$.
\item in case (ii), outputs 0 with probability at most $4 \cdot 2^n \cdot a$.
\end{itemize}
Moreover, the algorithm uses a number of auxiliary qubits that is \emph{logarithmic} in the number of projections. Since the number of projections is $2^n$, the number of auxiliary qubits is only polynomial in $n$, and thus the algorithm can be implemented by invoking the QPSPACE oracle\footnote{For the algorithm to be implementable by a QPSPACE machine, we additionally need that each measurement $\{\Lambda_k, I-\Lambda_k\}$ be also implementable by a QPSPACE machine, which is the case in this setting since $\gen(k)$ and the ``swap test'' are efficient. The attentive reader will notice that there is one subtlety about the latter, namely that $\gen(k)$ is itself allowed to make queries to the QPSPACE oracle! However, this is not an issue, since the resulting computation can still be simulated using a QPSPACE oracle. We again refer the reader to \ifproceedingelse{the full version}{the start of \cref{sec:oralce-sep}} for a definition of the QPSPACE oracle.}. 

Unfortunately, in the setting described above, $a,b$ are constant: in particular, $a$ is approximately $\frac12$, while $b=1$. Thus, the guarantee above is not useful because of the factor of $2^n$! There is a natural way to get around this, which is to use ``parallel repetition'': the projections $\Lambda_k$ should act on $\poly(n)$ copies of the state considered above, and perform $\poly(n)$ swap tests. As a result of the amplification, we then have $a = 2^{-\poly(n)}$, which is sufficient to give an exponentially small upper bound in case (ii), and to distinguish between cases (i) and (ii), thus breaking security of the $\prs$. Crucially, this attack can be carried out because the security game of a $\prs$ allows the adversary access to polynomially many copies of $\ket{\tilde{\phi}}$. The same attack does not work in the case of a $\oprs$!

\begin{remark}
One might wonder whether a different attack based on shadow tomography would work here (along the lines of the attack described by Kretschmer in \cite[Subsection 1.3]{Kre21}). The issue is that here $\tr[\Lambda_k^2]$ is exponentially large, and so the estimation of the quantity $\tr[\Lambda_k \tilde{\phi}]$ given by shadow tomography has too large of a variance. Thus, shadow tomography does not seem to be sample-efficient in this setting.
\end{remark}

\subsection{Upgrading our separations from a ``state'' oracle to a unitary oracle}
\label{sec:tech-state-to-unitary}
Recall that the oracle separating $\oprs$ and $\prs$ in~\cref{sec:separation-techoverview} is an \emph{isometry}. In particular, the CHRS part of the oracle provides copies of a Haar random state. Thus, so far, such a separation only rules out a fully black-box construction of a $\prs$ from ``isometry access'' to a $\oprs$ (as defined precisely in \ifproceedingelse{the full version}{\cref{def:bb-isometry}}). Informally, such a black-box construction is only allowed to use the generation procedure of the $\oprs$ as an ``isometry'', i.e.\ it does not have the ability to initialize the auxiliary qubits in an arbitrary state.

In this section, we informally describe how our separation can be upgraded to be relative to a \emph{unitary} oracle (and its inverse). For the full details, see \ifproceedingelse{the full version}{\cref{sec:unitary-oracle-simulation}}. In particular, we introduce a unitary oracle, which is self-inverse, that is approximately equivalent to the isometry oracle that gives out copies of a Haar random state $\ket{\psi}$: access to this unitary oracle allows one to exactly simulate access to copies of $\ket{\psi}$, and, conversely, the unitary oracle can be simulated \emph{approximately} using copies of $\ket{\psi}$. Replacing the isometry oracle with the new unitary oracle, we are able to establish impossibility of the most general kind of a fully black-box construction of $\prs$ from $\oprs$ (as in \ifproceedingelse{the full version}{\cref{def:bb-unitary-and-inverse}}). Our technique is inspired by techniques by Ji, Liu, and Song~\cite{JLS18} and Zhandry~\cite{zhandry2024space}, with some differences, which we describe further in the full version.

\subsubsection{Unitary corresponding to a state}
\label{sec:unitaryoracle0-1}
Throughout the section, let $\ket{\psi}$ be an $n$-qubit state orthogonal to $\ket{0^n}$. In the CHRS model, the common Haar state $\ket{\psi}$ is not necessarily orthogonal to $\ket{0^n}$, but we take them to be be orthogonal at first for simplicity. The result we prove will extend straightforwardly to the case of arbitrary $\ket{\psi}$. For convenience of notation, we will write $\ket{0}$ instead of $\ket{0^n}$ (more generally, we will use $\ket{0}$ to denote the all zero state of a system whose dimension is clear from the context).

We define a corresponding unitary $\upsi$ as follows: $\upsi$ flips $\ket{0}$ and $\ket{\psi}$, and acts as the identity on everything orthogonal to the subspace spanned by $\ket{0}$ and $\ket{\psi}$, i.e.\ $\upsi\ket{0}=\ket{\psi}$, $\upsi\ket{\psi}=\ket{0}$, and $\upsi\ket{\phi}=\ket{\phi}$ for any $\ket{\phi}$ orthogonal to $\ket{0}$ and $\ket{\psi}$. Notice that $\upsi$ is self-inverse.

It is clear that access to $\upsi$ allows one to simulate the isometry oracle (which provides copies of $\ket{\psi}$), by simply applying $\upsi$ on copies of $\ket{0}$. However, the reduction in the other direction is nontrivial. First of all, notice that we cannot hope to simulate $\upsi$ in the most general sense using the isometry oracle alone, because the phase information is entirely lost: the states of the form $\alpha \ket{\psi}$, for $|\alpha|=1$, are all identical up to a global phase, and so $\alpha$ cannot be detected given only copies of the state. On the other hand, the unitaries of the form $\uapsi$ are in general very different from each other: applying $\uapsi$ or $U_{\alpha' \ket{\psi}}$ (for $\alpha \neq \alpha'$) to a superposition of $\ket{0}$ and $\ket{\psi}$ produces different states in general.

So, instead, our goal will be to show that $\upsi$ can be simulated using copies of $\ket{\psi}$ in a weaker sense, which will still be sufficient to upgrade our oracle separation results. Our simulation technique is similar to the one proposed by Zhandry~\cite{zhandry2024space}, with some differences which we remark in the full version. Our key observation is that, while a general simulation is not possible, one might be able to simulate the behaviour of $\uapsi$ ``on average over $\alpha$''. Consider an algorithm $\calA^{\upsi}$ that makes $T$ queries to $\upsi$, we will show that one can simulate $\calA^{\upsi}$ with $\epsilon$ precision given $O\left(\frac {T^2}{\epsilon^2}\right)$ copies of $\ket{\psi}$ in the following average sense. 

For any $\ket{\psi}$, and an arbitrary input state $\ket{\sigma}$, we can write the output of $\calA^{\upsi}$ as
\begin{equation*}
    \ket{\Psi_{\psi,T}} = B_T\upsi B_{T-1} \dots B_1 \upsi B_0 \ket{\sigma},
\end{equation*}
for some fixed unitaries $B_0, \dots, B_T$ that do not depend on $\ket{\psi}$. Then, we consider the average of this output over a uniformly random phase $\alpha$, namely $\alpha$ is sampled as a random point on the unit circle $|\alpha|=1$:
\begin{equation}
\label{eq:600-1}
    \rho_{\psi, T} = \E_\alpha \left[ \ket{\Psi_{\alpha{\ket{\psi}}, T}}\bra{\Psi_{\alpha\ket{\psi}, T}} \right].
\end{equation}
We establish that $\rho_{\psi,T}$ can be simulated approximately given copies of $\ket{\psi}$.
\begin{theorem}
\label{thm:isometry-to-unitary-1}
    Let $n\in \mathbb{N}$. Let $\ket{\psi}$ be any $n$-qubit state orthogonal to $\ket{0^n}$. Let $\epsilon >0$, and $T\in \mathbb{N}$. Let $\upsi$ be the $n$-qubit unitary defined as above, and let $\rho_{\psi, T}$ be as in Equation \eqref{eq:600-1}. For any oracle algorithm $\calA^{(\cdot)}$ making $T$ queries to $\upsi$, there is an algorithm $\wt{\cal{A}}$ that, with access to $O\left(\frac {T^2}{\epsilon^2}\right)$ copies of $\ket{\psi}$, outputs a state $\wt{\rho}_{\psi, T}$ that is $\eps$-close to $\rho_{\psi, T}$ in trace distance.
\end{theorem}

\begin{corollary}
\label{cor:20-1}
Let $n\in \mathbb{N}$. Let $\ket{\psi}$ be any $n$-qubit state.  Let $\epsilon >0$, and $T\in \mathbb{N}$. Define the $(n+1)$-qubit state $\ket{\psi'} = \ket{\psi} \otimes \ket{1}$. Let $U_{\ket{\psi'}}$ be the $(n+1)$-qubit unitary defined as above, and let $\rho_{\psi', T}$ be as in Equation \eqref{eq:600-1}.
For any oracle algorithm $\calA^{(\cdot)}$ making $T$ queries to $U_{\ket{\psi'}}$, there is an algorithm $\wt{\cal{A}}$ that, with access to $O\left(\frac {T^2}{\epsilon^2}\right)$ copies of $\ket{\psi}$, outputs a state $\wt{\rho}_{\psi', T}$ that is $\eps$-close to $\rho_{\psi', T}$ in trace distance.
\end{corollary}

Corollary \ref{cor:20-1} follows immediately from Theorem \ref{thm:isometry-to-unitary-1}. We prove Theorem \ref{thm:isometry-to-unitary-1} in \ifproceedingelse{the full version}{\cref{sec:unitary-oracle-simulation}}.

The proof proceeds in two steps. The first step \ifproceedingelse{(as in the full version)}{(\cref{sec:unitaryoracle1})} is to show that $\rho_{\psi, T}$ can be produced \emph{perfectly} with access to $T$ copies of $\ket{\psi}$ \emph{and} a certain auxiliary unitary oracle $C_{\ket{\psi}}$. The second step \ifproceedingelse{(as in the full version)}{(\cref{sec:unitaryoracle2})} is to show that $C_{\ket{\psi}}$ can be simulated approximately using copies of $\ket{\psi}$. In \ifproceedingelse{the full version}{\cref{sec:weak-simulation}}, we justify why the weak notion of simulation that we achieve is sufficient to lift our separation results to be relative to a unitary oracle. Our lifting result applies to any Common Reference Quantum State (CRQS) oracle (i.e.\ an oracle providing copies of a state -- not necessarily Haar random) which has a ``global-phase'' invariance -- see \ifproceedingelse{the full version~\cite{arxivfullversion}}{\cref{sec:unitary-oracle-simulation}}. Stated informally, we show the following.
\begin{theorem}[Informal]
Suppose $\oprs$ exist relative to a global-phase invariant state oracle $\mathcal{O}$, and $\prs$ with output length $\omega(\log n)$, do not. Then, there also exists a parametrized unitary oracle $\mathcal{U}$ relative to which $\oprs$ exist, but $\prs$, with output length $\omega(\log n)$ do not.
\end{theorem}

\section{Preliminaries}
\label{sec:preliminaries}

\paragraph{Notation.}
We will use the letter $n$ to denote the security parameter. We denote by $\mu_d$ the Haar measure in $d$ dimensional Hilbert space. The notation $\ket{\psi}\gets \mu_d$ denotes sampling a state according to $\mu_d$. For any finite set $K$, we write $k \gets K$ to mean that $k$ is sampled uniformly at random from $K$. We use the notation $A^{(\cdot)}$ to refer to an algorithm (classical or quantum) that makes queries to an oracle. For an operator $H$, we use the notation $\norm{H}$ to denote its trace norm. For a pure state $\ket{\psi}$, we denote by $\psi$ the density matrix $\ket{\psi}\bra{\psi}$. We will use $\Pi^{sym}$ to refer to the projector corresponding to a swap test. The definition of swap test can be found, for example, in \cite{buhrman2001quantum}.

\begin{definition}[Pseudorandom States ($\prs$), adapted from~\cite{JLS18}]
  \label{def:prs}
A pseudorandom states family is a QPT algorithm $\gen$ 
 that, on input $k\in \{0,1\}^n$, outputs a pure state $\ket{\phi_k}$ consisting of $m=m(n)$ qubits. For security, we require the following pseudorandomness property: for any polynomial $t=t(n)$ and any QPT adversary $\mathcal A$, there exists a negligible function $\negl$
 such that for all $n$,

    \begin{equation}
      \abs{\Pr_{k\gets\{0,1\}^n } \bigl[ \mathcal A(\ket{\phi_k}^{\otimes t}) = 1 \bigr] -
        \Pr_{\ket{\phi}\gets \mu_{2^m} } \bigl[ \mathcal A(\ket{\phi}^{\otimes t}) = 1 \bigr] } =
      \negl(n),
      \label{eq:prs}
    \end{equation}
    where $\mu_{2^m}$ is the Haar measure on $m(n)$ qubit states. We say that the construction is \emph{statistically secure} if \cref{eq:prs} holds for computationally unbounded adversaries. We emphasize that these unbounded adversaries receive only polynomially many copies of the Haar random state.  
    For constructions relative to an oracle $\mathcal O$, both the generation algorithm $G$ and the adversary $\mathcal A$ get oracle access to $\mathcal O$.
\end{definition}

\begin{definition}[Single-copy Pseudorandom States ($\oprs$), adapted from ~\cite{MY22a}]
  \label{def:oprs}
  Single-copy pseudorandom states ($\oprs$) with computational and statistical security are defined as \cref{def:prs}, with two modifications:
  \begin{enumerate}
      \item (single-copy security) \cref{eq:prs} holds only for $t=1$.
      \item (stretch) For every $n$, $m(n)>n$.
  \end{enumerate} 
\end{definition}
Several aspects are worth mentioning regarding this definition:
\begin{itemize}
    \item Any pseudorandom generator ($\prg$) is also a $\oprs$.
    \item A $\prg$ is never a (multi-time) $\prs$: a distinguisher can measure in the standard basis multiple copies. For the $\prg$, the outputs from the different copies will always be the same with probability 1, but not so for a Haar-random state.
    \item Without the stretch requirement, the family $\ket{\psi_k}=\ket{k}$ would have been a $\oprs$: the security requirement is that $\frac{1}{|\mathcal K|} \sum_{k\in \mathcal K} \ket{\psi_k}\bra{\psi_k}$ is computationally indistinguishable from the maximally mixed state, which holds for this simple construction. 
    \item It has been shown in~\cite[Theorem C.2 only in the arXiv version]{GJMZ23} that $\prs$ implies $\oprs$ via a black-box construction. This is non-trivial since $m$ may be shorter than $n$ in a $\prs$.
\end{itemize}

We also need some technical lemmas throughout the proof.
\begin{lemma}[L\'evy's lemma, e.g., adapted from {\cite[Theorem 7.37]{watrous2018theory}}] \label{lem:levy-lemma} 
Let $\eta > 0, \delta >0$, and $m \in \mathbb{N}$. Let $f:\mathbb C^{2^m} \to \mathbb R$ be an $\eta$-Lipschitz function. Then,
\begin{equation*}
    \Pr_{\ket{\psi}\gets \mu_{2^m}} \Big[\big|f(\ket{\psi})-\E_{\ket{\psi}\gets \mu_{2^m}} f(\ket{\psi})\big| \geq \delta \Big] \leq 4\exp\parens{-\frac{C_12^m\delta^2}{\eta^2}},
\end{equation*}
where $C_1$ can be taken to be $\frac{2}{9\pi^3}$.
\end{lemma}
 \section{Construction of 1PRS in the CHRS model}
\label{sec:oprs_in_chrs_model}
In this section, we prove one of the main technical contributions of the paper: $\oprs$ exist unconditionally in the CHRS model.
\begin{theorem}
\label{thm:existence-1prs}
    Statistically secure $\oprs$ exist in the \CHRS\ model\footnote{See~\cref{def:1prs_CHRS} in~\cref{sec:1prs_CHRS}.}.
\end{theorem}
This section is organized as follows. In \cref{sec:1prs_CHRS}, we formally define the CHRS model, as well as the notions of $\prs$ and $\oprs$ in this model. In \cref{ssec:random-max-ent-state}, we show that a one-time pad acting on \emph{exactly half} of the qubits of a Haar random state is sufficient to ``scramble'' it, so that it is statistically indistinguishable from a maximally mixed state (even given polynomially many copies of the same Haar random state). The main tool in the proof is a theorem from Harrow \cite{harrow2023approximate}, about applying Haar random unitaries to one half of a maximally entangled state.
In \cref{ssec:ampli-random}, we show a key technical step: the ``scrambling'' property persists even if the quantum one-time pad is applied to \emph{slightly less than half} of the qubits of the Haar random state, which can be interpreted as saying that the quantum pseudorandomness can be ``amplified'' slightly. This is enough to yield a $\oprs$.

\subsection{The CHRS model}
\label{sec:1prs_CHRS}
The Common Haar Random State (CHRS) model can be viewed as a quantum state generalization of the Common Reference String (CRS) model 
introduced by  \cite{canetti2001universally}. In the CHRS model, we assume a trusted third party, who prepares a family of states $\calS = \{\ket{\psi_m}\}_{m \in \N}$, where $\ket{\psi_m}$ is sampled according to the Haar measure on $m$ qubits $\mu_{2^m}$. All parties in a protocol (including the adversary) have access to polynomially many (in the security parameter $\mylambda$) copies of states from $\calS$. Formally, parties have access to the family of isometries $\{V_m\}_{m \in \mathbb{N}}$, where $V_m: \mathbb{C} \rightarrow \mathbb{C}^{2^m}$\footnote{Notice that the domain is one-dimensional.} is such that 
$$ V_m: \ket{0} \mapsto \ket{\psi_m}\,.$$
Equivalently, for any state $\ket{\alpha}$ of any dimension, one query to $V_m$ performs the map:
$$ \ket{\alpha} \mapsto \ket{\alpha} \ket{\psi_m} \,.$$
We clarify that, in this model, parties cannot query the different isometries ``in superposition''. Rather, they can query each $V_m$ individually (provided they have enough space to store the $m$-qubit output state $\ket{\psi_m}$). The model is meant to capture the scenario where parties can request copies of $\ket{\psi_m}$, for any $m$ of their choice, from the trusted third party, as long as they have enough space to store the requested state.

\paragraph{Pseudorandom states in the \CHRS\ model}
We formally define the notion of (single-copy) pseudorandom states in the \CHRS\ model. The definition is as in the ``plain model'' (\cref{def:prs,def:oprs}), except that both the generation algorithm and the adversary may use polynomially many copies of the CHRS states.

\begin{definition}[$\prs$ in the \CHRS\ model]\label{def:1prs_CHRS}
    Let $\calS = \{\ket{\psi_m}\}_{m \in \N}$ denote the CHRS family of states. A pseudorandom state (PRS) family in the \CHRS\ model is a QPT algorithm $\gen$ satisfying the following.  There exist polynomials $m,r: \mathbb{N} \rightarrow \mathbb{N}$ such that
    \begin{itemize}
    \item $\gen$: takes as input a security parameter $1^n$, a string $k \in \{0,1\}^n$, and states $\ket{\psi_1}^{\otimes r(n)}$, \ldots $,\ket{\psi_{r(n)}}^{\otimes r(n)} \in \calS$, and outputs a pure state $\ket{\phi_k}$ consisting of $m=m(n)$ qubits\footnote{Clearly, taking $\gen$ of this form is without loss of generality.}.
    \end{itemize}
Moreover, the following computational (resp.\ statistical) pseudorandomness property should be satisfied: for any polynomials $t, r': \mathbb{N} \rightarrow \mathbb{N}$, and any QPT (resp.\ unbounded quantum) adversary $\calA$, there exists a negligible function $\negl$ such that, for all $n$,
    \begin{align*}
      \Bigg|&\Pr_{k\gets\{0,1\}^n,\, \calS } \bigl[ \mathcal A(\ket{\phi_k}^{\otimes t(n)}, \ket{\psi_1}^{\otimes r'(n)}, \dots ,\ket{\psi_{r'(n)}}^{\otimes r'(n)}) = 1 \bigr] - \\
        &\Pr_{\ket{\psi}\gets \mu_{2^m},\, \calS } \bigl[ \mathcal A(\ket{\phi}^{\otimes t(n)}, \ket{\psi_1}^{\otimes r'(n)}, \dots ,\ket{\psi_{r'(n)}}^{\otimes r'(n)} ) = 1 \bigr] \Bigg| =
      \negl(n) \,,
      \label{eq:prs-chrs}
    \end{align*}
    where we clarify that the probabilities are also over sampling the states in $\mathcal{S}$.
    The definition of $\oprs$ in the \CHRS\ model is analogous, except that $t = 1$, and it must be that $m(n)>n$ for all $n$.
\end{definition}

For clarity, we state the \emph{statistical} pseudorandomness property of a $\oprs$ explicitly. We focus on the case where $\gen$, for security parameter $1^n$, only takes as input a \emph{single} Haar random state $\ket{\psi_{m(n)}}$, since this is the setting of our construction. In this case, the \emph{statistical} pseudorandomness property simplifies to the following\footnote{While the construction itself may only use the state $\ket{\psi_{m(n)}}$, the (unbounded) adversary may still access other states from $\mathcal{S}$. However, it is clear that these additional states do not affect the trace distance in \cref{eq:1prs-pseudorandomness} at all.}: for any $r = \poly(n)$, there exists a negligible function $\negl$ such that, for all $n$,
\begin{equation}
\label{eq:1prs-pseudorandomness}
     \left\lVert\E_{k\gets \{0,1\}^n} \E_{\ket{\psi_m} \gets \mu_{2^m}} U_k\psi_m U_k^\dagger \otimes \psi_m^{\otimes r-1} - \E_{\ket{\psi_m} \gets \mu_{2^m}} \frac{1}{2^m} \I \otimes \psi_m^{\otimes r-1} \right\rVert = \negl(n) \,.
\end{equation}

\subsection{Quantum one-time pad on \emph{exactly half} of the qubits of a Haar random state}
\label{ssec:random-max-ent-state}
In this section, we show that a quantum one-time pad (QOTP) acting on \emph{exactly half} of the qubits of a Haar random state is sufficient to ``scramble'' it, so that it is statistically indistinguishable from a maximally mixed state (even given polynomially many copies of the same Haar random state). The main tool in the proof is the following theorem from Harrow \cite{harrow2023approximate}.

Let $\ket{\phi_U} \coloneqq (U \otimes I) \ket{\Phi_d}$, where $\ket{\Phi_d} = \frac{1}{\sqrt{d}} \sum_{i=0}^{d-1} \ket{ii}$ denotes the maximally entangled state in $\C^d \otimes \C^d$ and $U\in SU(d)$ is a $d$-dimensional unitary.
\begin{lemma}[adapted from {\cite[Theorem 3]{harrow2023approximate}}]\label{lem:harrow-max-entangle}
    Assume $r^2 \leq d$, then
    \begin{equation*}
        \norm{\E_{\ket{\psi} \gets \mu_{d^2}} [\psi^{\otimes r}] - \E_{U \gets SU(d)} [\phi_U^{\otimes r}]} \leq \frac{r^2}{d} \,,
    \end{equation*}
    where the norm on the LHS is the trace norm.
\end{lemma}

We now describe a ``toy construction'' of a $\oprs$ in the CHRS model, which consists of applying a QOTP to exactly the first half of the qubits of the Haar random state.  Crucially, this construction \emph{does not} satisfy the length stretching requirement of a $\oprs$ (which is handled in \cref{ssec:ampli-random}). Nonetheless, we prove that the construction in \cref{fig:toy-1prs-protocol} satisfies the statistical pseudorandomness property of a $\oprs$ (from \cref{eq:1prs-pseudorandomness}). Recall that to describe the construction we just need to specify, for each value $n$ of the security parameter, a family $\{U_k\}_{k \in \{0,1\}^n}$ of $m$-qubit unitaries, where, in the case of this ``toy'' example, $m = n$. Then, for a seed $k$, and a common Haar random $m$-qubit state $\ket{\psi}$, the corresponding $\oprs$ state is $\ket{\phi_k} = U_k \ket{\psi}$.


\begin{figure}[htb]
\vspace{10pt}
\centering
\begin{mdframed}[userdefinedwidth=0.8\textwidth, align=center]
Let $n \in \mathbb{N}$ be even (otherwise redefine $n$ to be $n-1$). Let $U_k = X^aZ^b \otimes \I_{n/2}$, where $a, b \in \{0,1\}^{n/2}$ are the first and second halves of $k$ respectively.


\end{mdframed}
\caption{A construction that satisfies the statistical pseudorandomness property of a $\oprs$ in the CHRS model, but not the length-stretching requirement.}
\label{fig:toy-1prs-protocol}
\end{figure}



We will use the following ``Pauli twirl'' lemma.
\begin{lemma}[Pauli twirl]\label{lem:quat-otp}
    Let $m \in \mathbb{N}$. Let $\rho$ be an arbitrary linear operator on the space of $m$ qubits. Let $\calP_m$ be the set of Pauli operators on $m$ qubits, 
    Then, we have 
    \begin{equation}
    \label{eq:pauli-channel-twirl}
        \E_{P \gets \calP_m} P\rho P^{\dagger} = \frac{\tr [\rho]}{2^m} \I \,.
    \end{equation}
\end{lemma}



We now show that the construction in \cref{fig:toy-1prs-protocol} satisfies the statistical pseudorandomness property (from \cref{eq:1prs-pseudorandomness}).
\begin{theorem}\label{thm:partial-one-time-pad}
    Let $m,r \in \mathbb{N}$ such that $m$ is even, and $r \leq 2^{\frac{m}{2}}$. Then, the family of unitaries $\braces{U_k}_{k\in \{0,1\}^m}$ from \cref{fig:toy-1prs-protocol} satisfies
    \begin{equation*}
         \left\lVert\E_{k\gets \{0,1\}^m} \E_{\ket{\psi} \gets \mu_{2^m}} U_k\psi U_k^\dagger \otimes \psi^{\otimes r-1} - \E_{\ket{\psi} \gets \mu_{2^m}} \frac{1}{2^m} \I \otimes \psi^{\otimes r-1} \right\rVert \leq \frac{2r^2}{2^{m/2}}
    \end{equation*}
\end{theorem}
\begin{proof}
    Recall that $U_k = X^aZ^b \otimes \I_{m/2}$, where $a, b \in \{0,1\}^{m/2}$ are the first and second halves of $k$.

    Then, we have
\if\widemargin=0
\begin{equation}\label{eq:max-entangle-1}
    \begin{aligned}
        \bigg\lVert\E_k \E_\psi &(U_k \ot \I^{\ot r-1}) \psi^{\otimes r} (U_k^\dagger \ot \I^{\ot r-1}) - \E_\psi \frac{\I}{2^m} \otimes \psi^{\otimes r-1}\bigg\rVert \leq \\
        &\bigg\lVert \E_k (U_k \otimes \I^{\ot r-1}) \E_\psi \psi^{\otimes r} (U_k^\dagger \otimes \I^{\ot r-1}) -  \E_k (U_k \otimes \I^{\otimes r-1}) \E_U \phi_U^{\otimes r} (U_k^\dagger \otimes \I^{\otimes r-1})\bigg\rVert +\\&  \norm{\E_k (U_k \otimes \I^{\otimes r-1}) \E_U \phi_U^{\otimes r} (U_k^\dagger \otimes \I^{\otimes r-1}) - \E_U \frac{\I}{2^m} \otimes \phi_U^{r-1}} +\norm{ \E_U \frac{\I}{2^m} \otimes \phi_U^{r-1} - \E_\psi \frac{\I}{2^m} \otimes \psi^{\otimes r-1}}  \\
        &\leq \frac{2r^2}{2^{m/2}} + \norm{\E_k (U_k \otimes \I^{\otimes r-1}) \E_U \phi_U^{\otimes r} (U_k^\dagger \otimes \I^{\otimes r-1}) - \E_U \frac{\I}{2^m} \otimes \phi_U^{r-1}},
    \end{aligned}
    \end{equation}
\else
\begin{equation}\label{eq:max-entangle-1}
    \begin{aligned}
        \bigg\lVert&\E_k \E_\psi (U_k \ot \I^{\ot r-1}) \psi^{\otimes r} (U_k^\dagger \ot \I^{\ot r-1}) - \E_\psi \frac{\I}{2^m} \otimes \psi^{\otimes r-1}\bigg\rVert \leq \\
        &\bigg\lVert \E_k (U_k \otimes \I^{\ot r-1}) \E_\psi \psi^{\otimes r} (U_k^\dagger \otimes \I^{\ot r-1}) -  \E_k (U_k \otimes \I^{\otimes r-1}) \E_U \phi_U^{\otimes r} (U_k^\dagger \otimes \I^{\otimes r-1})\bigg\rVert +\\&  \norm{\E_k (U_k \otimes \I^{\otimes r-1}) \E_U \phi_U^{\otimes r} (U_k^\dagger \otimes \I^{\otimes r-1}) - \E_U \frac{\I}{2^m} \otimes \phi_U^{r-1}} +\\&\norm{ \E_U \frac{\I}{2^m} \otimes \phi_U^{r-1} - \E_\psi \frac{\I}{2^m} \otimes \psi^{\otimes r-1}}  \\
        &\leq \frac{2r^2}{2^{m/2}} + \norm{\E_k (U_k \otimes \I^{\otimes r-1}) \E_U \phi_U^{\otimes r} (U_k^\dagger \otimes \I^{\otimes r-1}) - \E_U \frac{\I}{2^m} \otimes \phi_U^{r-1}},
    \end{aligned}
    \end{equation}
\fi
    where the first inequality follows from the triangle inequality, and the second inequality follows from \cref{lem:harrow-max-entangle}. Notice that 
    \if\widemargin=0
     \begin{equation*}
    \begin{split}
        \E_k (U_k \otimes \I^{\otimes r-1}) \E_U \phi_U^{\otimes r} (U_k^\dagger \otimes \I^{\otimes r-1})  &= \E_{P\gets \mathcal{P}_{m/2}}\E_U (PU \ot \I) \Phi_{2^{m/2}} (U^\dagger P^\dagger \ot \I) \ot \phi_U^{\ot r-1} \\
        &= \frac{1}{2^{m/2}} \E_U \left[ \sum_{i,j} PU\ket{i}\bra{j}U^\dagger P^\dagger \ot \ket{i} \bra{j} \ot \phi_U^{\ot r-1} \right]\\
        & = \frac{1}{2^{m/2}} \E_U \left[ \sum_{i} \frac{1}{2^{m/2}} \I \ot \ket{i}\bra{i} \ot \phi_U^{\ot r-1} \right]\\
        & = \frac{\I}{2^m} \otimes \E_U \left[\phi_U^{\ot r-1}\right],
    \end{split}
    \end{equation*}
    \else
    \begin{equation*}
    \begin{split}
        \E_k &(U_k \otimes \I^{\otimes r-1}) \E_U \phi_U^{\otimes r} (U_k^\dagger \otimes \I^{\otimes r-1})  \\&= \E_{P\gets \mathcal{P}_{m/2}}\E_U (PU \ot \I) \Phi_{2^{m/2}} (U^\dagger P^\dagger \ot \I) \ot \phi_U^{\ot r-1} \\
        &= \frac{1}{2^{m/2}} \E_U \left[ \sum_{i,j} PU\ket{i}\bra{j}U^\dagger P^\dagger \ot \ket{i} \bra{j} \ot \phi_U^{\ot r-1} \right]\\
        & = \frac{1}{2^{m/2}} \E_U \left[ \sum_{i} \frac{1}{2^{m/2}} \I \ot \ket{i}\bra{i} \ot \phi_U^{\ot r-1} \right]\\
        & = \frac{\I}{2^m} \otimes \E_U \left[\phi_U^{\ot r-1}\right],
    \end{split}
    \end{equation*}
    \fi
    where, in the third equality, we use \cref{lem:quat-otp}. So, the second term in the last line of \cref{eq:max-entangle-1} vanishes. Therefore, we have
    \begin{equation*}
         \norm{\E_k \E_\psi U_k \psi U_k^\dagger \otimes \psi^{r-1} - \E_\psi \frac{\I}{2^m} \otimes \psi^{\otimes r-1}} \leq \frac{2r^2}{2^{m/2}}\,,
    \end{equation*}
as desired.
\end{proof}

\subsection{``Stretching'' the quantum pseudorandomness}\label{ssec:ampli-random}
In this section, we show that the ``1PRS'' from \cref{thm:partial-one-time-pad} is still secure even if we the the QOTP is applied only to $0.45m$ qubits, and thus the key length is shrunk slightly to $n=\prsm$ bits.

More precisely, we show that the following construction (\cref{fig:full-1prs-protocol}) is a \emph{statistical} $\oprs$ in the CHRS model, i.e.\ it satisfies \cref{eq:1prs-pseudorandomness}. Again, recall that to describe the construction we just need to specify, for each value $n$ of the security parameter, a family $\{U_k\}_{k \in \{0,1\}^n}$ of $m$-qubit unitaries, where $m$ is the output length. Then, for a seed $k$, and a common Haar random $m$-qubit state $\ket{\psi}$, the corresponding $\oprs$ state is $\ket{\phi_k} = U_k \ket{\psi}$.
\vspace{.5cm}
\begin{figure}[htb]
\centering
\begin{mdframed}[userdefinedwidth=0.8\textwidth, align=center]
Let $n,m \in \mathbb{N}$, where $\prsm \leq n <m$, and $n$ is even (otherwise, redefine $n$ to be the $n-1$).
Define $U_k = X^aZ^b \otimes \I^{\otimes(m-n/2)}$, where $a, b \in \{0,1\}^{n/2}$ are the first and second halves of $k$ respectively. 
\end{mdframed}
\caption{Construction of a $\oprs$ in the CHRS model}
\label{fig:full-1prs-protocol}
\end{figure}
\vspace{0.5cm}

In the rest of this section, we show that the construction of \cref{fig:full-1prs-protocol} is indeed a $\oprs$. The key ingredient of our proof is a ``stretching'' result for quantum pseudorandomness in the CHRS model. Informally, this says the following: if there is a way to obtain ``$m$ qubits of single-copy pseudorandomness'' from $n$ bits of classical randomness (where $n$ should be thought of as being linear in $m$), then one can also obtain ``$m$ qubits of pseudorandomness'' from $n-1$ bits of classical randomness, with a slight loss in statistical distance (i.e.\ it is possible to save one classical bit of randomness). We emphasize that this ``stretching'' result applies specifically to the CHRS model, and, as is, does not apply to the plain model. We will eventually apply this result recursively starting from the construction of \cref{fig:toy-1prs-protocol} (QOTP on exactly half of the qubits), which by \cref{thm:partial-one-time-pad} yields ``$m$ qubits of pseudorandomness'' from $m$ bits of classical randomness. The stretching result is the following.
\begin{theorem}\label{thm:ammplify-random}
    Let $m,n,r \in \mathbb{N}$ with $r<m$. If $\{U_k\}_{k\in \{0,1\}^n}$ is a set of unitaries acting on $m-1$ qubits states, then we have
    \begin{align}
        &\norm{\E_{k \gets \{0,1\}^n} \E_{\ket{\psi} \gets \mu_{2^m}} (\I\ot U_k) \psi (\I\ot U_k^\dagger) \ot \psi^{\ot r-1} - \E_{\ket{\psi} \gets \mu_{2^m}} \frac{\I}{2^m} \ot \psi^{\ot r-1} } \nonumber\\&\leq 5\norm{\E_{k \gets\{0,1\}^n} \E_{\ket{\psi'} \gets \mu_{2^{m-1}}} U_k \psi' U_k^\dagger \ot \psi'^{\ot r-1 } - \E_{\ket{\psi'}} \frac{\I}{2^{m-1}} \ot \psi'^{\ot r-1}} + \frac{800r\sqrt{m}}{2^{m/2}} \,, \label{eq:thm48}
    \end{align}
\end{theorem}
Since it is easy to miss, we emphasize that, in the above theorem, $\ket{\psi}$ is a Haar random $m$-qubit state, while $\ket{\psi'}$ is a Haar random $(m-1)$-qubit state.

To prove \cref{thm:ammplify-random}, we will need two lemmas. The first says that a typical Haar random state on $m$ qubits is ``close'' to being maximally entangled across the $(1, m-1)$ bipartition (i.e.\ the bipartition that considers the first qubit as the ``left'' register, and the remaining $m-1$ qubits as the ``right'' register). More concretely, the mixed state obtained by sampling a Haar random $m$-qubit state is close (in trace distance) to the state obtained by sampling two Haar random $(m-1)$-qubit states $\ket{\psi_1}$ and $\ket{\psi_2}$, and outputting $ \ket{\psi'}= \frac{1}{\sqrt{2}} \ket{0}\ket{\psi_1} + \frac{1}{\sqrt{2}} \ket{1}\ket{\psi_2}$. More precisely, we establish the following lemma, which considers $r$ copies of the state.

\begin{lemma}\label{lem:haar-coeff-estimate}
    Let $m,r \in \mathbb{N}$. We have
    \begin{equation*}
        \norm{\E_{\ket{\psi} \gets \mu_{2^m}} \psi^{\ot r} - \E_{\ket{\psi_1}, \ket{\psi_2} \gets \mu_{2^{m-1}} } \psi'^{\ot r} } \leq \frac{80r\sqrt{m}}{2^{m/2}} \,,
    \end{equation*}
where $\ket{\psi'}= \frac{1}{\sqrt{2}} \ket{0}\ket{\psi_1} + \frac{1}{\sqrt{2}} \ket{1}\ket{\psi_2}$.
\end{lemma}

The proof of \cref{lem:haar-coeff-estimate} can be found in \ifproceedingelse{the full version~\cite{arxivfullversion}}{\cref{sec:proof-sandwich-lemma}}. We also need the following technical lemma, whose proof can also be found in \ifproceedingelse{the full version~\cite{arxivfullversion}}{\cref{sec:proof-sandwich-lemma}}.

\begin{lemma}\label{lem:qubit-amplify-ineq}
    For a Hermitian matrix $A$, if the inequality $\norm{\bra{a}_1 A \ket{a}_1} < \epsilon$ holds for all $\ket{a} \in \{\ket{0}, \ket{1}, \ket{+}, \ket{+i}\}$, then $\norm{A} < 10 \eps$.
\end{lemma}

The proof of \cref{lem:qubit-amplify-ineq} can be found in \ifproceedingelse{the full version~\cite{arxivfullversion}}{\cref{sec:proof-sandwich-lemma}}. We are now ready to prove \cref{thm:ammplify-random}.
\begin{proof}[Proof of \cref{thm:ammplify-random}]
    According to \cref{lem:qubit-amplify-ineq}, it suffices to show that, for all $\ket{a} \in \{\ket{0}, \ket{1}, \ket{+}, \ket{+i}\}$,
    \begin{align*}
        &\norm{\E_{k} \E_{\psi} \bra{a}_1 (\I\ot U_k) \psi (\I\ot U_k^\dagger) \ket{a}_1 \ot \psi^{\ot r-1} - \E_{\psi} \bra{a}_1  \frac{\I}{2^{m}} \ket{a}_1 \ot \psi^{\ot r-1} } \leq 
        \\
        &\quad\frac{1}{2}\norm{\E_k \E_{\psi_1} U_k \psi_1 U_k^\dagger \ot \psi_1^{\ot r-1 } - \E_{\psi_1} \frac{\I}{2^{m-1}} \ot \psi_1^{\ot r-1}} + \frac{80r\sqrt{m}}{2^{m/2}} \,.
    \end{align*}
By the unitary invariance of the Haar measure, the LHS is identical for all $\ket{a} \in \{\ket{0}, \ket{1}, \ket{+}, \ket{+i}\}$. Thus, it suffices to show that 
\begin{align*}
        &\norm{\E_{k} \E_{\psi} \bra{0}_1 (\I\ot U_k) \psi (\I\ot U_k^\dagger) \ket{0}_1 \ot \psi^{\ot r-1} - \E_{\psi} \bra{0}_1  \frac{\I}{2^{m}} \ket{0}_1 \ot \psi^{\ot r-1} } \leq 
        \\&\quad\frac{1}{2}\norm{\E_k \E_{\psi_1} U_k \psi_1 U_k^\dagger \ot \psi_1^{\ot r-1 } - \E_{\psi_1} \frac{\I}{2^{m-1}} \ot \psi_1^{\ot r-1}} + \frac{80r\sqrt{m}}{2^{m/2}} \,.
    \end{align*}
To keep the notation simple in the next calculations, we write $\E_{\ket{\psi_1}, \ket{\psi_2}}$ as short for $\E_{\ket{\psi_1}, \ket{\psi_2} \gets \mu_{2^{m-1}}}$, and we denote $\ket{\psi'} = \frac{1}{\sqrt{2}}\ket{0}\ket{\psi_1} + \frac{1}{\sqrt{2}}\ket{1}\ket{\psi_2}$. For $U \in U(2^{m-1})$, define the controlled-$U$ gate $CU = \ket{0}\bra{0} \ot \I + \ket{1}\bra{1} \ot U$. For convenience, we denote $\ket{\psi_{CU}} = CU \ket{+} \ket{\psi_1}$, and we write $\E_{U}$ as short for $E_{U \gets SU(2^{m-1})}$ respectively. We have
\ifnum\widemargin=0
\begin{align*}
    &\norm{\E_{k} \E_{\psi} \bra{0}_1 (\I\ot U_k) \psi (\I\ot U_k^\dagger) \ket{0}_1 \ot \psi^{\ot r-1} - \E_{\psi} \bra{0}_1  \frac{\I}{2^{m}} \ket{0}_1 \ot \psi^{\ot r-1} } \leq  \\  &\Bigg \| \E_{k} \E_{\psi} \bra{0}_1 (\I\ot U_k) \psi (\I\ot U_k^\dagger) \ket{0}_1 \ot \psi^{\ot r-1} - \E_{k} \E_{\ket{\psi_1}, \ket{\psi_2}} \bra{0}_1 (\I\ot U_k) \psi' (\I\ot U_k^\dagger) \ket{0}_1 \ot \psi'^{\ot r-1} \Bigg \|  \\ &+ \Bigg\| \E_{k} \E_{\ket{\psi_1}, \ket{\psi_2} } \bra{0}_1 (\I\ot U_k) \psi' (\I\ot U_k^\dagger) \ket{0}_1 \ot \psi'^{\ot r-1} -\E_{\ket{\psi_1}, \ket{\psi_2}} \bra{0}_1  \frac{\I}{2^{m}} \ket{0}_1 \ot \psi_{CU}^{\ot r-1} \Bigg\| \leq  \\
    & \frac{80 r\sqrt{m}}{2^{m/2}} + \Bigg\| \E_{k} \E_{\ket{\psi_1}, \ket{\psi_2}} \bra{0}_1 (\I\ot U_k) \psi' (\I\ot U_k^\dagger) \ket{0}_1 \ot \psi'^{\ot r-1} -\E_{\ket{\psi_1}, \ket{\psi_2} } \bra{0}_1  \frac{\I}{2^{m}} \ket{0}_1 \ot \psi'^{\ot r-1} \Bigg\|
\end{align*}
\else

\begin{align*}
    &\norm{\E_{k} \E_{\psi} \bra{0}_1 (\I\ot U_k) \psi (\I\ot U_k^\dagger) \ket{0}_1 \ot \psi^{\ot r-1} - \E_{\psi} \bra{0}_1  \frac{\I}{2^{m}} \ket{0}_1 \ot \psi^{\ot r-1} } \leq  \\  
    &\Bigg \| \E_{k} \E_{\psi} \bra{0}_1 (\I\ot U_k) \psi (\I\ot U_k^\dagger) \ket{0}_1 \ot \psi^{\ot r-1} - \\&\qquad\qquad\qquad\E_{k} \E_{\ket{\psi_1}, \ket{\psi_2}}  \bra{0}_1 (\I\ot U_k) \psi' (\I\ot U_k^\dagger) \ket{0}_1 \ot \psi'^{\ot r-1} \Bigg \|  \\ 
    &+ \Bigg\| \E_{k} \E_{\ket{\psi_1}, \ket{\psi_2} } \bra{0}_1 (\I\ot U_k) \psi' (\I\ot U_k^\dagger) \ket{0}_1 \ot \psi'^{\ot r-1} \\&\qquad\qquad\qquad-\E_{\ket{\psi_1}, \ket{\psi_2}} \bra{0}_1  \frac{\I}{2^{m}} \ket{0}_1 \ot \psi_{CU}^{\ot r-1} \Bigg\| \leq  \\
    & \frac{80 r\sqrt{m}}{2^{m/2}} + \Bigg\| \E_{k} \E_{\ket{\psi_1}, \ket{\psi_2}} \bra{0}_1 (\I\ot U_k) \psi' (\I\ot U_k^\dagger) \ket{0}_1 \ot \psi'^{\ot r-1} \\&\qquad\qquad\qquad-\E_{\ket{\psi_1}, \ket{\psi_2} } \bra{0}_1  \frac{\I}{2^{m}} \ket{0}_1 \ot \psi'^{\ot r-1} \Bigg\|
\end{align*}
\fi
where the first inequality is by a triangle inequality, and the second uses \cref{lem:haar-coeff-estimate} combined with the fact that the trace norm is decreasing under taking projections. Thus, it suffices for us to show that
\begin{align}
        &\Bigg\| \E_{k} \E_{\ket{\psi_1}, \ket{\psi_2}} \bra{0}_1 (\I\ot U_k) \psi' (\I\ot U_k^\dagger) \ket{0}_1 \ot \psi'^{\ot r-1} -\E_{\ket{\psi_1}, \ket{\psi_2} } \bra{0}_1  \frac{\I}{2^{m}} \ket{0}_1 \ot \psi'^{\ot r-1} \Bigg\| \nonumber \\ &\leq \frac{1}{2}\norm{\E_k \E_{\psi_1} U_k \psi_1 U_k^\dagger \ot \psi_1^{\ot r-1 } - \E_{\psi_1} \frac{\I}{2^{m-1}} \ot \psi_1^{\ot r-1}} \,, \label{eq:266}
\end{align}

Now, notice that the distribution of states $\ket{\psi'} =\frac{1}{\sqrt{2}}(\ket{0}\ket{\psi_1} + \ket{1}\ket{\psi_2})$, where $\ket{\psi_1}, \ket{\psi_2} \gets \mu_{2^{m-1}}$, is identical to the distribution of states $\ket{\psi'} = CU \ket{+} \ket{\psi_1} = \frac{1}{\sqrt{2}}(\ket{0}\ket{\psi_1} + \ket{1}U\ket{\psi_1}) $, where $\ket{\psi_1} \gets \mu_{2^{m-1}}$ and $U \gets SU(2^{m-1})$ (this equivalence implicitly uses the unitary invariance of the Haar measure). 
Thus, \cref{eq:266} is equivalent to 
\begin{align*}
    \label{eq:26}
        &\norm{\E_{k} \E_{\psi_1, U} \bra{0}_1 (\I\ot U_k) \psi_{CU} (\I\ot U_k^\dagger) \ket{0}_1 \ot \psi_{CU}^{\ot r-1} - \E_{\psi_1, U} \bra{0}_1  \frac{\I}{2^{m}} \ket{0}_1 \ot \psi_{CU}^{\ot r-1} } \\&\leq \frac{1}{2}\norm{\E_k \E_{\psi_1} U_k \psi_1 U_k^\dagger \ot \psi_1^{\ot r-1 } - \E_{\psi_1} \frac{\I}{2^{m-1}} \ot \psi_1^{\ot r-1}} \,,
\end{align*}
So, we are left with showing that the latter inequality is true, which is equivalent to
\ifnum\widemargin=0
    \begin{equation*}
        \norm{\E_k \E_{\psi_1, U} U_k\psi_1U_k^\dagger \otimes \psi_{CU}^{\ot r-1} -  \E_{\psi_1, U} \frac{\I}{2^{m-1}} \ot \psi_{CU}^{\ot r-1}} \leq \norm{\E_k \E_{\psi_1} U_k \psi_1 U_k^\dagger \ot \psi_1^{\ot r-1 } - \E_{\psi_1} \frac{\I}{2^{m-1}} \ot \psi_1^{\ot r-1}}
    \end{equation*}
\else
    \begin{align*}
        &\norm{\E_k \E_{\psi_1, U} U_k\psi_1U_k^\dagger \otimes \psi_{CU}^{\ot r-1} -  \E_{\psi_1, U} \frac{\I}{2^{m-1}} \ot \psi_{CU}^{\ot r-1}} \\&\quad\leq \norm{\E_k \E_{\psi_1} U_k \psi_1 U_k^\dagger \ot \psi_1^{\ot r-1 } - \E_{\psi_1} \frac{\I}{2^{m-1}} \ot \psi_1^{\ot r-1}}
    \end{align*}
\fi
    Let us denote $\ket{\tilde{\psi_1}} =\ket{+}\ket{\psi_1}$. Notice that
    \begin{equation*}
    \begin{split}
         &\norm{\E_k \E_{\psi_1, U} U_k\psi_1U_k^\dagger \otimes \psi_{CU}^{\ot r-1} - \E_{\psi_1, U} \frac{\I}{2^{m-1}} \ot \psi_{CU}^{\ot r-1}}\\& = 
          \norm{\E_U \left(\E_k\E_{\psi_1} U_k \psi_1 U_k^\dagger \otimes (CU \tilde{\psi_1} CU^\dagger)^{\ot r-1} - \E_{\psi_1} \frac{\I}{2^{m-1}} \ot (CU \tilde{\psi_1}CU^\dagger)^{\ot r-1}\right)} \\
         &\leq \E_U\norm{ \E_k\E_{\psi_1} U_k \psi_1 U_k^\dagger \otimes  (CU \tilde{\psi_1} CU^\dagger)^{\ot r-1} - \E_{\psi_1} \frac{\I}{2^{m-1}} \ot (CU \tilde{\psi_1}CU^\dagger)^{\ot r-1}} \\
         &= \E_U\norm{ \E_k\E_{\psi_1} U_k \psi_1 U_k^\dagger \otimes \tilde{\psi_1}^{\ot r-1} - \E_{\psi_1} \frac{\I}{2^{m-1}} \ot \tilde{\psi_1}^{\ot r-1}} \\
         &= \norm{\E_k \E_{\psi_1} U_k \psi_1 U_k^\dagger \ot \psi_1^{\ot r-1 } - \E_{\psi_1} \frac{\I}{2^{m-1}} \ot \psi_1^{\ot r-1}} \,.
    \end{split}
    \end{equation*}
    This concludes the proof of \cref{thm:ammplify-random}.
\end{proof}

We now have all the ingredients to show that the $\oprs$ construction from \cref{fig:full-1prs-protocol} is secure. 

\begin{corollary}
\label{cor:49}
Let $n,m \in \mathbb{N}$, where $\prsm \leq n <m$, and $n$ is even. Let $\{U_k\}_{k \in \{0,1\}^{n}}$ be the family of $m$-qubit unitaries from \cref{fig:full-1prs-protocol}, i.e.\ $U_k  = X^a Z^b \otimes \mathds{1}^{\otimes (m-n/2)}$, where $a,b \in \{0,1\}^{\frac{n}{2}}$ are the first and second halves of $k$. Then, for any $r < 2^{\frac{m}{2}}$,
\ifnum\widemargin=0
\begin{equation}
\label{eq:cor49}
\norm{\E_{k \gets \{0,1\}^n} \E_{\ket{\psi} \gets \mu_{2^m}} (\I\ot U_k) \psi (\I\ot U_k^\dagger) \ot \psi^{\ot r-1} - \E_{\ket{\psi} \gets \mu_{2^m}} \frac{\I}{2^m} \ot \psi^{\ot r-1} } \leq \frac{(2r^2+800rm\sqrt{m})5^{0.1m}}{2^{0.45m}}  \,. 
\end{equation}
\else
\begin{equation}
\label{eq:cor49}
\begin{split}
&\norm{\E_{k \gets \{0,1\}^n} \E_{\ket{\psi} \gets \mu_{2^m}} (\I\ot U_k) \psi (\I\ot U_k^\dagger) \ot \psi^{\ot r-1} - \E_{\ket{\psi} \gets \mu_{2^m}} \frac{\I}{2^m} \ot \psi^{\ot r-1} } \\&\quad\leq \frac{(2r^2+800rm\sqrt{m})5^{0.1m}}{2^{0.45m}}  \,. 
\end{split}
\end{equation}
\fi
\end{corollary}
\begin{proof}
Let $\ell = m-n$. Recursively apply \cref{thm:ammplify-random} $\ell$ times, using \cref{thm:partial-one-time-pad} to bound the RHS of \cref{eq:thm48} the first time that \cref{thm:ammplify-random} is applied.
\end{proof}


\begin{corollary}\label{cor:cor410}
The construction from \cref{fig:full-1prs-protocol} is a $\oprs$ in the CHRS model (as in \cref{def:1prs_CHRS}).
\end{corollary}
\begin{proof}
Take $n = 0.9m$. Then, for any $r = \poly(m)$, and for all large enough $m$, the RHS of \cref{eq:cor49} is less than $0.86^m$ (since ${5^{0.1m}}/{2^{0.45m}} = (0.85987\dots)^m$). Note that, in \cref{cor:49}, the adversary only gets access to $r$ copies of a \emph{single} $m$-qubit Haar random state $\ket{\psi}$, whereas in the definition of a $\oprs$ in the CHRS model (\cref{def:1prs_CHRS}), the adversary has also access to the other states from $\mathcal{S}$. However, as pointed out earlier, since our construction only uses the $m$-qubit state (for output of length $m$), and all of the states in $\mathcal{S}$ are independently sampled, the security property of \cref{def:1prs_CHRS} is equivalent to that of \cref{eq:1prs-pseudorandomness}.
\end{proof}

Note that in our definition of $\oprs$ in the CHRS model (\cref{def:1prs_CHRS}), the security guarantee is ``on average over $\mathcal{S}$''. However, for the purpose of utilizing this result in the context of an oracle separation (as we will do in \ifproceedingelse{the full version~\cite{arxivfullversion}}{\cref{sec:oralce-sep}}), it is important that we can find a \emph{fixed} family of states $\mathcal{S}$ relative to which $\oprs$ exist. We show that this is the case: with probability $1$ over $\mathcal{S}$, the $\oprs$ security holds (against \emph{all} adversaries). See \ifproceedingelse{the full version~\cite{arxivfullversion} }{Section \ifnum\shortver=0 \ref{ssec:ampli-random}
\else\ref{sec:borel-cantelli}\fi}for more details.

\if\shortver=0
\if\shortver=0
\else
\section{Reduction from a distribution over oracles to a fixed oracle}
\label{sec:borel-cantelli}
\fi
\begin{corollary}
    Let $\calS = \{\ket{\psi_m}\}_{m \in \N}$ denote the CHRS family of states. Then, with probability $1$, $\calS$ satisfies the following property: for any adversary $\calA$ with access to polynomially many copies of states in $\calS$, there exists a negligible function $\negl$, such that, for all $m$, $$\left| \P[ \calA^{\calS}(1^m, \E_k U_k\psi_m U_k^\dagger) \to 1] - \P [ \calA^{\calS}(1^m, \frac{\I}{2^m}) ] \right| = \negl(m) \,,$$
    where the $U_k$ are as defined in \cref{fig:full-1prs-protocol}, and the notation $\calA^{\mathcal{S}}$ denotes that $\calA$ has access to polynomially many copies of states from $\mathcal{S}$.
\end{corollary}
\begin{proof}
    As pointed out earlier, it suffices to consider the case where $\mathcal{A}^{\calS}(1^m, \cdot)$ only gets polynomially many copies of the \emph{single} state $\ket{\psi_{m}}$ (rather than various states in $\calS$). Any adversary can be described by a Turing machine that on input $1^m$ outputs a distinguishing quantum circuit. Denote the length of the Turing machine by $|\calA|$. 
    Then, by \cref{cor:cor410}, we know that for any adversary $\calA^{\calS}$ with access to $\poly(m)$ copies of $\ket{\psi_m}$,
    \begin{align*}
        \E_{\calS} \left( \adv (\calA^{\calS}(1^m,\cdot)) \right) < 0.86^m 
    \end{align*}
    for large enough $m$, where $\adv(\calA^\calS(1^m,\cdot))$ denotes $\calA$'s distinguishing advantage (since ${5^{0.1m}}/{2^{0.45m}} = (0.85987\dots)^m$). Thus, by an averaging argument, for any adversary $\calA$,
    \begin{equation}
        \P_{\calS} \Big[\adv ( \calA^{\calS}(1^m,\cdot) ) > 0.95^m \Big] < 0.95^m \label{eq:267}
    \end{equation}
    for large enough $m$. For an adversary $\calA$, let $E_{\calA, m}$ be the event that $\adv (\calA^\calS(1^m,\cdot))$ is greater than $0.95^m$. Then, \cref{eq:267} can be equivalently restated as: $\Pr_\calS E_{\calA, m} < 0.95^m$ holds for all but finite $m$. Thus $\sum \Pr_\calS E_{\calA, m}$ is finite. Hence, by the Borel-Cantelli lemma, with probability $1$ over randomly sampling $\calS$, the event $E_{\calA,m}$ happens only for finitely many $m$, i.e.\ $\adv (\calA^\calS(1^m,\cdot)) < 0.95^m$ holds for all large enough $m$. 
    
    For an adversary $\calA$, denote by $F_\calA$ the event that $E_{\calA, m}$ holds for infinitely many $m$. Then we can restate what we found above as $\Pr_\calS F_\calA = 0$. Now, notice that there are only countably many different adversaries $\calA$ (because $\calA$ can be described by a string of finite length). So, by a union bound, we have
    $$
    \Pr_\calS [ \exists \calA \,\textnormal{ s.t. }F_\calA \text{ happens} ] \leq \sum_\calA \Pr_\calS [F_\calA] = 0 \,.
    $$
All in all, we have established that, with probability $1$ over sampling $\mathcal{S}$, it holds that, for all adversaries $\calA$, $\adv (\calA^\calS(1^m,\cdot)) < 0.95^m$ holds for all large enough $m$.
\end{proof}

Thus, we have the following. 
\begin{corollary}\label{thm:1prs}
    With probability $1$ over sampling a family of Haar random states $\calS = \{\ket{\psi_m}\}_{m \in \mathbb{N}}$ where $\ket{\psi_m} \gets \mu_{2^m}$, the construction from \cref{fig:full-1prs-protocol} is a statistically secure $\oprs$ (relative to $\calS$).
\end{corollary}

\fi
 
\ifnum\shortver=0
    \section{Oracle separation of \texorpdfstring{$\prs$}{PRS} and \texorpdfstring{$\oprs$}{1PRS}}\label{sec:oralce-sep}
In this section, we show that there is an oracle relative to which $\oprs$ exist, but $\prs$ do not. This implies that there does not exist a (certain variant of a) \emph{fully black-box construction} of a $\prs$ from a $\oprs$ (the precise variant is stated in Corollary \ref{cor:no-black-box-construction}, and a detailed explanation of the terminology is provided in \cref{sec:black-box-reductions}). We start by describing the separating oracle. 

\paragraph{Separating oracle} 
The separating oracle, which we denote as $\mathcal{O}$, consists of two oracles $\mathcal{O}_1$ and $\mathcal{O}_2$. The first oracle $\calO_1$ is identical to the oracle of the \CHRS\ model. This is best thought of as a distribution over oracles (although we show that it is possible to fix one particular instance from the distribution). To remind the reader, $\mathcal{O}_1$ is obtained by sampling a sequence of Haar random states $\{\ket{\psi_m}\}_{m=1}^\infty$, where $\ket{\psi_m}$ is on $m$ qubits. Then, given a unary input $1^m$, $\mathcal{O}_1$ outputs the state $\ket{\psi_m}$. We emphasize that $\mathcal{O}_1$ only takes inputs of the form $1^m$ (and not superpositions of these). Thus, formally, each call to the oracle can be thought of as applying an isometry (see \cref{sec:1prs_CHRS}). Informally, 
the second oracle $\calO_2$ is a quantum oracle that provides the ability to perform any quantum operations that a $\QPSPACE$ machine can apply: it receives as input a state $\ket{\alpha}$ on $s$ qubits, a concise description of a polynomial space quantum circuit $C$ acting on these $s$ qubits, and it returns the result of $C$ acting on $\ket{\alpha}$. Formally, $\calO_2$ acts as follows: 
the input consists of a quantum state $\ket{\alpha}$ on some number $s$ of qubits, a classical Turing Machine $M$, and 
a number $t$.
The oracle runs the classical Turing machine $M$ for $t$ steps. The output of the Turing machine should represent a quantum circuit $C$ that acts on exactly $s$ qubits. Note that since the Turing machine runs only for $t$ steps, clearly, the quantum circuit has at most $t$ gates. If the quantum circuit that was printed does not use exactly $s$ qubits, or if the Turing Machine does not terminate after $t$ steps, the oracle aborts (and outputs the $\bot$ symbol). Otherwise, the oracle applies the circuit $C$ on $\ket{\alpha}$, and returns the output.
 
We show the following.
\begin{theorem}
\label{thm:oracle-sep}
With respect to $\mathcal{O} = (\mathcal{O}_1, \mathcal{O}_2)$, $\oprs$ exist, but $\prs$ (with output length at least $\log n + 10$, where $n$ is the seed length) do not.
\end{theorem}

The existence of $\oprs$ relative to $\mathcal{O} = (\mathcal{O}_1, \mathcal{O}_2)$ follows immediately from Corollary \ref{thm:existence-1prs}: the construction of the $\oprs$ is the same as in \cref{fig:full-1prs-protocol}, and Corollary \ref{thm:existence-1prs} says that the construction is \emph{statistically} secure against adversaries with polynomially many queries to $\mathcal{O}_1$. Since the $\QPSPACE$ machine is independent of the sampled Haar random state, it can be simulated by a computationally unbounded adversary. Note that, as argued in Corollary \ref{thm:existence-1prs}, the construction is a secure $\oprs$ \emph{with probability $1$} over sampling $\mathcal{O}_1$, i.e.\ over sampling the family of Haar random states.  

Thus the crux of this section is dedicated to showing that $\prs$ do not exist relative to the oracle. We show this by describing a concrete attack on any $\prs$ scheme, relative to $\mathcal{O}$. The attack breaks any PRS, \emph{with probability $1$} over sampling $\calO_1$.

In \cref{sec:or-lemma}, we review the ``quantum OR lemma'', which is a key ingredient in our attack. In \cref{sec:oracle-sep}, we describe our attack, and in \cref{sec:black-box-reductions}, we provide a detailed discussion of the relation between black-box constructions and oracle separations in the quantum setting.

\subsection{Quantum OR lemma}\label{sec:or-lemma}
Informally, the ``quantum OR lemma'' says that there exists a quantum algorithm that takes as input a family of projectors, as well as a \emph{single copy} of a quantum state $\rho$, and decides whether either:
\begin{itemize}
\item $\rho$ has a significant overlap with one of the projectors, or
\item $\rho$ has small overlap with all of the projectors.
\end{itemize}
The space complexity of this quantum algorithm is especially important for us.

\begin{lemma}[{Quantum OR lemma, adapted from~\cite[Corollary 3.1]{harrow2017sequential}}]\label{lem:or-lemma}
Let $\Lambda_1, \dots, \Lambda_N$ be projectors, and fix real positive numbers $\eps \leq \frac{1}{2}$, and $\delta$. Let $\rho$ be a state such that either there exists $i \in [N]$ such that $\tr[\Lambda_i \rho] \geq 1-\eps$ (case 1) or, for all $i \in [N]$, $\tr[\Lambda_i \rho] \leq \delta$ (case 2).

Then, there is a quantum circuit, $C_{OR}$, which we refer to as the ``OR tester'', such that measuring the first qubit in case 1 yields:
\
\begin{equation*}
\Pr\left(C_{OR}(\rho) \to 1\right) \geq \frac{(1-\eps)^2}{7}
\end{equation*}
and in case 2:
\begin{equation*}
\Pr\left(C_{OR}(\rho) \to 1\right) \leq 4 N\delta.
\end{equation*}
\end{lemma} 
\begin{remark}
\label{remark:53}
Even when the number of measurements, $N$, is exponential in the number of qubits of $\rho$, denoted $n$, the circuit $C_{OR}$ which is constructed in Ref.~\cite{harrow2017sequential} can be implemented by a unitary $\QPSPACE$ machine\footnote{i.e., the family of unitary circuits $C_{OR}$, indexed by $n$, is a uniform family of quantum unitary circuits using $\poly(n)$ qubits of space.} as long as each $\Lambda_i$ can be implemented by a $\QPSPACE$ machine, and the set of measurements has a concise polynomial description. \ifnum\shortver=0 We justify this claim in \cref{sec:OR_leamma_unitary}.\fi
\label{rem:or_lemma_qpspace}
\end{remark}

\subsection{An attack on any \texorpdfstring{$\prs$}{PRS} relative to the separating oracle}\label{sec:oracle-sep}
We describe an attack, based on the quantum OR lemma, that breaks \emph{any} $\prs$ relative to the oracle $\calO$ described at the beginning of the section. Before describing our attack, we first introduce some technical tools. First, we need the following concentration bound.

\begin{lemma}\label{lem:fidelity-concentration}
    Let $N \in \mathbb{N}$, and $\ket{\psi_0}$ a $N$-dimensional state. Then, 
    \begin{equation*}
        \Pr_{\ket{\psi} \gets \mu_N} \bracks{\abs{\braket{\psi|\psi_0}}^2 \geq \frac{1}{2}} < 8\exp\left(\frac{-N}{600}\right)
    \end{equation*}
\end{lemma}
\begin{proof}
Let $\calS(N)$ be the unit $N$-dimensional sphere, i.e.\ the set of all $N$-dimensional pure states. Define functions $f_1, f_2: \calS(N) \to \R$ such that $f_1(\ket{\psi}) = \text{Re} \braket{\psi_0|\psi}$, and $f_2(\ket{\psi}) = \text{Im} \braket{\psi_0|\psi}$.

$f_1$ and $f_2$ are 1-Lipschitz functions. In fact, for any $N$-dimensional states $\ket{\psi_1}$ and $\ket{\psi_2}$
\ifnum\widemargin=0
\begin{equation*}
\abs{f_1(\ket{\psi_1}) - f_1(\ket{\psi_2})} = \abs{\text{Re}\left( \bra{0} (\ket{\psi_1} - \ket{\psi_2})\right)} \leq \abs{\bra{0} (\ket{\psi_1} - \ket{\psi_2})} \leq \norm{\ket{\psi_1} - \ket{\psi_2}} \,.
\end{equation*}
\else
\begin{align*}
\abs{f_1(\ket{\psi_1}) - f_1(\ket{\psi_2})} &= \abs{\text{Re}\left( \bra{0} (\ket{\psi_1} - \ket{\psi_2})\right)} \\ &\leq \abs{\bra{0} (\ket{\psi_1} - \ket{\psi_2})} \leq \norm{\ket{\psi_1} - \ket{\psi_2}} \,.
\end{align*}
\fi
Similarly for $f_2$. Now, notice that, for any $\ket{\psi}$, we have $f_1(\ket{\psi}) = - f_1(-\ket{\psi})$, and $f_2(\ket{\psi}) = - f_2(-\ket{\psi})$. This implies that $\E_{\ket{\psi}} f_1(\ket{\psi}) = \E_{\ket{\psi}} f_2(\ket{\psi}) = 0$. Hence, we can invoke Levy's lemma (\cref{lem:levy-lemma}) to deduce that 
\begin{equation*}
\Pr_{\ket{\psi}} \bracks{\abs{f_1(\ket{\psi})} \geq \frac{1}{2} } \leq 4\exp\parens{-\frac{N}{18\pi^3}} < 4\exp\parens{-\frac{N}{600}}.
\end{equation*}
A similar concentration bound holds for $f_2$. Note that $\abs{\braket{\psi|\psi_0}}^2 = f_1(\ket{\psi})^2 + f_2(\ket{\psi})^2$, and hence, by a union bound,
\begin{align*}
\Pr_{\ket{\psi}} \bracks{ \abs{\braket{\psi|\psi_0}}^2 \geq \frac{1}{2} } &= \Pr_{\ket{\psi}} \bracks{f_1^2+f_2^2 \geq \frac{1}{2}}\\
&\leq \Pr_{\ket{\psi}} \bracks{\abs{f_1(\ket{\psi})} \geq 1/2} + \Pr_{\ket{\psi}} \bracks{\abs{f_2(\ket{\psi})} \geq 1/2} \\
&< 8 \exp\parens{\frac{-N}{600}}. \qedhere
\end{align*}
\end{proof}

Now, we are ready to describe our attack, and complete the proof of \cref{thm:oracle-sep}.

\begin{proof}[Proof of \cref{thm:oracle-sep}]
Consider a $\prs$ relative to $\calO$. This consists of a generation procedure $\gen^{\calO}$ that takes as input a seed $k$, and outputs a state $\ket{\phi_k}$. We denote by $n$ the length of $k$, and by $m$ the number of qubits of $\ket{\phi_k}$. Recall that $\gen^{\calO} = (\mathcal O_1,\mathcal O_2)$, where $\mathcal O_1$ is an oracle that provides states from a family of Haar random states $\{\ket{\psi_m}\}$, and $\mathcal O_2$ is the $\QPSPACE$ machine oracle (see the start of \cref{sec:oralce-sep} for a precise definition).

Similarly as in \cref{def:1prs_CHRS}, without loss of generality, we can take the generation procedure to be of the following form: there is a polynomial $s = s(n)$ and a family $\{\gen^{\calO_2}_k\}_{k \in \{0,1\}^n}$ of efficiently generatable $poly(n)$-size unitary circuits that include calls to $\calO_2$ (but not $\calO_1$) such that
$$ \ket{\phi_k} = \gen^{\calO_2}_k (\ket{\psi_1}^{\otimes s} \otimes \ket{\psi_2}^{\otimes s}\ldots \otimes \ket{\psi_s}^{\otimes s}) \,.$$
In other words, the $\prs$ generation procedure first obtains polynomially many copies of states from the family $\{\ket{\psi_m}\}$, and then, on input $k$, applies an efficiently generatable unitary that makes calls to $\calO_2$ as a black-box. Note that $\gen_k$ may discard some of the qubits, and those would be traced out and not be considered as part of the output state $\ket{\phi_k}$, and therefore the entire transformation is not necessarily unitary.

We denote by $U_k$ the unitary implemented by $\gen_k^{\mathcal O_2}$ before tracing out some of the registers.\footnote{Note that the pseudorandom state must be a pure state; therefore, we can assume without loss of generality that the $\mathcal O_2$ QPSPACE machine does nor perform any measurements.} Recall that the number of qubits in $\ket{\phi_k}$ is denoted by $m$, and we name the output register as $\mathsf{A}$, and the register containing the qubits which are traced out by $\gen_k^{\mathcal O_2}$ is denoted by $\mathsf B$. We let $\mathsf{C}$ be another $m$-qubits register. 
Consider the family of projectors
\begin{equation}
\Pi_k = \left(\Big((U_k^\dagger)_{\mathsf{AB}} \otimes \I_{\mathsf{C}} \Big) (\Pi^{sym}_{\mathsf{AC}} \otimes \I_{\mathsf{B}} )\Big((U_k)_{\mathsf{AB}} \otimes \I_{\mathsf{C}}\Big)\right)^{\otimes 10n},
\label{eq:projectors_attack}
\end{equation}
where $\Pi^{sym}_{\mathsf{AC}}$ is the projection onto the symmetric subspace across the two registers $\mathsf{A}$ and $\mathsf{C}$.

The attack is the following: the adversary queries $\mathcal O_1$ to generate $(\ket{\psi_1}^{\otimes s} \otimes \ket{\psi_2}^{\otimes s}\ldots \otimes \ket{\psi_s}^{\otimes s})^{\otimes 10 n}$ and stores each copy in the $\mathsf{AB}$ register, and receives $10n$ copies of $\ket{\phi}$, where $\ket{\phi}$ is either a pseudorandom state or a Haar random state, which is stored in the $\mathsf C$ register. We denote this combined state as $\rho$.  It then uses the $\calO_2$ oracle (the $\QPSPACE$ machine) to run the ``OR tester'' from the quantum OR lemma (\cref{lem:or-lemma}), where, using the notation from \cref{lem:or-lemma}, with $\rho$ as defined above, and $\Lambda_k = \Pi_k$ as defined in \cref{eq:projectors_attack}. Recall that the ``OR tester'' can indeed be implemented by a $\QPSPACE$ machine, as discussed in \cref{rem:or_lemma_qpspace}. 

We now argue that the ``OR tester'' successfully distinguishes between pseudorandom and random $\ket{\phi}$.

\begin{itemize}
\item Suppose $\ket{\phi} = \ket{\phi_k}$ for some $k$. It is clear that the state $$\Big((\ket{\psi_1}^{\otimes s} \otimes \ket{\psi_2}^{\otimes s}\otimes \ldots \otimes \ket{\psi_s}^{\otimes s})_{\mathsf{AB} } \otimes \ket{\phi_k}_{\mathsf{C}}\Big)^{\otimes 10n}$$ lies in the range of $\Pi_k = \left(\Big((U_k^\dagger)_{\mathsf{AB}} \otimes \I_{\mathsf{C}} \Big) (\Pi^{sym}_{\mathsf{AC}} \otimes \I_{\mathsf{B}} )\Big((U_k)_{\mathsf{AB}} \otimes \I_{\mathsf{C}}\Big)\right)^{\otimes 10n}$. Thus, we are in ``case 1'' of \cref{lem:or-lemma} with $\eps=0$. Hence, the probability that the ``OR tester'' outputs $1$ is at least $1/7$. 
\item Suppose $\ket{\phi}$ is Haar random. Then, by Lemma \ref{lem:fidelity-concentration}, we have that, with probability at least \\ $1-8\exp(-\frac{2^{m}}{600})$, $$\abs{\braket{\phi|\phi_k}} \leq \frac{1}{\sqrt{2}}\,.$$
Notice that the probability that $\ket{\phi}\otimes \ket{\phi_k}$ passes the ``swap test'' (i.e.\ it is found to lie in the symmetric subspace across the two registers when the measurement $\{\Pi_{sym}, I - \Pi_{sym}\}$ is performed) is exactly $\frac{1}{2} + \frac{1}{2} \abs{\braket{\phi|\phi_k}}^2$ (cf. \cite{buhrman2001quantum}). Since $\Pi_k$ corresponds to a projection onto $10n$ such swap tests \emph{all} accepting, we have that, with probability at least $1-8\exp(-\frac{2^{m}}{600})$ over the sampling of $\ket{\phi}$, 
$$\tr[\Pi_k \rho] \leq \parens{\frac34}^{10n}\,.$$

Now, by a union bound over $k \in \{0,1\}^n$, we have that, except with probability at most $8\cdot 2^n \cdot \exp(-\frac{2^m}{600})$ over the sampling of $\ket{\phi}$, the inequality $\tr[\Pi_k \phi'] \leq \parens{\frac{3}{4}}^{10n}$ holds for all $k$, and we are in ``case 2'' of \cref{lem:or-lemma} with $\delta = \parens{\frac{3}{4}}^{10n}$. Hence, in this case, the ``OR tester'' outputs $1$ with probability at most $4 \cdot 2^n \cdot \left(\frac34\right)^{10n}$.
All in all, by a final union bound, the ``OR tester'' outputs $1$ with probability at most $8\cdot 2^n \cdot \exp(-\frac{2^m}{600}) + 4 \cdot 2^n \cdot \left(\frac34\right)^{10n}$, which is exponentially small in $n$ when $m > \log n + \log 600$ (note that here the base of $\exp$ is $e$, and the base of $\log$ is $2$). Notice that our attack breaks the $\prs$ regardless of what family of the reference states $\{\ket{\psi_m}\}_{m=1}^{\infty}$ is. Thus, the attack works not only with ``probability $1$'' over such families, but, in fact, \emph{for all} possible families $\{\ket{\psi_m}\}_{m=1}^\infty$. \qedhere
\end{itemize}
\end{proof}

\begin{remark}
The proof of Theorem~\ref{thm:oracle-sep} also shows that the $\oprs$ family generated in \cref{fig:full-1prs-protocol} is not statistically secure when we allow multiple-copy access to the generated state, i.e.\ the family in \cref{fig:full-1prs-protocol} is a $\oprs$ against query-bounded adversaries but not a $\prs$ against such adversaries.
\end{remark}

\begin{remark}
    The QPSPACE machine is quite a powerful oracle, and one might wonder whether a different attack based on shadow tomography would work here (along the lines of the attack described by Kretschmer in \cite[Subsection 1.3]{Kre21}). This would only require a $\mathsf{PP}$ oracle to carry out the classical post-processing. As pointed out earlier though, the issue is that here the projectors $\Pi_k = \left(\Big(\I_{\mathsf{A}} \otimes (U_k)_{\mathsf{A'B'}}^\dagger \Big) (\Pi^{sym}_{\mathsf{AA'}} \otimes \I_{\mathsf{B'}} )\Big(\I_{\mathsf{A}} \otimes (U_k)_{\mathsf{A'B'}}\Big)\right)^{\otimes 10n}$ have large 2-norm: $\tr \Pi_k^2$ is exponential in $n$~\cite{Har13}, and so the estimation of the quantity $\tr[\Lambda_k \tilde{\phi}]$ given by shadow tomography has too large of a variance. Thus, shadow tomography does not seem to be sample-efficient in our setting.
\end{remark}

\begin{remark}
    Our attack against $\prs$ is not relativizing: if a $\prs$ family is constructed relative to an oracle $\calO$, then our attack based on the OR lemma needs exponentially many queries to $\calO$, thus it cannot be simulated by a $\BQP$ adversary with access to a $\QPSPACE$ machine. Therefore, it does not violate the oracle construction of $\prs$ by Kretschmer~\cite{Kre21}, nor a conjecture by Kretschmer et al.~\cite[Sections 7.1--7.2 ]{KQST23} about the existence of $\prs$ relative to a classical oracle.
\end{remark}

A detailed discussion of the relation between black-box constructions and oracle separations in the quantum setting is postponed to \cref{sec:black-box-reductions}. Combining Theorem~\ref{thm:oracle-sep} with Theorem~\ref{thm:quantumIR} from \cref{sec:black-box-reductions} (and using the terminology introduced there), we immediately have:
\begin{corollary}
\label{cor:no-black-box-construction}
There is no fully black-box construction of a $\prs$ from isometry access to a $\oprs$ (as in Definition~\ref{def:bb-isometry}).
\end{corollary}

\subsection{Clarifying the relationship between quantum oracle separations and black-box constructions}
\label{sec:black-box-reductions}

In this section, we clarify what we mean by a ``black-box construction'' of primitive $\mathcal{Q}$ from primitive $\mathcal{P}$ when the primitives involve \emph{quantum} algorithms (and possibly quantum state outputs). We also clarify the relationship between a \emph{quantum} oracle separation of $\mathcal{P}$ and $\mathcal{Q}$ and the (im)possibility of a black-box construction of one from the other. To the best of our knowledge, while black-box separations in the quantum setting have been the topic of several recent works, a somewhat formal treatment of the terminology and basic framework is missing. This section is a slightly extended version of a section that appears almost verbatim in the concurrent work \cite{coladangelo2024blackbox}.

In the quantum setting, it is not immediately obvious what the correct notion of ``black-box access'' is. There are a few reasonable notions of what it means for a construction to have ``black-box access'' to another primitive. We focus on three variants: \emph{unitary} access, \emph{isometry} access, and access to \emph{both the unitary and its inverse}.

The summary is that, similarly to the classical setting, a \emph{quantum} oracle separation of primitives $\mathcal{P}$ and $\mathcal{Q}$ (i.e.\ a quantum oracle relative to which $\mathcal{P}$ exists but $\mathcal{Q}$ does not) implies the impossibility of a black-box construction of $\mathcal{Q}$ from $\mathcal{P}$, but with one caveat: the type of oracle separation corresponds directly to the type of black-box construction that is being ruled out. For example, if one wishes to rule out black-box constructions of $\mathcal{Q}$ that are allowed to make use of the inverse of unitary implementations of $\mathcal{P}$, then the oracle separation needs to be ``closed under giving access to the inverse of the oracle'', i.e.\ the separation needs to hold relative to an oracle \emph{and} its inverse.

We start by introducing some terminology.

\paragraph{Terminology.} 
A quantum channel is a CPTP (completely-positive-trace-preserving) map. The set of quantum channels captures all admissible ``physical'' processes in quantum information, and it can be thought of as the quantum analogue of the set of functions  $f: \{0,1\}^* \rightarrow \{0,1\}^*$. 

For the purpose of this section, a quantum channel is specified by a family of unitaries $\{U_n\}_{n \in \mathbb{N}}$ (where $U_n$ acts on an input register of size $n$, and a work register of some size $s(n)$). The quantum channel maps an input (mixed) state $\rho$ on $n$ qubits to the (mixed) state obtained as follows: apply $U_n (\cdot) U_n^{\dagger}$ to $\rho \otimes (\ket{0}\bra{0})^{\otimes s(n)}$; measure a subset of the qubits; output a subset of the qubits (measured or unmeasured). We say that the family $\{U_n\}_{n \in \mathbb{N}}$ is a \emph{unitary implementation} of the quantum channel. We say that the quantum channel is QPT if it possesses a unitary implementation $\{U_n\}_{n \in \mathbb{N}}$ that is additionally a uniform family of efficiently computable unitaries. In other words, the quantum channel is implemented by a QPT algorithm.

One can also consider the family of isometries $\{V_n\}_{n \in \mathbb{N}}$ where $V_n$ takes as input $n$ qubits, and acts like $U_n$, but with the work register fixed to $\ket{0}^{s(n)}$, i.e.\
$V_n: \ket{\psi} \mapsto U_n(\ket{\psi}\ket{0}^{\otimes s(n)})$. We refer to $\{V_n\}_{n \in \mathbb{N}}$ as the \emph{isometry implementation} of the quantum channel. 

We will also consider QPT algorithms with access to some oracle $\mathcal{O}$. In this case, the unitary (resp. isometry) implementation $\{U_n\}_{n \in \mathbb{N}}$ should be \emph{efficiently computable given access to $\mathcal{O}$}.

Before diving into formal definitions, a bit informally, a \emph{primitive} $\mathcal{P}$ can be thought of as a set of conditions on tuples of algorithms $(G_1, \ldots, G_k)$. 
For example, for a digital signature scheme, a valid tuple of algorithms is a tuple $(\textit{Gen}, \textit{Sign}, \textit{Verify})$ that satisfies ``correctness'' (honestly generated signatures are accepted by the verification procedure with overwhelming probability) and ``security'' (formalized via an unforgeability game). Equivalently, one can think of the tuple of algorithms $(G_1, \ldots, G_k)$ as a \emph{single} algorithm $G$ (with an additional control input). 

A thorough treatment of black-box constructions and reductions in the classical setting can be found in \cite{RTV04}. Our definitions are a quantum analog of those found there. They follow the style of~\cite{RTV04} whenever possible and depart from it whenever necessary.

\begin{definition}
A \emph{primitive} $\mathcal{P}$ is a pair $\mathcal{P} = (\mathcal{F}_{\mathcal{P}}, \mathcal{R}_{\mathcal{P}})$\footnote{Here $\mathcal{F}_{\mathcal{P}}$ should be thought of as capturing the ``correctness'' property of the primitive, while $\mathcal{R}_{\mathcal{P}}$ captures ``security''.} where $\mathcal{F}_{\mathcal{P}}$ is a set of quantum channels, and $\mathcal{R}_{\mathcal{P}}$ is a relation over pairs $(G, A)$ of quantum channels, where $G \in \mathcal{F}_{\mathcal{P}}$.

A quantum channel $G$ is an \emph{implementation} of $\mathcal{P}$ if $G \in \mathcal{F}_{\mathcal{P}}$. If $G$ is additionally a QPT channel, then we say that $G$ is an \emph{efficient implementation} of $\mathcal{P}$ (in this case, we refer to $G$ interchangeably as a QPT channel or a QPT algorithm). 

A quantum channel $A$ (usually referred to as the ``adversary'') $\mathcal{P}$-breaks $G \in \mathcal{F}_{\mathcal{P}}$ if $(G, A) \in \mathcal{R}_{\mathcal{P}}$. We say that $G$ is a \emph{secure implementation} of $\mathcal{P}$ if $G$ is an implementation of $\mathcal{P}$ such that no QPT channel $\mathcal{P}$-breaks it. The primitive $\mathcal{P}$ \emph{exists} if there exists an efficient and secure implementation of $\mathcal{P}$.

Let $U$ be a unitary (resp.\ isometry) implementation of $G \in \mathcal{P}$. Then, we say that $U$ is a \emph{unitary (resp. isometry) implementation} of $\mathcal{P}$. For ease of exposition, we also say that quantum channel $A$ $\mathcal{P}$-breaks $U$ to mean that $A$ $\mathcal{P}$-breaks $G$.

\end{definition}
Since we will discuss oracle separations, we give corresponding definitions \emph{relative to an oracle}. Going forward, for ease of exposition, we often identify a quantum channel with the algorithm that implements it.
\begin{definition}[Implementations relative to an oracle]
\label{def:oracle-implementation}
Let $\mathcal{O}$ be a unitary (resp. isometry) oracle. An \emph{implementation} of primitive $\mathcal{P}$ relative to $\mathcal{O}$ is an oracle algorithm $G^{(\cdot)}$ such that $G^{\mathcal{O}} \in \mathcal{F}_\mathcal{P}$\footnote{We clarify that here $G^{\mathcal{O}}$ is only allowed to query the unitary $\mathcal{O}$, not its inverse. However, as will be the case later in the section, $\mathcal{O}$ itself could be of the form $\mathcal{O} = (W, W^{-1})$ for some unitary $W$.}. We say the implementation is efficient if $G^{(\cdot)}$ is a QPT oracle algorithm.

Let $U$ be a unitary (resp.\ isometry) implementation of $G^{\mathcal{O}}$. Then, we say that $U$ is a \emph{unitary (resp.\ isometry) implementation} of $\mathcal{P}$ \emph{relative to $\mathcal{O}$}.
\end{definition}

\begin{definition}
\label{def:exist-relative-to-oracle}
    We say that a primitive $\mathcal{P}$ exists relative to an oracle $\mathcal{O}$ if: 
\begin{itemize}
	\item[(i)] There exists an efficient implementation $G^{(\cdot)}$ of $\mathcal{P}$ relative to $\mathcal{O}$, i.e.\ $G^{\mathcal{O}} \in \mathcal{P}$ (as in Definition~\ref{def:oracle-implementation}).
	\item[(ii)] The security of $G^{\mathcal{O}}$ holds against all QPT adversaries that have access to $\mathcal{O}$. More precisely, for all QPT $A^{(\cdot)}$, $(G^{\mathcal O},A^{\mathcal O})\notin \mathcal{R}_{\mathcal{P}}$.
\end{itemize} 
\end{definition}

\vspace{3mm}
There are various notions of black-box constructions and reductions (see, for example, \cite{RTV04}). Here, we focus on (the quantum analog of) the notion of a \emph{fully black-box construction}. We identify and define three analogs based on the type of black-box access available to the construction and the security reduction.

\begin{definition}
\label{def:bb-isometry}
	A QPT algorithm $G^{(\cdot)}$ is a \emph{fully black-box construction of $\mathcal{Q}$ from \textbf{isometry access} to $\mathcal{P}$} if the following two conditions hold:
	\begin{enumerate}
	\item (\emph{black-box construction with isometry access}) For every isometry implementation $V$ of $\mathcal{P}$, $G^{V}$ is an implementation of $\mathcal{Q}$.
	\item (\emph{black-box security reduction with isometry access}) There is a QPT algorithm $S^{(\cdot)}$ such that, for every isometry implementation $V$ of $\mathcal{P}$, every adversary $A$ that $\mathcal{Q}$-breaks $G^V$, and every isometry implementation $\tilde{A}$ of $A$, it holds that $S^{\tilde{A}}$ $\mathcal P$-breaks $V$.
\end{enumerate}
\end{definition}

\begin{definition}
\label{def:bb-unitary}
A QPT algorithm $G^{(\cdot)}$ is a \emph{fully black-box construction of $\mathcal{Q}$ from \textbf{unitary access} to $\mathcal{P}$} if the following two conditions hold:
\begin{enumerate}
	\item (\emph{black-box construction with unitary access}) For every unitary implementation $U$ of $\mathcal{P}$, $G^{U}$ is an implementation of $\mathcal{Q}$.
	\item (\emph{black-box security reduction with unitary access}) There is a QPT algorithm $S^{(\cdot)}$ such that, for every unitary implementation $U$ of $\mathcal{P}$, every adversary $A$ that $\mathcal{Q}$-breaks $G^U$, and every unitary implementation $\tilde{A}$ of $A$, it holds that $S^{\tilde{A}}$ $\mathcal P$-breaks $U$.
\end{enumerate}
\end{definition}

\begin{definition}
\label{def:bb-unitary-and-inverse}
	A QPT algorithm $G^{(\cdot)}$ is a \emph{fully black-box construction of $\mathcal{Q}$ from $\mathcal{P}$ \textbf{with access to the inverse}} if the following two conditions hold:
	\begin{enumerate}
		\item (\emph{black-box construction with access to the inverse}) For every unitary implementation $U$ of $\mathcal{P}$, $G^{U, U^{-1}}$ is an implementation of $\mathcal{Q}$.
		\item (\emph{black-box security reduction with access to the inverse}) There is a QPT algorithm $S^{(\cdot)}$ such that, for every unitary implementation $U$ of $\mathcal{P}$, every adversary $A$ that $\mathcal{Q}$-breaks $G^{U, U^{-1}}$, and every unitary implementation $\tilde{A}$ of $A$, it holds that $S^{\tilde{A}, \tilde{A}^{-1}}$ $\mathcal P$-breaks $U$\footnote{One could define even more variants of "fully black-box constructions" by separating the type of access that $G$ has to the implementation of $\mathcal{P}$ from the type of access that $S$ has to $A$ (currently they are consistent in each of Definitions \ref{def:bb-unitary}, \ref{def:bb-isometry}, and \ref{def:bb-unitary-and-inverse}). Here, we choose to limit ourselves to the these three definitions.}.
	\end{enumerate}
\end{definition}
These three notions of black-box constructions are related to each other in the following (unsurprising) way.
\begin{theorem}
    If there is a fully black-box construction $G^{(\cdot)}$ of primitive $\mathcal{Q}$ from isometry access to primitive $\mathcal{P}$ (as in \cref{def:bb-isometry}), then there is a fully black-box construction $\tilde{G}^{(\cdot)}$ of $\mathcal{Q}$ from unitary access to $\mathcal{P}$ (as in \cref{def:bb-unitary}).
\end{theorem}
\begin{proof}
$\tilde{G}$ is defined in a natural way: for a unitary implementation $U$ of $\mathcal{P}$, $\tilde{G}^{U}$ runs $G^V$, where $V$ is the isometry induced by $U$. The latter can of course be simulated with queries to $U$, by setting the work register to $\ket{0}$. An $\tilde{S}^{(\cdot)}$ satisfying item 2 of \cref{def:bb-unitary} can be defined analogously from an $S^{\cdot}$ satisfying item 2 of \cref{def:bb-isometry}.
\end{proof}
We also have the following.
\begin{theorem}
    A fully black-box construction $G^{(\cdot)}$ of primitive $\mathcal{Q}$ from isometry access to primitive $\mathcal{P}$ (as in \cref{def:bb-unitary}) is also a fully black-box construction of $\mathcal{Q}$ from $\mathcal{P}$ with access to the inverse (as in \cref{def:bb-unitary-and-inverse}).
\end{theorem}
\begin{proof}
This is immediate since \cref{def:bb-unitary-and-inverse} gives $G^{(\cdot)}$ and $S^{(\cdot)}$ access to strictly ``more'', namely the inverses.
\end{proof}

We thus point out that our separation result (\cref{thm:oracle-sep}) rules out only the strongest notion of fully black-box construction of $\prs$ from $\oprs$ (as in \cref{def:bb-isometry}), and thus is the ``weakest'' separating result that one could hope to obtain. 

As an example to help motivate these different definitions, the original construction of commitments from $\prs$ by Morimae and Yamakawa \cite{MY22a} is fully black-box, but \emph{with access to the inverse} (i.e.\ the weakest notion of fully black-box construction). This distinction is important, for example, when working in the CHRS model, or in the quantum auxiliary-input model considered in \cite{morimae2023unconditionally} and \cite{Qia23}: a construction of a $\prs$ in this model does not immediately yield a commitment scheme via the black-box construction of \cite{MY22a}, because the inverse of the $\prs$ generation procedure is not necessarily available in this model (since the generation procedure may use auxiliary states, and thus the ``inverse'' is not well-defined). On the other hand, the slight variation on the \cite{MY22a} construction, proposed in \cite{morimae2023unconditionally}, is fully black-box with unitary access (but without needing the inverse, as in \cref{def:bb-unitary}).

We now clarify the relationship between a \emph{quantum} oracle separation of primitives $\mathcal{P}$ and $\mathcal{Q}$ and the (im)possibility of a black-box construction of one from the other. 

The following is a quantum analog of a result by Impagliazzo and Rudich~\cite{IR89} (formalized in \cite{RTV04} using the above terminology).
\begin{theorem}
	\label{thm:quantumIR}
	 Suppose there exists a fully black-box construction of primitive $\mathcal{Q}$ from unitary (resp.\ isometry) access to primitive $\mathcal{P}$. Then, for every unitary (resp.\ isometry) $\mathcal{O}$, if $\mathcal{P}$ exists relative to $\mathcal{O}$, then $\mathcal{Q}$ also exists relative to $\mathcal{O}$.
\end{theorem}
This implies that a unitary (resp.\ isometry) oracle separation (i.e.\ the existence of an oracle relative to which $\mathcal{P}$ exists but $\mathcal{Q}$ does not) suffices to rule out a fully black-box construction of $\mathcal{Q}$ from unitary (resp.\ isometry) access to $\mathcal{P}$.

\begin{proof}[Proof of Theorem \ref{thm:quantumIR}]
We write the proof for the case of unitary access to $\mathcal{P}$. The proof for the case of isometry access is analogous (replacing unitaries with isometries).
	Suppose there exists a fully black-box construction of $\mathcal{Q}$ from $\mathcal{P}$. Then, by definition, there exist QPT $G^{(\cdot)}$ and $S^{(\cdot)}$ such that:
	\begin{enumerate}
		\item (\emph{black-box construction}) For every unitary implementation $U$ of $\mathcal{P}$, $G^{U}$ is an implementation of $\mathcal{Q}$.
		\item (\emph{black-box security reduction}) For every implementation $U$ of $\mathcal{P}$, every adversary $A$ that $\mathcal{Q}$-breaks $G^U$, and every unitary implementation $\tilde{A}$ of $A$, it holds that $S^{\tilde{A}}$ $\mathcal P$-breaks $U$.
	\end{enumerate}
Let $\mathcal O$ be a quantum oracle relative to which $\mathcal{P}$ exists. Since, by Definition~\ref{def:exist-relative-to-oracle}, $\mathcal{P}$ has an \emph{efficient} implementation relative to $\mathcal{O}$, there exists a uniform family of unitaries $U$ that is \emph{efficiently computable} with access to $\mathcal{O}$, such that $U$ is a unitary implementation of $\mathcal{P}$. Moreover, $U$ (or rather the quantum channel that $U$ implements) is a secure implementation of $\mathcal{P}$ relative to $\mathcal{O}$.

We show that the following QPT oracle algorithm $\tilde{G}^{(\cdot)}$ is an efficient implementation of $\mathcal{Q}$ relative to $\mathcal{O}$, i.e.\   $\tilde{G}^{\mathcal O} \in \mathcal{Q}$. $\tilde{G}^{\mathcal O}$ runs as follows: implement $G^{U}$ by running $G$, and simulate each call to $U$ by making queries to $\mathcal O$. Note that $\tilde{G}^{(\cdot)}$ is QPT because $U$ is a uniform family of efficiently computable unitaries given access to $\mathcal{O}$. Since $\tilde{G}^{\mathcal O}$ is equivalent to $G^{U}$, and $G^U \in \mathcal{Q}$ (by property 1 above), then $\tilde{G}^{\mathcal O} \in \mathcal{Q}$.

We are left with showing that $\tilde{G}^{\mathcal O}$ is a secure implementation relative to $\mathcal O$, i.e.\ that there is no QPT adversary $A^{(\cdot)}$ such that  $A^{\mathcal O}$ $\mathcal{Q}$-breaks $\tilde{G}^{\mathcal O}$. Suppose for a contradiction that there was a QPT adversary $A^{(\cdot)}$ such that $\mathcal{A}^{\mathcal{O}}$ $\mathcal{Q}$-breaks $\tilde{G}^{\mathcal O}$ (which is equivalent to $G^{U}$). Then, by property 2, $S^{A^\mathcal O}$ $\mathcal{P}$-breaks $U$. Note that adversary $S^{A^{\mathcal{O}}}$ can be implemented efficiently with oracle access to $\mathcal O$, because both $S^{(\cdot)}$ and $A^{(\cdot)}$ are QPT. Thus, this contradicts the security of $U$ relative to $\mathcal{O}$ (formally, of the quantum channel that $U$ implements).
\end{proof}

Similarly, we state a version of Theorem \ref{thm:quantumIR} for fully black-box constructions with access to the inverse.
\begin{theorem}
	\label{thm:quantumIRinverse}
	Suppose there exists a fully black-box construction of primitive $\mathcal{Q}$ from primitive $\mathcal{P}$ with access to the inverse.  Then, for every unitary $\mathcal O$, if $\mathcal{P}$ exists relative to $(\mathcal O, \mathcal O^{-1})$, then $\mathcal{Q}$ also exists relative to the oracle $(\mathcal O, \mathcal O^{-1})$. 
\end{theorem}
\begin{proof}
The proof is analogous to the proof of Theorem \ref{thm:quantumIR}. The only difference is that now $G^{(\cdot)}$ additionally makes queries to the inverse of the unitary implementation $U$ of $\mathcal{P}$. Since $U^{-1}$ can be implemented efficiently given access to $(\mathcal O, \mathcal O^{-1})$, we can now define an efficient implementation $\tilde{G}^{(\cdot)}$ of $\mathcal{P}$ relative to $(\mathcal O, \mathcal O^{-1})$. Proving that $\tilde{G}^{\mathcal{O}, \mathcal{O}^{-1}}$ is a secure implementation of $\mathcal{P}$ relative to $(\mathcal O, \mathcal O^{-1})$ also proceeds analogously.
\end{proof}	

    \section{Reduction from a ``state'' oracle to a unitary oracle}\label{sec:unitary-oracle-simulation}
Recall that the oracle separating $\oprs$ and $\prs$ in~\cref{sec:oracle-sep} is an \emph{isometry}. In particular, the CHRS part of the oracle provides copies of a Haar random state. Thus, so far, such a separation only rules out a fully black-box construction of a $\prs$ from ``isometry access'' to a $\oprs$ (as defined precisely in \cref{def:bb-isometry}). Informally, such a black-box construction is only allowed to use the generation procedure of the $\oprs$ as an ``isometry'', i.e.\ it does not have the ability to initialize the auxiliary qubits in an arbitrary state.

In this section, we show that our separation can be upgraded to be relative to a ``parametrized'' \emph{unitary} oracle and its inverse (for $\prs$ that have output length at least $\omega(\log n)$).\footnote{Recall that a parametrized oracle is a a family of oracles $\{O_\lambda\}$. Existence relative to $\{O_\lambda\}$ means that, for a security parameter $\lambda$, both the construction and the adversary are only allowed to query $O_\lambda$such that the construction and the adversary are only allowed to query $O_\lambda$. An oracle of this kind does not rule out the most general kind of black-box construction (which can make use of an arbitrary unitary implementation of primitive $A$, and its inverse, in order to build primitive $B$), but only rules out black-box constructions of primitive $B$ that, for a fixed security parameter $n$, only make use of a unitary implementation of $A$ for the same fixed security parameter $n$.} In particular, we introduce a unitary oracle, which is self-inverse that is approximately equivalent to the isometry oracle that gives out copies of a Haar random state $\ket{\psi}$: access to this unitary oracle allows one to exactly simulate access to copies of $\ket{\psi}$, and, conversely, the unitary oracle can be simulated \emph{approximately} using copies of $\ket{\psi}$.

As mentioned earlier, a separation of $\oprs$ and $\prs$ relative to a standard unitary oracle can be achieved via different techniques as in \cite{BMM+24} and \cite{GZ25}. 
The technique that we describe here is inspired by techniques in the works of Ji, Liu, and Song~\cite{JLS18} and Zhandry~\cite{zhandry2024space}. The former also considers simulation of a state-dependent unitary oracle, but the latter performs a ``reflection'' across a state, rather than a ``swap'' or ``replacement''. In this sense, our technique is more similar to Zhandry's~\cite{zhandry2024space}, with the difference that we consider the notion of ``global-phase'' invariance of a distribution over unitaries (which we define below), instead of ``relative-phase'' invariance. Overall, we show the following:
\begin{itemize}
    \item[(i)] If a primitive (with a security game consisting of a single-round, e.g.\ a $\oprs$ or $\efi$) exists relative to the CHRS oracle (or, in fact, relative to any distribution over states that is ``global-phase invariant'', as defined in Definition~\ref{def:global-phase-invariant} below), then it also exists relative to a corresponding parametrized unitary oracle. This is \cref{thm:16}.
    \item[(ii)] Conversely, $\prs$ with $\omega(\log n)$ output length do not exist relative to the parametrized unitary oracle (induced by the CHRS oracle), since we can still carry out (a suitably modified version of) the OR lemma attack on $\prs$ that we described in Section~\ref{sec:oralce-sep}.
\end{itemize}

\subsection{Unitary corresponding to a state}
\label{sec:unitaryoracle0}
Throughout the section, let $\ket{\psi}$ be an $n$-qubit state orthogonal to $\ket{0^n}$. In the CHRS model, the common Haar state $\ket{\psi}$ is not necessarily orthogonal to $\ket{0^n}$, but we take them to be be orthogonal at first for simplicity. The result we prove will extend straightforwardly to the case of arbitrary $\ket{\psi}$. For convenience of notation, we will write $\ket{0}$ instead of $\ket{0^n}$ (more generally, we will use $\ket{0}$ to denote the all zero state of a system whose dimension is clear from the context).

We define a corresponding unitary $\upsi$ as follows: $\upsi$ flips $\ket{0}$ and $\ket{\psi}$, and acts as the identity on everything orthogonal to the subspace spanned by $\ket{0}$ and $\ket{\psi}$, i.e.\ $\upsi\ket{0}=\ket{\psi}$, $\upsi\ket{\psi}=\ket{0}$, and $\upsi\ket{\phi}=\ket{\phi}$ for any $\ket{\phi}$ orthogonal to $\ket{0}$ and $\ket{\psi}$. Notice that $\upsi$ is self-inverse.



Consider an algorithm $\calA^{\upsi}$ that makes $T$ queries to $\upsi$, we will show that one can simulate $\calA^{\upsi}$ with $\epsilon$ precision given $O\left(\frac {T^2}{\epsilon^2}\right)$ copies of $\ket{\psi}$ in the following average sense. 

For any $\ket{\psi}$, and an arbitrary input state $\ket{\sigma}$, we can write the output of $\calA^{\upsi}$ as
\begin{equation*}
    \ket{\Psi_{\psi,T}} = B_T\upsi B_{T-1} \dots B_1 \upsi B_0 \ket{\sigma},
\end{equation*}
for some fixed unitaries $B_0, \dots, B_T$ that do not depend on $\ket{\psi}$. Then, we consider the average of this output over a uniformly random phase $\alpha$, namely $\alpha$ is sampled as a random point on the unit circle $|\alpha|=1$:
\begin{equation}
\label{eq:600}
    \rho_{\psi, T} = \E_\alpha \left[ \ket{\Psi_{\alpha{\ket{\psi}}, T}}\bra{\Psi_{\alpha\ket{\psi}, T}} \right].
\end{equation}
We establish that $\rho_{\psi,T}$ can be simulated approximately given copies of $\ket{\psi}$.
\begin{theorem}
\label{thm:isometry-to-unitary}
    Let $n\in \mathbb{N}$. Let $\ket{\psi}$ be any $n$-qubit state orthogonal to $\ket{0^n}$. Let $\epsilon >0$, and $T\in \mathbb{N}$. Let $\upsi$ be the $n$-qubit unitary defined as above, and let $\rho_{\psi, T}$ be as in Equation \eqref{eq:600}. For any oracle algorithm $\calA^{(\cdot)}$ making $T$ queries to $\upsi$, there is an algorithm $\wt{\cal{A}}$ that, with access to $O\left(\frac {T^2}{\epsilon^2}\right)$ copies of $\ket{\psi}$, outputs a state $\wt{\rho}_{\psi, T}$ that is $\eps$-close to $\rho_{\psi, T}$ in trace distance.
\end{theorem}

\begin{corollary}
\label{cor:20}
Let $n\in \mathbb{N}$. Let $\ket{\psi}$ be any $n$-qubit state.  Let $\epsilon >0$, and $T\in \mathbb{N}$. Define the $(n+1)$-qubit state $\ket{\psi'} = \ket{\psi} \otimes \ket{1}$. Let $U_{\ket{\psi'}}$ be the $(n+1)$-qubit unitary defined as above, and let $\rho_{\psi', T}$ be as in Equation \eqref{eq:600}.
For any oracle algorithm $\calA^{(\cdot)}$ making $T$ queries to $U_{\ket{\psi'}}$, there is an algorithm $\wt{\cal{A}}$ that, with access to $O\left(\frac {T^2}{\epsilon^2}\right)$ copies of $\ket{\psi}$, outputs a state $\wt{\rho}_{\psi', T}$ that is $\eps$-close to $\rho_{\psi', T}$ in trace distance.
\end{corollary}

Corollary \ref{cor:20} follows immediately from Theorem \ref{thm:isometry-to-unitary}. We will prove Theorem \ref{thm:isometry-to-unitary} in the next two sections.

The proof proceeds in two steps. The first step (\cref{sec:unitaryoracle1}) is to show that $\rho_{\psi, T}$ can be produced \emph{perfectly} with access to $T$ copies of $\ket{\psi}$ \emph{and} a certain auxiliary unitary oracle $C_{\ket{\psi}}$. The second step (\cref{sec:unitaryoracle2}) is to show that $C_{\ket{\psi}}$ can be simulated approximately using copies of $\ket{\psi}$. In Section \ref{sec:weak-simulation}, we justify why the weak notion of simulation that we achieve is sufficient to lift our separation results to be relative to the new unitary oracle.

\subsection{Weak simulation of the unitary oracle with a \texorpdfstring{``$\ket{\psi}$''}{''|psi>''}-controlled gate}\label{sec:simulation-unitary}
\label{sec:unitaryoracle1}
Let $\calA^{(\cdot)}$ be an algorithm that makes $T$ queries to $U_{\psi}$. Consider the auxiliary unitary oracle $C_{\ket{\psi}}$ that acts on two registers and performs a ``control-NOT'', controlled on the first register being $\ket{\psi}$. Formally, this is defined as follows:
\begin{align*}
C_{\ket{\psi}}\ket{\psi}\ket{b} &= \ket{\psi}\ket{b \oplus 1} \\
C_{\ket{\psi}}\ket{\phi}\ket{b} &= \ket{\phi} \ket{b} \text{ for any }\braket{\phi|\psi}=0.
\end{align*}
As before, recall that we can, without loss of generality, write the output of $\calA^{\upsi}$ as
\begin{equation*}
    \ket{\Psi_{\psi,T}} = B_T\upsi B_{T-1} \dots B_1 \upsi B_0 \ket{\sigma},
\end{equation*}
where $\ket{\sigma}$ is an arbitary quantum input to the algorithm $\calA^{\upsi}$.
And $B_0, \dots, B_T$ are some fixed unitareis that do not depend on $\ket{\psi}$. Then,
\begin{equation}
    \rho_{\psi, T} = \E_\alpha \left[ \ket{\Psi_{\alpha{\ket{\psi}}, T}}\bra{\Psi_{\alpha\ket{\psi}, T}} \right].
\end{equation}
We show that there is an algorithm $\wt{\cal{A}}$ that outputs exactly $\rho_{\psi, T}$, given access to $T$ copies of $\ket{\psi}$ as well as the unitary $C_{\ket{\psi}}$. 

The simulation algorithm $\wt{\cal{A}}$ will run $\cal{A}$ normally, except that, in order to simulate queries to $U_{\ket{\psi}}$, it will leverage a ``pool'' of $T$ copies of $\ket{\psi}$, and the ``control-NOT'' unitary $C_{\ket{\psi}}$.


Very informally, the idea behind is the following. For each query that $\cal{A}$ makes to $U_{\ket{\psi}}$, we first check whether the query register is $\ket{0}$, $\ket{\psi}$ or a state orthogonal to it (we can do this with the assistance of $C_{\ket{\psi}}$). If it is $\ket{0}$, $\wt{\cal{A}}$ swaps it with a $\ket{\psi}$ from the pool, and vice versa. If it is orthogonal to both, $\wt{\cal{A}}$ applies the identity. In this way, the ``pool'' register can be viewed as counting the number of ``net'' queries made on a particular branch. This approach might seem suspicious at first as it entangles the query register with the ``pool''. In particular, the state of the simulation will be in a superposition of ``pools'' with a different number of $\ket{\psi}$. Moreover, note that, since $\ket{0}$ and $\ket{\psi}$ states are orthogonal, states representing ``pools'' with distinct numbers of $\ket{\psi}$ are also orthogonal to each other. Thus, tracing out the ``pool'' register results in a mixture of states, each corresponding to a different number of ``effective'' queries (here ``effective'' captures the fact that, for example, making two consecutive queries results in an identity, and so the number of effective queries would be zero -- this point of view is somewhat reminiscent of Zhandry's compressed oracle technique for recording queries~\cite{zhandry2019record}). Recall that we claimed to be able to achieve \emph{perfect} simulation: why would the traced out be exactly the original state output by ${\cal{A}}^{\upsi}$?

Recall that we are only hoping to achieve a simulation that is faithful \emph{on average over $\alpha$}. Then, the key insight is the following: while for a fixed $\alpha$, the state output by ${\cal{A}}^{\upsi}$ is in general a \emph{superposition} (rather than a mixture) over branches corresponding to a different number of ``effective'' 
queries, averaging over $\alpha$ causes the cross terms of the density matrix (corresponding to a different number of effective queries) to vanish. One nice way to see this is that the state output by ${\cal{A}}^{\upsi}$ can be viewed as as polynomial in $\alpha$, where the term of degree $i$ corresponds to the branches of the superposition with $i$ effective queries. The corresponding density matrix can also be thought of as a polynomial in $\alpha$, and the observation is that entries of the density matrix that have non-zero degree vanish when averaging over $\alpha$ (such terms are precisely the cross terms corresponding to branches with a different number of effective queries).



We now formally describe how $\wt{\cal{A}}$ simulates queries to $\upsi$. $\wt{\cal{A}}$ acts on the following registers:
\begin{itemize}
\item $\sfA$, consisting of $\sfA_1$ and $\sfA_2$. These are respectively the ``query'' and work registers of the original algorithm $\mathcal{A}$. In particular, $\sfA_1$ contains the state we wish to apply $\upsi$ to.
\item $\sfB$, which will store the pool of copies of (initially) $\ket{\psi}$. The auxiliary pool $\sfB$ is initialized as $\ket{\psi}^{\otimes T} \otimes \ket{0}^{\otimes T}$, and the algorithm can retrieve or deposit $\ket{\psi}$ from $\sfB$. We denote the $2T$ sub-registers of $\sfB$ as $\sfB_1, \ldots, \sfB_{2T}$.
\item $\sfC$, a ``counting'' register that is initialized as $\ket{0}$, and counts how many $\ket{\psi}$ have been ``borrowed'' from the pool. Register $\sfC$ is of dimension $2T+1$, and we denote its standard basis as $\{\ket{-T}, \ldots, \ket{0}, \ldots, \ket{T}\}$ (where a negative value means that we ``deposited'' more $\ket{\psi}$ than we have ``borrowed'').
\item $\sfD$, consisting of $\sfD_1, \sfD_2, \sfD_3$ is an additional control register.
\end{itemize} 
$\wt{\cal{A}}$ proceeds as follows:
\begin{itemize}
    \item[(i)] Apply $C_{\ket{0}}$ to $\sfA_1$ and $\sfD_1$, where $C_{\ket{0}}$ acts as follows: $C_{\ket{0}} \ket{0}_{\sfA_1} \ket{b}_{\sfD_1} = \ket{0}_{\sfA_1}\ket{b \oplus 1}_{\sfD_1}$, and $ C_{\ket{0}}\ket{i}_{\sfA_1} \ket{b}_{\sfD_1} = \ket{i}\ket{b}$ for all $i \neq 0$.
    \item[(ii)] Apply $C_{\ket{\psi}}$ to $\sfA_1$ and $\sfD_2$.
    \item[(iii)] Update the counter in $\sfC$ by subtracting the value in $\sfD_2$. 
    Formally, this subtraction is modulo $2T+1$ (with values represented in $\{-T,\ldots, T\}$) although our algorithm is such that a ``wrap around'' is never required.
    \item[(iv)] Compute the OR of $\sfD_1$ and $\sfD_2$ in $\sfD_3$.
    \item[(v)] Perform a ``controlled-SWAP'' on registers $\sfC$, $\sfD_3$, $\sfA_1$ and $\sfB$ that acts as follows on the standard basis: if $\sfD_3$ is $\ket{0}$, act as the identity; if $\sfD_3$ is $\ket{1}$ and $\sfC$ is $\ket{i}$, then swap the register $\sfA_1$ with $\sfB_{T-i}$.
    \item[(vi)] Add the value of $\sfD_1$ to the counter $\sfC$.
    \item[(vii)] ``Uncompute'' $\sfD_1, \sfD_2, \sfD_3$ (so that they return to zero): first, compute the OR of $\sfD_1$ and $\sfD_2$ in $\sfD_3$ (this uncomputes the OR that we performed previously); then apply $C_{\ket{0}}$ to $\sfA_1$ and $\sfD_2$, followed by $C_{\ket{\psi}}$ to $\sfA_1$ and $\sfD_1$ (note that we have reversed the role of the registers $\sfD_1$ and $\sfD_2$ here, since we have now swapped $\ket{0}$ and $\ket{\psi}$ in $\sfA_1$).
\end{itemize}
We will show that the reduced density matrix on $\sfA$ is exactly $\rho_{\psi, T}$. We start by noticing that the output of ${\cal{A}}^{\uapsi}$ can be viewed as a polynomial in $\alpha$ and $\alpha^{-1}$ of degree at most $T$.

\begin{lemma}\label{lem:unitary-output-poly}
   Let $\ket{\psi}$ be any state, and let $\calA^{(\cdot)}$ be any algorithm making $T$ queries to an oracle of the form $U_{\alpha \ket{\psi}}$ for $\alpha \in \mathbb{C}$ with $|\alpha| = 1$. Let $\ket{\Psi_{\alpha\ket{\psi}, T}}$ denote the output of ${\cal{A}}^{\uapsi}$. When $\ket{\psi}$ is fixed, the amplitudes of $\ket{\Psi_{\alpha\ket{\psi}, T}}$ are polynomials in $\alpha$ and $\alpha^{-1}$ of degree at most $T$. More precisely, there exist un-normalized states $\ket{\phi_i}$, such that
\begin{equation*}
\ket{\Psi_{\alpha\ket{\psi}, T}} = \sum_{i=-T}^T \alpha^i \ket{\phi_i}_\sfA \,.
\end{equation*}
\end{lemma}
\begin{proof}
We prove the lemma by induction on $T$. When $T=0$, the algorithm $\calA^{U_{\alpha\ket{\psi}}}$ does not call the unitary oracle, thus the output will be a fixed state $\ket{\phi_0}$.

Assume the proposition holds for some number $T-1$ of queries. Then, the state
\begin{equation*}
\ket{\Psi_{\alpha\ket{\psi}, T-1}} = B_{T-1}U_{\alpha \ket{\psi}} \dots U_{\alpha\ket{\psi}} B_0 \ket{0}
\end{equation*}
can be expressed as $\ket{\Psi_{\alpha\ket{\psi}, T-1}} = \sum_{i=-T+1}^{T-1}\alpha^i\ket{\phi_i}_\sfA$ for some un-normalized $\ket{\phi_i}$. We can decompose the states $\ket{\phi_i}$ as $\ket{\phi_i}_\sfA = a_i\ket{0}_{\sfA_1}\ket{\phi_{i,1}}_{\sfA_2} + b_i \ket{\psi}_{\sfA_1}\ket{\phi_{i,2}}_{\sfA_2} + c_i \ket{\phi_i^\perp}_\sfA$, for some $a_i, b_i, c_i \in \mathbb{C}$, and normalized states $\ket{\phi_{i,1}}$, $\ket{\phi_{i,2}}$, and $\ket{\phi_i^{\perp}}$, where $\braket{0|_{\sfA_1} \otimes \I_{\sfA_2} |\phi_i^{\perp}}_{\sfA_1\sfA_2} = \braket{\psi|_{\sfA_1} \otimes \I_{\sfA_2} |\phi_i^{\perp}}_{\sfA_1\sfA_2}  = 0$.
Then after applying $U_{\alpha\ket{\psi}}$, the state becomes
\ifnum\widemargin=0
\begin{align*}
    U_{\alpha\ket{\psi}} \ket{\Psi_{\alpha\ket{\psi}, T-1}} &= U_{\alpha\ket{\psi}} \sum_{i=-T+1}^{T-1} \alpha^i ( a_i \ket{0}_{\sfA_1}\ket{\phi_{i,1}}_{\sfA_2} + b_i \ket{\psi}_{\sfA_1}\ket{\phi_{i,2}}_{\sfA_2} + c_i\ket{\phi_i^\perp}_\sfA ) \\
    &= \sum_{i=-T+1}^{T-1} ( \alpha^{i+1}a_i\ket{\psi}_{\sfA_1}\ket{\phi_{i,1}}_{\sfA_2} + \alpha^{i-1} b_i\ket{0}_{\sfA_1}\ket{\phi_{i,2}}_{\sfA_2} + \alpha^i c_i \ket{\phi_i^\perp}_{\sfA} ) \\
    &= \sum_{i=-T}^T \alpha^i ( a_{i-1}\ket{\psi}_{\sfA_1}\ket{\phi_{i-1,1}}_{\sfA_2} + b_{i+1}\ket{0}_{\sfA_1}\ket{\phi_{i+1,2}}_{\sfA_2} + c_i\ket{\phi_i^\perp}_\sfA ),
\end{align*}
\else
\begin{align*}
    &U_{\alpha\ket{\psi}} \ket{\Psi_{\alpha\ket{\psi}, T-1}} \\&= U_{\alpha\ket{\psi}} \sum_{i=-T+1}^{T-1} \alpha^i ( a_i \ket{0}_{\sfA_1}\ket{\phi_{i,1}}_{\sfA_2} + b_i \ket{\psi}_{\sfA_1}\ket{\phi_{i,2}}_{\sfA_2} + c_i\ket{\phi_i^\perp}_\sfA ) \\
    &= \sum_{i=-T+1}^{T-1} ( \alpha^{i+1}a_i\ket{\psi}_{\sfA_1}\ket{\phi_{i,1}}_{\sfA_2} + \alpha^{i-1} b_i\ket{0}_{\sfA_1}\ket{\phi_{i,2}}_{\sfA_2} + \alpha^i c_i \ket{\phi_i^\perp}_{\sfA} ) \\
    &= \sum_{i=-T}^T \alpha^i ( a_{i-1}\ket{\psi}_{\sfA_1}\ket{\phi_{i-1,1}}_{\sfA_2} + b_{i+1}\ket{0}_{\sfA_1}\ket{\phi_{i+1,2}}_{\sfA_2} + c_i\ket{\phi_i^\perp}_\sfA ),
\end{align*}
\fi
where we set $a_i=b_i=c_i=0$ if $|i|\geq T$. Thus the state $U_{\alpha\ket{\psi}}\ket{\Psi_{\alpha\ket{\psi}, T-1}}$ can be written as polynomial in $\alpha$ and $\alpha^{-1}$ with degree less than $T$. The fixed unitary $B_T$ (which is independent of $\alpha$) does not alter this form. Thus, $\ket{\Psi_{\alpha\ket{\psi}, T}} = B_TU_{\alpha\ket{\psi}} \ket{\Psi_{\alpha\ket{\psi}, T-1}}$ has the desired form.
\end{proof}

\begin{lemma}
\label{lem:16}
    Let $\ket{\psi}$ be any state, and let $\calA^{(\cdot)}$ be any algorithm making $T$ queries to an oracle of the form $U_{\alpha \ket{\psi}}$ for $\alpha \in \mathbb{C}$ with $|\alpha| = 1$. Let the $\ket{\phi_i}$ be un-normalized states such that, for all $\alpha$, the output of $\calA^{U_{\alpha\ket{\psi}}}$is $$\ket{\Psi_{\alpha\ket{\psi}, T}} = \sum_{i=-T}^T \alpha^i \ket{\phi_i}_\sfA$$ (such $\ket{\phi_i}$ exist by \cref{lem:unitary-output-poly}). Then, the simulation algorithm $\wt{\calA}$ outputs the state
    \begin{equation*}
    \ket{\wt{\Psi}_{\ket{\psi},T}} = \sum_{i=-T}^{T} \ket{\phi_i}_\sfA \otimes ( \ket{\psi}^{\otimes (T-i)} \otimes \ket{0}^{\otimes (T+i)})_\sfB \otimes \ket{i}_\sfC \otimes \ket{0}_{\sfD} \,.
    \end{equation*}
    As an immediate corollary, the reduced density matrix of $\wt{\Psi}_{\ket{\psi},T}$ on $\sfA$ is $$\tr_{\sfB \sfC \sfD} \wt{\Psi}_{\ket{\psi},T} = \sum_{i=-T}^T \ket{\phi_i}\bra{\phi_i}\,,$$ which is exactly $\rho_{\psi, T} = \mathbb{E}_{\alpha} \Psi_{\alpha \ket{\psi}, T}$.
\end{lemma}
\begin{proof}
    We prove the lemma by induction on $T$. When $T=0$, the statement is trivial. Assume the statement is true for $T-1$, i.e.\ $\ket{\wt{\Psi}_{\ket{\psi}, T-1}}$. According to~\cref{lem:unitary-output-poly}, there exist $\beta_i, \ket{\phi_i}$ such that, for all $\alpha$, 
    $$\ket{\Psi_{\alpha\ket{\psi}, T-1}} = \sum_{i=-T+1}^{T-1}\alpha^i \ket{\phi_i}_\sfA \,.$$
    Then, by induction hypothesis, $$\ket{\wt{\Psi}_{\ket{\psi}, T-1}}  = \sum_{i=-T+1}^{T-1} \ket{\phi_i}_\sfA \otimes ( \ket{\psi}^{\otimes (T-i)} \otimes \ket{0}^{\otimes (T+i)})_\sfB \otimes \ket{i}_\sfC \otimes \ket{0}_\sfD \,.$$
Let $B_{T-1}$ be any fixed unitary, and let $\ket{\phi'_i} = B_{T_1} \ket{\phi_i}$. Then, we have, by linearity, that
    \begin{align*}
        B_{T-1} \ket{\Psi_{\alpha\ket{\psi}, T-1}} &= \sum_{i=-T+1}^{T-1} \alpha^i \ket{\phi'_i}_\sfA \\
        (B_{T-1} \otimes \I) \ket{\wt{\Psi}_{\ket{\psi}, T-1}} &= \sum_{i=-T+1}^{T-1} \ket{\phi'_i}_\sfA \otimes (\ket{\psi}^{\otimes (T-i)} \otimes \ket{0}^{\otimes (T+i)}_\sfB \otimes \ket{i}_\sfC \otimes \ket{0}_\sfD.
    \end{align*}
    We can decompose each $\ket{\phi_i'}$ as 
    \begin{equation*}
        \ket{\phi_i'} =  a_i \ket{0}_{\sfA_1}\ket{\phi_{i,1}}_{\sfA_2} + b_i \ket{\psi}_{\sfA_1}\ket{\phi_{i, 2}}_{\sfA_2} + c_i \ket{\phi_i^\perp}_{\sfA_1 \sfA_2} \,,
    \end{equation*}    
    where $\braket{\psi|_{\sfA_1} \otimes \I_{\sfA_2} |\phi_i^\perp}_{\sfA_1\sfA_2} = \braket{0|_{\sfA_1} \otimes \I_{\sfA_2}|\phi_i^\perp}_{\sfA_1\sfA_2} = 0$. Thus the state $\ket{\Psi_{\ket{\psi}, T}}$ can be expressed as
    \begin{align*}
        \ket{\Psi_{\ket{\psi}, T}} &= \upsi B_{T-1} \ket{\Psi_{\ket{\psi}, T-1}} \\
        &= \sum_{i=-T+1}^{T-1} \alpha^i(a_i \alpha \ket{\psi}_{\sfA_1}\ket{\phi_{i,1}}_{\sfA_2} + \alpha^{-1}b_i\ket{0}_{\sfA_1}\ket{\phi_{i,2}}_{\sfA_2} + c_i \ket{\phi_i^\perp}_{\sfA_1\sfA_2}) \\
        &= \sum_{i = -T}^{T} \alpha^i ( a_{i-1}\ket{\psi}_{\sfA_1}\ket{\phi_{i-1,1}}_{\sfA_2} + b_{i+1}\ket{0}_{\sfA_1}\ket{\phi_{i+1, 2}}_{\sfA_2} + c_i \ket{\phi_i^\perp}_{\sfA_1\sfA_2}) \,,
    \end{align*}
    where we set $a_i=b_i=c_i=0$ if $|i|\geq T$.
    On the other hand, we need to consider the output of the simulation on $(B_{T-1}\otimes \I)\ket{\wt{\Psi}_{\ket{\psi}, T-1}}$. After we apply $C_{\ket{0}}, C_{\ket{\psi}}$, the state turns into (we will abbreviate $\ket{\psi}^{\otimes (T-i)} \otimes \ket{0}^{\otimes (T+i)}$ as $\ket{\psi}^{\otimes (T-i)}$): 
    \ifnum\widemargin=0
    \begin{equation*}
        \sum_{i = -T}^{T}  \left( a_i\ket{0}_{\sfA_1}\ket{\phi_{i,1}}_{\sfA_2} \ket{100}_{\sfD} + b_i \ket{\psi}_{\sfA_1}\ket{\phi_{i,2}}_{\sfA_2}\ket{010}_\sfD + c_i \ket{\phi_i^\perp}_\sfA \ket{000}_\sfD\right) \otimes \ket{\psi}^{\otimes (T+i)}_\sfB \otimes \ket{i}_\sfC  \,.
    \end{equation*}
    \else
    \begin{align*}
       \sum_{i = -T+1}^{T-1}  &\left( a_i\ket{0}_{\sfA_1}\ket{\phi_{i,1}}_{\sfA_2} \ket{100}_{\sfD} + b_i \ket{\psi}_{\sfA_1}\ket{\phi_{i,2}}_{\sfA_2}\ket{010}_\sfD + c_i \ket{\phi_i^\perp}_\sfA \ket{000}_\sfD\right) \\&\otimes \ket{\psi}^{\otimes (T-i)}_\sfB \otimes \ket{i}_\sfC  
    \end{align*}
    \fi
    After updating the counter $\sfC$, we get
    \ifnum\widemargin=0
    \begin{equation*}
        \sum_{i =-T+1}^{T-1}  \left( a_i \ket{0}_{\sfA_1} \ket{\phi_{i,1}}_{\sfA_2} \ket{i}_\sfC \ket{100}_\sfD + b_i \ket{\psi}_{\sfA_1}\ket{\phi_{i,2}}_{\sfA_2}\ket{i+1}\ket{010}_\sfD + c_i\ket{\phi_i^\perp}_\sfA\ket{i}_\sfC\ket{000}_\sfD \right) \otimes \ket{\psi}_B^{\otimes (T-i)}.
    \end{equation*}
    \else
    \begin{align*}
        \sum_{i = -T+1}^{T-1} \Big( a_i \ket{0}_{\sfA_1} \ket{\phi_{i,1}}_{\sfA_2} \ket{i}_\sfC \ket{100}_\sfD +& b_i \ket{\psi}_{\sfA_1}\ket{\phi_{i,2}}_{\sfA_2}\ket{i-1}_{\sfC}\ket{010}_\sfD \\&+ c_i\ket{\phi_i^\perp}_\sfA\ket{i}_\sfC\ket{000}_\sfD \Big) \otimes \ket{\psi}_B^{\otimes (T-i)}.
    \end{align*}
    \fi
     After computing the OR of $\sfD_1$ and $\sfD_2$ in $\sfD_3$, we get
    \ifnum\widemargin=0
    \begin{equation*}
        \sum_{i =-T+1}^{T-1}  \left( a_i \ket{0}_{\sfA_1} \ket{\phi_{i,1}}_{\sfA_2} \ket{i}_\sfC \ket{101}_\sfD + b_i \ket{\psi}_{\sfA_1}\ket{\phi_{i,2}}_{\sfA_2}\ket{i+1}\ket{011}_\sfD + c_i\ket{\phi_i^\perp}_\sfA\ket{i}_\sfC\ket{000}_\sfD \right) \otimes \ket{\psi}_B^{\otimes (T-i)}.
    \end{equation*}
    \else
    \begin{align*}
        \sum_{i = -T+1}^{T-1} \Big( a_i \ket{0}_{\sfA_1} \ket{\phi_{i,1}}_{\sfA_2} \ket{i}_\sfC \ket{101}_\sfD +& b_i \ket{\psi}_{\sfA_1}\ket{\phi_{i,2}}_{\sfA_2}\ket{i-1}_{\sfC}\ket{011}_\sfD \\&+ c_i\ket{\phi_i^\perp}_\sfA\ket{i}_\sfC\ket{000}_\sfD \Big) \otimes \ket{\psi}_B^{\otimes (T-i)}.
    \end{align*}
    \fi
    After the ``controlled-SWAP'', the state becomes
    \begin{align*}
        \sum_{i = -T+1}^{T-1} & \Big( a_i \ket{\psi}_{\sfA_1} \ket{\phi_{i,1}}_{\sfA_2}\ket{\psi}_\sfB^{\otimes (T-i-1)} \ket{i}_\sfC \ket{101}_\sfD \\&+ b_i \ket{0}_{\sfA_1}\ket{\phi_{i,2}}_{\sfA_2}\ket{\psi}_\sfB^{\otimes (T-i+1)}\ket{i-1}\ket{011}_\sfD  + c_i\ket{\phi_i^\perp}_\sfA\ket{\psi}^{\otimes (T-i)}_\sfB\ket{i}_\sfC\ket{000}_\sfD \Big) \\
        &= \sum_{i=-T}^T \Big( a_{i-1}\ket{\psi}_{\sfA_1}\ket{\phi_{i-1,1}}_{\sfA_2}\ket{i-1}_\sfC\ket{101}_\sfD  \\&+b_{i+1}\ket{0}_{\sfA_1}\ket{\phi_{i,2}}_{\sfA_2}\ket{i+1}_\sfC\ket{011}_\sfD + c_i\ket{\phi_i^\perp}_\sfA \ket{i}_\sfC\ket{000}_\sfD \Big) \otimes \ket{\psi}_\sfB^{\otimes (T-i)}.
    \end{align*}
    After updating the counter and the uncomputation, the state becomes
    \ifnum\widemargin=0
    \begin{equation*}
        \sum_{i = -T}^{T} \left( a_{i-1}\ket{\psi}_{\sfA_1}\ket{\phi_{i-1,1}}_{\sfA_2} + b_{i+1}\ket{0}_{\sfA_1}\ket{\phi_{i+1,1}}_{\sfA_2} + c_i \ket{\phi_i^\perp}_\sfA \right) \otimes \ket{\psi}^{\otimes (T-i)}_\sfB \otimes \ket{i}_\sfC \otimes \ket{000}_\sfD
    \end{equation*}
    \else
    \begin{align*}
        \sum_{i=-T}^{T} \Big( a_{i-1}\ket{\psi}_{\sfA_1}\ket{\phi_{i-1,1}}_{\sfA_2} + b_{i+1}\ket{0}_{\sfA_1}&\ket{\phi_{i+1,2}}_{\sfA_2} + c_i \ket{\phi_i^\perp}_\sfA \Big) \\ &\otimes \ket{\psi}^{\otimes (T-i)}_\sfB \otimes \ket{i}_\sfC \otimes \ket{000}_\sfD
    \end{align*}
    \fi
    as desired.
\end{proof}

\subsection{Approximating the \texorpdfstring{``$\ket{\psi}$''}{''|psi>''}-controlled gate using copies of \texorpdfstring{$\ket{\psi}$}{|psi>}}
\label{sec:unitaryoracle2}
In~\cref{sec:simulation-unitary}, we have described how to produce $\rho_{\psi, T}$ perfectly with the assistance of the gate $C_{\ket{\psi}}$. 
In this section, we show how to implement $C_{\ket{\psi}}$ approximately, with some precision $\epsilon$, using $O\left(\frac{1}{\epsilon^2}\right)$ copies of the state $\ket{\psi}$. \anote{clarified the number of copies}Notice that our simulation algorithm $\wt{\cal{A}}$ applies $C_{\ket{\psi}}$ $2T$ times in total. Using $O\left(\frac{T^2}{\epsilon^2}\right)$ copies of $\ket{\psi}$, we can implement one $C_{\ket{\psi}}$ to precision $\frac{\epsilon}{T}$. Thus, using $O\left(\frac{T^2}{\epsilon^2}\right)$ copies of $\ket{\psi}$, we can implement $2T$ $C_{\ket{\psi}}$ gates, each to precision $\frac{\epsilon}{T}$. By a triangle inequality, this suffices to approximate the output of $\wt{\cal{A}}$, and thus $\rho_{\psi, T}$, with an overall precision of $\epsilon$.

In order to simulate $C_{\ket{\psi}}$, we consider a generalized $N$-copy SWAP test. Assume we have $N$ copies of $\ket{\psi}$ at our disposal. We define a unitary that is meant to act on a state of the form $\ket{\phi}_\sfA \otimes \ket{\psi}^{N}_\sfB \otimes \ket{b}_\sfC$, as follows: controlled on the first $N+1$ registers being in the symmetric subspace, it flips the $\sfC$ register, otherwise it applies the identity. Formally,
\begin{equation*}
    C_{\swap} = \Pi_{\sfA\sfB}^{sym} \otimes X_\sfC + (I - \Pi_{\sfA\sfB}^{sym}) \otimes \I_\sfC.
\end{equation*}
We claim that the behavior of $C_{\swap}$ is inverse-polynomially close to $C_{\ket{\psi}}$. More formally,
\begin{lemma}
\label{lem:13}
    For any $\ket{\psi}$ and any state $\ket{\phi}_{\sfA\sfC\sfD}$, we have
    \begin{equation*}
        \norm{ \left((C_{\swap})_{\sfA\sfB\sfC} \otimes \I_\sfD \right)(\ket{\phi}_{\sfA\sfC\sfD} \otimes \ket{\psi}_\sfB^{\otimes N}) - \left((C_{\ket{\psi}})_{\sfA\sfC} \otimes \I_{\sfB\sfD} \right) (\ket{\phi}_{\sfA\sfC\sfD} \otimes \ket{\psi}_\sfB^{\otimes N})} \leq \frac{2}{\sqrt{N+1}}.
    \end{equation*}
\end{lemma}
\begin{proof}
    First we compute the action of $C_{\swap}$ more explicitly. The state $\ket{\psi}^{\otimes N+1}$ lies in $\Pi^{sym}$, so $C_{\swap}(\ket{\psi}_\sfA\ket{\psi}^{\otimes N}_\sfB\ket{b}_\sfC) = \ket{\psi}_\sfA\ket{\psi}^{\otimes N}_\sfB\ket{b \oplus 1}_\sfC$.

    Note that for any state $\ket{\chi}$ orthogonal to $\ket{\psi}$, we have \begin{align}
    \ket{\chi} \ket{\psi}^{\otimes N} &= \Pi^{sym} \ket{\chi} \ket{\psi}^{\otimes N} + (\I-\Pi^{sym}) \ket{\chi} \ket{\psi}^{\otimes N} \\
    &= \frac{1}{\sqrt{N+1}} \ket{\chi, \psi} + \frac{\sqrt{N}}{\sqrt{N+1}}\ket{\chi^{\perp}}  \,,
    \end{align}
    where
    \begin{equation*}
        \ket{\chi, \psi} = \frac{1}{\sqrt{N+1}} ( \ket{\chi\psi\dots\psi} + \dots + \ket{\psi\psi\dots\chi}),
    \end{equation*}
    and $\ket{\chi^\perp}$ is some state orthogonal to $\ket{\chi, \psi}$ that lies in the span of $\I-\Pi^{sym}$. So, we have
    \begin{equation}
    \label{eq:cswap}
        C_{\swap}(\ket{\chi}_\sfA\ket{\psi}^{\otimes N}_\sfB\ket{b}_\sfC) = \frac{1}{\sqrt{N+1}} \ket{\chi, \psi}_{\sfA\sfB}\ket{b\oplus 1}_{\sfC} + \frac{\sqrt{N}}{\sqrt{N+1}}\ket{\chi^\perp}\ket{b}_{\sfC} \,.
    \end{equation}
    Note also that, by a triangle inequality, $\norm{\ket{\chi^\perp} - \ket{\chi} \ket{\psi}^{\otimes N}} \leq \frac{2}{\sqrt{N+1}}$.
    
    We can express the state $\ket{\phi}_{\sfA\sfC\sfD}$ as $\ket{\phi}_{\sfA\sfC\sfD} = \sum_{i,b} \alpha_{i,b}\ket{\phi_i}_\sfA\ket{b}_\sfC \ket{\xi_{i,b}}_\sfD$, for some $\alpha_{i,b} \in \mathbb{C}$, and some normalized states $\ket{\xi_{i,b}}$ and $\ket{\phi_i}$ such that  $\ket{\phi_0}=\ket{\psi}$ and all of the $\ket{\phi_i}$ are orthogonal to each other. Then, the ideal state $\ket{\psi_{\sf{Ideal}}} = (C_{\ket{\psi}}\ket{\phi})\otimes \ket{\psi}^{\otimes N}$ can be expressed as
    \begin{equation*}
        \ket{\psi_{\sf{Ideal}}} = \sum_b \alpha_{0,b}\ket{\psi}_\sfA\ket{\psi}_{\sfB}^{\otimes N}\ket{b\oplus 1}_\sfC\ket{\xi_{0,b}}_\sfD + \sum_{i\neq 0, b} \alpha_{i,b}\ket{\phi_i}_\sfA\ket{\psi}_{\sfB}^{\otimes N}\ket{b}_\sfC\ket{\xi_{i,b}}_\sfD.
    \end{equation*}
    On the other hand, the real state $\ket{\psi_{\sf{Real}}} =
 C_\swap(\ket{\phi}\ket{\psi}^{\otimes N})$ can be expressed as
    \begin{align*}
        \ket{\psi_{\sf{Real}}} = &\sum_b \alpha_{0,b} \ket{\psi}_\sfA  \ket{\psi}_\sfB^{\otimes N}\ket{b\oplus 1}_\sfC \ket{\xi_{0,b}}_\sfD \\&+ \sum_{i\neq 0, b} \alpha_{i,b} \left(\frac{1}{\sqrt{N+1}} \ket{\phi_i, \psi}_{\sfA\sfB}\ket{b\oplus 1}_\sfC + \frac{\sqrt{N}}{\sqrt{N+1}}\ket{\phi_i^\perp}_{\sfA\sfB}\ket{b}_\sfC \right) \ket{\xi_{i,b}}\,,
    \end{align*}
    where, when writing $\ket{\phi_i, \psi}$ and $\ket{\phi_i^\perp}$, we are using the notation introduced earlier for $\ket{\chi, \psi}$ and $\ket{\chi^{\perp}}$.
    Thus we have
    \ifnum\widemargin=0
    \begin{align*}
        \norm{\ket{\psi_{\sf{Real}}} - \ket{\psi_{\sf{Ideal}}}} &\leq  \bigg\|\sum_{i\neq 0, b} \alpha_{i, b} \Big( \frac{1}{\sqrt{N+1}}\ket{\phi_i,\psi}_{\sfA\sfB}\ket{b\oplus 1}_\sfC\ket{\xi_{i,b}}_\sfD \\&\qquad\qquad\qquad+ \Big( \frac{\sqrt{N}}{\sqrt{N+1}}\ket{\phi_i^\perp} - \ket{\phi_i}\ket{\psi}^{\otimes N} \ \Big)_{\sfA\sfB}\ket{b}_\sfC \ket{\xi_{i,b}}_\sfD \Big)\bigg\| \\
        &\leq  \norm{ \sum_{i\neq 0,b} \frac{\alpha_{i,b}}{\sqrt{N+1}}\ket{\phi_i,\psi}_{\sfA\sfB}\ket{b\oplus 1}_\sfC\ket{\xi_{i,b}}_\sfD} + \norm{ \sum_{i\neq 0, b} \frac{\alpha_{i,b}}{\sqrt{N+1}}\ket{\phi_i, \psi}_{\sfA\sfB}\ket{b}_\sfC\ket{\xi_{i,b}}_\sfD} \\
        &= \frac{2}{\sqrt{N+1}}\norm{ \sum_{i\neq 0, b} \alpha_{i,b}\ket{\phi_i, \psi}_{\sfA\sfB}\ket{b}_\sfC\ket{\xi_{i,b}}_\sfD} \\
        &= \frac{2}{\sqrt{N+1}} \sqrt{\sum_{i\neq 0,b} \alpha_{i,b}^2} \leq \frac{2}{\sqrt{N+1}}
    \end{align*}
    \else
    \begin{align*}
            &\norm{\ket{\psi_{\sf{Real}}} - \ket{\psi_{\sf{Ideal}}}} \\&\leq  \Bigg\|\sum_{i\neq 0, b} \alpha_{i, b} \Big( \frac{1}{\sqrt{N+1}}\ket{\phi_i,\psi}_{\sfA\sfB}\ket{b\oplus 1}_\sfC\ket{\xi_{i,b}}_\sfD \\&\qquad\qquad\qquad+ \Big( \frac{\sqrt{N}}{\sqrt{N+1}}\ket{\phi_i^\perp} - \ket{\phi_i}\ket{\psi}^{\otimes N} \ \Big)_{\sfA\sfB}\ket{b}_\sfC \ket{\xi_{i,b}}_\sfD \Big)\Bigg\| \\
        &\leq  \norm{ \sum_{i\neq 0,b} \frac{\alpha_{i,b}}{\sqrt{N+1}}\ket{\phi_i,\psi}_{\sfA\sfB}\ket{b\oplus 1}_\sfC\ket{\xi_{i,b}}_\sfD} \nonumber\\&\qquad\qquad\qquad+ \norm{ \sum_{i\neq 0, b} \frac{\alpha_{i,b}}{\sqrt{N+1}}\ket{\phi_i, \psi}_{\sfA\sfB}\ket{b}_\sfC\ket{\xi_{i,b}}_\sfD} \\
        &= \frac{2}{\sqrt{N+1}}\norm{ \sum_{i\neq 0, b} \alpha_{i,b}\ket{\phi_i, \psi}_{\sfA\sfB}\ket{b}_\sfC\ket{\xi_{i,b}}_\sfD} \\
        &= \frac{2}{\sqrt{N+1}} \sqrt{\sum_{i\neq 0,b} \alpha_{i,b}^2} \leq \frac{2}{\sqrt{N+1}} \,,
    \end{align*}
    \fi
    where the second inequality follows from \eqref{eq:cswap}, the fact that $\frac{\sqrt{N}}{\sqrt{N+1}}\ket{\phi_i^\perp}$ is the projection of $\ket{\phi_i}$ onto $I-\Pi^{sym}_{\sfA\sfB}$. In more detail,
    \begin{align*}
        \ket{\phi_i}\ket{\psi}^{\otimes N} - \frac{\sqrt{N}}{\sqrt{N+1}}\ket{\phi_i^\perp} &= \ket{\phi_i}\ket{\psi}^{\otimes N} - (I-\Pi^{sym}_{\sfA\sfB})\ket{\phi_i}\ket{\psi}^{\otimes N} \\&= \Pi_{\sfA\sfB}^{sym}\ket{\phi_i}\ket{\psi}^{\otimes N} = \frac{1}{\sqrt{N+1}}\ket{\phi_i, \psi},
    \end{align*}
    and combined with a triangle inequality.

Together, \cref{lem:16} and \cref{lem:13} conclude the proof of \cref{thm:isometry-to-unitary}, and hence of \cref{cor:20}.
\end{proof}

\subsection{Weak simulation of the unitary oracle suffices to lift our separation results}\label{sec:weak-simulation}
In this section, we show that a weak simulation of the unitary oracle (as in Corollary \ref{cor:20}) suffices to establish the desired ``lifting'' result: a separation of $\oprs$ and $\prs$ relative to the CHRS oracle, which gives out copies of a state $\ket{\psi}$ sampled from the Haar measure (and possibly relative to some additional arbitrary unitary oracle $\mathcal{O}$), holds also relative to the unitary oracle $U_{\ket{\psi}}$, where $\ket{\psi}$ is sampled from the Haar measure (and the same unitary oracle $\mathcal{O}$)\footnote{Technically, the CHRS oracle consists of one state for each size (as described in \cref{sec:1prs_CHRS}), but the argument in this section applies just the same, since all of these states are sampled independently. The number of copies required to weakly simulate with precision $\epsilon$ is still $O(\frac{T^2}{\epsilon^2})$ where $T$ is now the total number of queries to unitaries $U_{\ket{\psi_m}}$ made by the algorithm, for states $\ket{\psi_m}$ possibly of different sizes.}. We proceed in two steps: 
\begin{itemize}
    \item[(i)] We first show that if a primitive (with a security game consisting of a single-round, e.g. a $\oprs$ or $\efi$) exists relative to the CHRS oracle (or, in fact, relative to any distribution over states that is ``global-phase invariant'', as defined in Definition~\ref{def:global-phase-invariant} below), then it also exists relative to a corresponding unitary oracle.
    \item[(ii)] Conversely, we show that $\prs$, with $\omega(\log n)$ output length, do not exist relative to the unitary oracle (induced by the CHRS oracle), since we can still carry out (a suitably modified version of) the OR lemma attack on $\prs$ that we described in Section~\ref{sec:oralce-sep}.
\end{itemize}
Now, for step (i), we start by defining the notion of a ``global-phase invariant distribution''. 
\begin{definition}
\label{def:global-phase-invariant}
    A distribution $\mathcal{D}$ over quantum states is said to be ``global-phase invariant'' if the following distribution over states is identical to $\mathcal{D}$, \emph{even up to global phases}: sample $\ket{\psi} \leftarrow \mathcal{D}$ and a uniformly random phase $\alpha$; output $\alpha \ket{\psi}$.
\end{definition}

As an example, the Haar measure is clearly global-phase invariant. However, for example, a distribution that outputs $\ket{0}$ with probability $\frac12$ and $\ket{1}$ with probability $\frac12$ is not, since almost all states of the form $\alpha \ket{0}$ are different from $\ket{0}$, when the global phase $\alpha$ is taken into consideration. It might seem strange to consider global phases, but the point is that some of the distributions we are considering are over \emph{unitaries} of the form $U_{\alpha \ket{\psi}}$, for which the ``global'' phase $\alpha$ gives rise to unitaries that are actually distinct. We remark that the notion of global-phase invariance is reminiscent of the notion of ``phase-invariance'' introduced by Zhandry in \cite{zhandry2024space}. The crucial difference is that here we consider a \emph{global} phase, rather than a \emph{relative} phase. 

We will make use of the following corollary of our previous weak simulation result from Section~\ref{sec:unitaryoracle0}.
\begin{corollary}
\label{cor:21}
For an $n$-qubit state $\ket{\psi}$, define the $(n+1)$-qubit state $\ket{\psi'} = \ket{\psi} \otimes \ket{1}$. Let $U_{\ket{\psi'}}$ be the corresponding $(n+1)$-qubit unitary defined in \cref{sec:unitaryoracle0}. Let $\epsilon >0$, and $T\in \mathbb{N}$. Let $\xi$ any map from $n$-qubit states to $m$-qubit states such that $\xi(\ket{\psi}) = \xi(\alpha\ket{\psi})$ for all $\alpha$ such that $|\alpha| = 1$. Then, let $\mathcal{D}$ be any global-phase invariant distribution over $n$-qubit states. For any $T$-query oracle algorithm $\calA^{(\cdot)}$ taking as input an $m$-qubit state, there is an algorithm $\wt{\cal{A}}$ such that:
$$\Big\| \mathbb{E}_{\ket{\psi} \leftarrow \mathcal{D}}\calA^{U_{\ket{\psi'}}}\Big( \xi(\ket{\psi}) \Big)  - \mathbb{E}_{\ket{\psi} \leftarrow \mathcal{D}} \wt{\cal{A}}\Big(\ket{\psi}^{\otimes O(\frac{T^2}{\epsilon^2})},\xi(\ket{\psi})\Big)\Big\| \leq \epsilon \,.$$
\end{corollary}
In the above corollary, the function $\xi$ captures the fact that the input to $\mathcal{A}$ can depend arbitrarily on $\ket{\psi}$. The outputs of the two algorithms are mixed states (and the norm is the trace norm).

\begin{proof}[Proof of \cref{cor:21}]
The proof is straightforward, and is a consequence of \cref{cor:20}. We have the following:
\begin{align}
&\Big\| \mathbb{E}_{\ket{\psi} \leftarrow \mathcal{D}}\calA^{U_{\ket{\psi'}}}\Big( \xi(\ket{\psi}) \Big)  - \mathbb{E}_{\ket{\psi} \leftarrow \mathcal{D}} \wt{\cal{A}}\Big(\ket{\psi}^{\otimes O(\frac{T^2}{\epsilon^2})},\xi(\ket{\psi})\Big)\Big\| \nonumber\\
=& \Big\|   \mathbb{E}_{\substack{\ket{\psi} \leftarrow \mathcal{D} \\ \alpha: |\alpha| = 1} } \calA^{U_{\alpha\ket{\psi'}}}\Big( \xi(\alpha\ket{\psi}) \Big) -  \mathbb{E}_{\substack{\ket{\psi} \leftarrow \mathcal{D} \\ \alpha: |\alpha| = 1} }\wt{\cal{A}}\Big(\ket{\psi}^{\otimes O(\frac{T^2}{\epsilon^2})}, \xi(\alpha\ket{\psi})\Big)   \Big\| \label{eq:200}\\
=&\Big\|  \mathbb{E}_{\ket{\psi} \leftarrow \mathcal{D}} \Big( \mathbb{E}_{ \alpha: |\alpha| = 1} \calA^{U_{\alpha\ket{\psi'}}}\Big( \xi(\ket{\psi}) \Big) - \wt{\cal{A}}\Big(\ket{\psi}^{\otimes O(\frac{T^2}{\epsilon^2})}, \xi(\ket{\psi})\Big) \Big)  \Big\| \nonumber\\
\leq& \, \mathbb{E}_{\ket{\psi} \leftarrow \mathcal{D}} \Big\|    \mathbb{E}_{ \alpha: |\alpha| = 1} \calA^{U_{\alpha\ket{\psi'}}}\Big( \xi(\ket{\psi}) \Big) - \wt{\cal{A}}\Big(\ket{\psi}^{\otimes O(\frac{T^2}{\epsilon^2})}, \xi(\ket{\psi})\Big)  \Big\| \nonumber\\
\leq& \epsilon \,, \nonumber
\end{align}
where the first equality follows from the fact that $\calD$ is global-phase invariant, and the second equality just interchanges the order of expectations and uses the fact that $\xi(\alpha\ket{\psi})=\xi(\ket{\psi})$ for all $\alpha, \ket{\psi}$. The first inequality is an application of~\cref{cor:20}.

\end{proof}


We are now ready to prove the first half of our lifting result (step (i))\footnote{The following theorem involves distributions over oracles. However, one can identify fixed oracles relative to which the same separations hold, by a similar argument as in \ifnum\shortver=0\cref{ssec:ampli-random}\else\cref{sec:borel-cantelli}\fi.}.
\begin{theorem}
\label{thm:16}
    Let $\calP$ be a primitive with a security game consisting of a single round. Suppose $\mathcal{P}$ exists relative to an oracle $\mathcal{O}$ that provides copies of a (fixed) state $\ket{\psi} \leftarrow \mathcal{D}$, where $\mathcal{D}$ is a ``global-phase invariant'' distribution. Then, $\mathcal{P}$ also exists relative to an oracle $\mathcal{U}$ that applies $U_{\ket{\psi'}}$, for a state $\ket{\psi} \leftarrow \mathcal{D}$, where $\ket{\psi'} = \ket{\psi}\ket{1}$.
\end{theorem}

\begin{proof}
Let $C$ be a secure construction of $\mathcal{P}$ relative to $\mathcal{O}$.

First, notice that any algorithm $\mathcal{A}$ that queries $\mathcal{O}$ can be replicated perfectly by querying the corresponding unitary oracle (making queries on $\ket{0}$ each time a copy is required). Thus, if a primitive exists relative to $\mathcal{O}$, any guarantee pertaining ``honest'' algorithms will hold verbatim (e.g.\ any ``correctness'' guarantee).

What about security? Suppose for a contradiction there is an adversary $\mathsf{Adv}^{(\cdot)}$ that breaks security of $C$ relative to $\mathcal{U}$.

Let $\langle \mathsf{Ch}^{U_{\ket{\psi'}}}, \mathsf{Adv}^{U_{\ket{\psi'}}} \rangle$ denote the interaction between the challenger $\mathsf{Ch}$ and adversary $\mathsf{Adv}$ in the security game for construction $C$, when $\mathcal{U}$ applies $U_{\ket{\psi'}}$ for some $\ket{\psi}$. Here $\mathsf{Ch}^{U_{\ket{\psi'}}}$ is identical to the challenger relative to the CHRS oracle (it simply queries $U_{\ket{\psi'}}$ whenever the original challenger would have requested a copy of $\ket{\psi}$).

Assume for simplicity that the security game has a ``threshold'' of $\frac12$ (this does not change the argument), i.e.\ security requires that no bounded adversary can win with probability non-negligibly greater than $\frac12$. Then, by the hypothesis that $\mathsf{Adv}$ breaks security of $C$, we have that $$\mathbb{E}_{\ket{\psi} \leftarrow \mathcal{D}}\Pr[\langle \mathsf{Ch}^{U_{\ket{\psi'}}} , \mathsf{Adv}^{U_{\ket{\psi'}}} \rangle = 1] = \frac12 + \textsf{non-negl}(\lambda) \,,$$
where $\lambda$ is the security parameter.  Now, let $T(\lambda)$ be the number of queries made by $\mathsf{Adv}$, and let $\epsilon(n)$ be a sufficiently small inverse polynomial in $\lambda$. 

Now, by \cref{cor:21}, there exists a simulator $\wt{\mathsf{Adv}}$ that only uses $t = O(\frac{T^2}{\epsilon^2})$ copies of $\ket{\psi}$, and satisfies 
$$\Big\| \mathbb{E}_{\ket{\psi} \leftarrow \mathcal{D}}\mathsf{Adv}^{U_{\ket{\psi'}}}\Big( \xi(\ket{\psi}) \Big)  - \mathbb{E}_{\ket{\psi} \leftarrow \mathcal{D}}\wt{\mathsf{Adv}}\Big( \ket{\psi}^{\otimes t}, \xi(\ket{\psi})\Big)\Big\| \leq \epsilon \,,$$
where $\xi$ is an arbitrary function as in \cref{cor:21}.
When $\epsilon$ is taken to be sufficiently small, and $\xi$ is taken to be precisely the challenger's message, we have
$$\mathbb{E}_{\ket{\psi} \leftarrow \mathcal{D}}\Pr[\langle \mathsf{Ch}^{U_{\ket{\psi'}}} , \wt{\mathsf{Adv}}\big(\ket{\psi}^{\otimes t}\big) \rangle = 1] = \frac12 + \textsf{non-negl}(\lambda) \,,$$
for some possibly different non-negligible function.

Finally, recall that $\mathsf{Ch}^{U_{\ket{\psi'}}}$ is identical to the challenger for the original construction $C$ relative to $\mathcal{O}$. So, $\wt{\mathsf{Adv}}$ breaks the security of $C$, which is a contradiction.

\end{proof}

Moving on to step (ii), we will show that $\prs$ do not exist relative to the unitary oracle (induced by the CHRS oracle), since the OR lemma attack can be lifted. 
\begin{theorem}
\label{thm:prs-nonexistence}
  $\prs$ with $\omega(\log n)$ output length do not exist relative to the paramterized unitary oracle induced by the CHRS oracle (as described in Section~\ref{sec:unitaryoracle0}).
\end{theorem}
One might think that a generic lifting theorem, such as Theorem~\ref{thm:16}, which lifts any ``existence'' results from a state to a unitary model, might also hold for lifting impossibility results (and that, as a consequence, one need not think about lifting a specific attack). However, this is not the case due to the following important subtlety. The natural argument would go as follows. By hypothesis, the primitive does not exist in the state model. Now, consider any candidate construction in the unitary model. We would like to assert that the attack in the state oracle can be lifted to the unitary model. Can we do so? Certainly each copy of the oracle state used by the attacker can be simulated perfectly with one query to the unitary oracle. However, the attacker is only guaranteed to break constructions in the state model (this may in fact even be a syntactic requirement). At first, this does not seem like an important issue because one can obtain a construction relative to the state oracle by simulating the construction relative to the unitary oracle. However, since the simulation is not perfect, the resulting construction may fail to satisfy correctness requirements of the primitive. Thus, we are no longer guaranteed the existence of an attacker that breaks this (invalid) construction.\footnote{In an earlier version of our paper, we identified this subtlety, and proved a generic lifting theorem for non-existence of primitives that do \emph{not} have a correctness condition. However, we erroneously claimed that $\prs$ fall into this category, i.e.\ they do not have any correctness condition. This is false, since PRS do have a correctness condition, namely that the output of the generator should be a pure state. We thank Mark Zhandry for pointing out this error to us.} 

In the rest of this section, we will show how to lift the OR-lemma attack on PRS in the CHRS model to the unitary model, i.e.\ prove Theorem~\ref{thm:prs-nonexistence}. We show that, while in the CHRS model we can only simulate a ``global-phase twirled'' version of the corresponding unitary oracle, this is enough to lift the attack.


The attack in the unitary model is a slight modification of the attack in~\cref{sec:oracle-sep}. Let us recall the attack in~\cref{sec:oracle-sep} first. Let $\gen_k$ be the generation unitary of the $\prs$ when the secret key is $k$, where $\gen_k$ is meant to act on multiple copies of the CHRS state $\ket{\psi}$. For each $k$, the adversary from \cref{sec:oralce-sep} takes sufficiently many copies of $\ket{\psi}$, applies $\gen_k$ on them (thereby generating sufficiently many copies of the PRS state on seed $k$). Then the adversary applies swap tests between the generated copies and the states received from the challenger (which are either copies of a PRS state or copies of a Haar random state). 

Now, in the unitary oracle model, let $U_{\ket{\psi}}$ be the unitary oracle. The generation algorithm $\gen_k$ makes queries to the unitary oracle $U_{\ket{\psi}}$, so we will denote this as $\gen^{U_{\psi}}_k$. According to~\cref{thm:isometry-to-unitary}, given polynomial many copies of $\ket{\psi}$, there is an efficient unitary $\widetilde{\gen}_k$ that takes as input $\frac{4T^3n} {\epsilon^2}$ copies of $\ket{\psi}$, and outputs a state that is $\epsilon$-close to the ``global-phase twirled'' state $\rho_{k,\psi} = \E_{\alpha} \gen_k^{U_{\alpha\ket{\psi}}}\ket{0}\bra{0}(\gen_k^{U_{\alpha\ket{\psi}}})^\dagger$. The key observation is that, while $\rho_{k,\psi}$ is not in general close to $\gen^{U_{\ket{\psi}}}_k \ket{0}$, it has an inverse polynomial overlap with  $\gen^{U_{\ket{\psi}}}_k\ket{0}\bra{0} (\gen^{U_{\ket{\psi}}}_k)^\dagger$.

\begin{lemma}\label{lem:state-simulation-overlap}
    Assume that $\gen_k$ makes at most $T$ queries to $U_{\psi}$, where $\ket{\psi}$ is the $m$-qubit CHRS state, then except with probability at most $\exp(-\Omega(2^m/T^2))$ over $\ket{\psi}$ sampled from the Haar measure, $\braket{0|(\gen^{U_{\ket{\psi}}}_k)^{\dagger} \rho_{k,\psi}\gen^{U_{\ket{\psi}}}_k|0} \geq \frac{1}{3T}$.
\end{lemma}
\begin{proof}
    First, we will show the average of the overlap $\braket{0|\gen_k^{\dagger} \rho_{k,\psi} \gen_k|0}$ is at least $\frac{1}{2T+1}$. In fact,
    \begin{align*}
        \E_{\psi} \braket{0|(\gen^{U_{\ket{\psi}}}_k)^{\dagger} \rho_{k, \psi}\gen_k^{U_{\ket{\psi}}}|0} &= \E_{\psi, \alpha} \braket{0|(\gen_k^{U_{\alpha\ket{\psi}}})^{\dagger}\rho_{k,\psi}\gen_k^{U_{\alpha\ket{\psi}}}|0} \\
        &= \E_{\psi} \tr \rho_{k, \psi}^2 \\
        &\geq \E_\psi \frac{1}{2T+1} = \frac{1}{2T+1},
    \end{align*}
    where the last inequality stems from the fact that $\rho_{k, \psi}$ is of rank at most $2T+1$ (Recall the proof in~\cref{sec:unitaryoracle1} that the counting register $\sfC$ ranges from $-T$ to $T$ so there are at most $2T+1$ branches in the purification of $\rho_{k, \psi}$) and Cauchy-Schwarz inequality.

    Thus, according to~\cref{lem:levy-lemma}, except with probability at most $\exp(-O(2^m/T^2))$\anote{Seems like we can pick $n = \log T$}\bnote{But it seems $T$ can be arbitrarily large polynomial}, the overlap $\braket{0|\gen_k^{\dagger}|\rho_{k,\psi}|\gen_k|0}$ is at least $1/3T$.
\end{proof}

Now, let $\widetilde{\gen}_k$ be an efficient unitary that takes as input $L = 4 T^3n/\epsilon^2$ copies of $\ket{\psi}$ and outputs a state that is $\epsilon$-close to $\rho_{k, \psi}$ (such a unitary exists by Lemma~\cref{lem:13}).

To lift the OR lemma attack, we define the new OR lemma projectors $\widetilde{\Pi}_k$ as follows:
\begin{align*}
    \widetilde{\Pi}_k = \left(\mathop{\bigotimes}_{i=1}^{Tn}\Big((\widetilde{\gen}_k^\dagger)_{\mathsf{A}_i \mathsf{B}_i} \otimes \I_{\mathsf{C}_i} \Big)\right) \Big((\Pi^{sym}_{\geq Tn/2+n/6})_{\mathsf{AC}} \otimes \I_{\mathsf{B}} \Big) \left(\mathop{\bigotimes}_{i=1}^{Tn}\Big((\widetilde{\gen}_k)_{\mathsf{A}_i\mathsf{B}_i} \otimes \I_{\mathsf{C}_i}\Big)\right)\,.
\end{align*}
Here, $\mathsf{A}_i$ is the $i$-th sub-register of $\mathsf{A}$, and similarly for $\mathsf{B}_i$ and $\mathsf{C}_i$. $\mathsf{C}_i$ contains the $i$-th copy of the challenge state, $\mathsf{A}_i$ is the output register of $\widetilde{\gen}_k$ (which is of the same length as $\mathsf{C}_i$), and $\mathsf{B}_i$ is an auxiliary register. $\Pi^{sym}_{\geq Tn/2+n/6}$ is the projector onto the subspace that is spanned by states that lie in the symmetric subspace on at least $Tn/2+n/6$ of all $Tn$ pairs of registers $\sfA_i$ and $\sfC_i$. 

These projectors replace the projectors $\Pi_k$ from~\cref{eq:projectors_attack}. In words, the projector $\widetilde{\Pi}_k$ corresponds to performing $Tn$ copies of a SWAP test between $Tn$ copies of a challenge state (PRS or Haar) state and $Tn$ copies of states produced by the simulated generation procedure $\widetilde{\gen}_k$. The projector ``accepts'' if slightly more than half of the SWAP tests accept.

Then, by~\cref{lem:state-simulation-overlap}, $(\ket{\phi_k}^{\otimes Tn})_{\mathsf{C}} \otimes (\ket{\psi}^{\otimes LTn})_{\mathsf{AB}}$\footnote{Register $\mathsf{AB}$ may contain additional auxiliary registers initialized in the zero state, but we omit writing them for simplicity.} has constant overlap with $\widetilde{\Pi}_k$, while $(\ket{\phi}^{\otimes Tn})_{\mathsf{C}} \otimes (\ket{\psi}^{\otimes LTn})_{\mathsf{AB}}$ has exponentially small overlap with $\widetilde{\Pi}_k$, with overwhelming probability over $\ket{\psi}$. We can thus run the OR lemma algorithm to break the $\prs$ construction.

\fi
\ifnum\llncs=1
    \bibliographystyle{splncs04}
\fi
\ifnum \preprint=1
    \bibliographystyle{alphaabbrurldoieprint}
\fi
\bibliography{ref,main}
\ifnum\proceeding=0
\appendix
\ifnum\shortver=1

\fi
\ifnum\shortver=0
    \section{Quantum OR lemma algorithm using a QPSPACE machine}
\label{sec:OR_leamma_unitary}

In this appendix, we justify the claim made in Remark~\ref{rem:or_lemma_qpspace} that we can implement the quantum OR lemma algorithm using a unitary QPSPACE machine, i.e.\ using a uniform family of unitary circuits, indexed by $n$ (the number of qubits of the state $\rho$), that utilizes only $\text{poly}(n)$ qubits of space. Figure~\ref{fig:alg1} describes the quantum OR lemma algorithm from \cite[Algorithm 1]{HLM17}.

\begin{figure}[!htb]
\centering
\begin{mdframed}[userdefinedwidth=0.8\textwidth, align=center]
The quantum OR algorithm, taken verbatim from \cite[Algorithm 1]{HLM17}:
\begin{enumerate}
\item Create the state $\rho \otimes \ket{0}\bra{0}^{\otimes m}$. 
\item Repeat $N$ times or until the algorithm accepts:
\begin{enumerate}
    \item Perform the projective measurement $\{\Pi,I-\Pi\}$. If the first result is returned, accept (and terminate).
    \item Perform the projective measurement $\{\Delta,I-\Delta\}$. If the second result is returned, accept (and terminate).
\end{enumerate}
\item Reject.
\end{enumerate}
\end{mdframed}
\caption{Algorithm 1}
\label{fig:alg1}
\end{figure}

The definitions of the projectors $\Pi$, $\Delta$, and $m$ are omitted; the only relevant detail (which is straightforward to verify) is that, in our setting, these measurements can be implemented using a polynomial-space quantum circuit. 

Note that the algorithm above uses measurements. We wish to use unitary gates only. The simplest approach to deal with this is to use delayed measurements: applying a CNOT gate to a fresh qubit, and measuring only the resulting qubit at the very end. Unfortunately, since the number of measurements is exponential, this requires exponential space, for all the intermediate results.

We show how the algorithm can be implemented coherently by a unitary QPSPACE machine, by introducing two additional algorithms, both of which have the same acceptance probability.

In Algorithm 1, the algorithm may accept and terminate early in steps 2(a) and 2(b). In Algorithm 2 below, we simplify the algorithm, without changing the worst-case running time.  The only difference is that there is no early termination.

\begin{figure}[!htb]
\centering
\begin{mdframed}
\begin{enumerate}
\item Create the state $\rho \otimes \ket{0}\bra{0}^{\otimes m}\otimes \ket{0}\bra{0}\otimes \ket{0}\bra{0}$, and initialize an $n$ qubit counter to $\ket{0}$. 
\item Repeat $N$ times or until the algorithm accepts:
\begin{enumerate}
    \item Apply the unitary $\Pi \otimes X \otimes I+ (I-\Pi)\otimes I \otimes I$. 
    \item Measure the third register, and increment the counter if the output is 1.
    \item Apply the unitary $\Delta \otimes I \otimes X + (I-\Delta) \otimes I \otimes I $.
    \item Measure the fourth register, and increment the counter if the output is 1.
\end{enumerate}
\item Measure the counter and accept if the outcome is 0.
\end{enumerate}
\end{mdframed}
\caption{Algorithm 2}
\label{fig:alg2}
\end{figure}

Algorithm 2 lends itself to a natural version, in which all the steps are unitary, except a measurement in the very last step, as depicted in Algorithm 3.

\begin{figure}[!htb]
\centering
\begin{mdframed}
\begin{enumerate}
\item Create the state $\rho \otimes \ket{0}\bra{0}^{\otimes m}\otimes \ket{0}\bra{0}\otimes \ket{0}\bra{0}$, and initialize an $n$ qubit counter to $\ket{0}$. 
\item Repeat $N$ times or until the algorithm accepts:
\begin{enumerate}
    \item Apply the unitary $\Pi \otimes X \otimes I+ (I-\Pi)\otimes I \otimes I$. 
    \item Apply a Controlled-Increment between the third register and the counter.
    \item Apply the unitary $\Delta \otimes I \otimes X + (I-\Delta) \otimes I \otimes I $.
    \item \item Apply a Controlled-Increment between the fourth register and the counter.
\end{enumerate}
\item Measure the counter and accept if the outcome is 0.
\end{enumerate}
\end{mdframed}
\caption{Algorithm 3}
\label{fig:alg3}
\end{figure}

A direct calculation shows that the acceptance probabilities of Algorithms 2 and 3 are equal. More specifically, let $p_i$ denote the probability that the counter is $0$ at the end of the $i$th iteration in Algorithm 2.  Additionaly, let $\ket{\psi_i}$ be the state at the end of the $i$th iteration of the loop in algorithm 3, and $\ket{\psi_i} =a_i \ket{\alpha_i}\otimes \ket{0} \otimes \ket{0} \otimes \ket{0} + \ket{\beta_i}$, where the last 3 registers of $\ket{\beta_i}$ are orthogonal to 000. It is easy to prove by induction that $p_i=|a_i|^2$.

In order to make the entire algorithm unitary, the measurement in the last step in Algorithm 3 is omitted. Of course, this measurement can be done directly by the $\BQP$ machine that breaks the~$\prs$.
\fi
\section{Proofs of \cref{lem:haar-coeff-estimate} and \cref{lem:qubit-amplify-ineq}}\label{sec:proof-sandwich-lemma}
\begin{proof}[Proof of~\cref{lem:haar-coeff-estimate}]
    First, notice that one can sample a Haar random state by sampling $\tilde{\ket{\psi}} = \alpha \ket{0} \ket{\psi_1} + \sqrt{1-\alpha^2}\ket{1} \ket{\psi_2}$, where $\ket{\psi_1}$ and $\ket{\psi_2}$ are Haar random $m-1$ qubit states, and $\alpha$ is sampled according to the marginal distribution of $|(\bra{0} \otimes \I)\ket{\psi}|$ where $\ket{\psi}$ is sampled from the Haar distribution. Denote the latter distribution by $\mathcal{D}_0$. For convenience, in the rest of the section, we use the notation $\braket{0_1|\psi} = (\bra{0} \otimes \I)\ket{\psi}$. The fact that $\ket{\tilde{\psi}}$ has the same distribution as a Haar random state follows from the unitary invariance of the Haar measure. More precisely, one can see this as follows, where for $(m-1)$-qubit unitaries $U_1$ and $U_2$ we write $C_{U_1, U_2} = \ket{0}\bra{0} \otimes U_1 + \ket{1}\bra{1} \otimes U_2$:
    \begin{align}
        \E_{\ket{\psi} \gets \mu_{2^m}} \psi^{\ot r} 
        &=  \E_{\substack{U_1, U_2 \gets SU(2^{m-1}) \\\ket{\psi} \gets \mu_{2^m}}} (C_{U_1, U_2} \,\psi \,C_{U_1, U_2}^{\dagger})^{\otimes r} \nonumber\\
        &=  \E_{\substack{U_1, U_2 \gets SU(2^{m-1}) \\\ket{\psi} \gets \mu_{2^m}, \\ \alpha, \ket{\psi_1}, \ket{\psi_2}, \ket{\tilde{\psi}}\,: \,\ket{\psi} = \alpha \ket{0} \ket{\psi_1} + \sqrt{1-\alpha^2}\ket{1} \ket{\psi_2} \,, \nonumber\\ \ket{\tilde{\psi}} =  \alpha \ket{0} U_1\ket{\psi_1} + \sqrt{1-\alpha^2}\ket{1} U_2 \ket{\psi_2}}} \,\tilde{\psi}^{\otimes r} \nonumber \\
        &= \E_{\substack{\alpha \gets \mathcal{D}_0, \ket{\psi_1}, \ket{\psi_2} \gets \mu_{2^{m-1}} \,, \\ \ket{\psi} = \alpha \ket{0} \ket{\psi_1} + \sqrt{1-\alpha^2}\ket{1} \ket{\psi_2}}}\psi^{\ot r} \,, \label{eq:9}
    \end{align}
where the first equality is by the unitary invariance of the Haar measure.

     Now, define a map $F$ such that, for any state $\ket{\psi} = \alpha \ket{0} \ket{\psi_1} + \beta \ket{1}\ket{\psi_2}$, with $\alpha, \beta \in \mathbb{R}^+$,  $F(\ket{\psi}) = \frac{1}{\sqrt{2}} \ket{0} \ket{\psi_1} + \frac{1}{\sqrt{2}} \ket{1}\ket{\psi_2}$. Then, $F(\ket{\psi})$ is well defined on all pure states, and, by \cref{eq:9}, the distribution of $F(\ket{\psi})$ for a Haar random $\ket{\psi}$ is identical to the distribution of $\ket{\psi'} = \frac{1}{\sqrt{2}} \ket{0} \ket{\psi_1} + \frac{1}{\sqrt{2}} \ket{1}\ket{\psi_2}$ for Haar random $\ket{\psi_1}$ and $\ket{\psi_2}$. It follows that
     \if\widemargin=0
    \begin{align}
        \norm{\E_{\ket{\psi} \gets \mu_{2^m}} \psi^{\ot r} - \E_{\substack{\ket{\psi_1}, \ket{\psi_2} \gets \mu_{2^{m-1}} \\  \ket{\psi'} = \frac{1}{\sqrt{2}} \ket{0} \ket{\psi_1} + \frac{1}{\sqrt{2}} \ket{1} \ket{\psi_2}}} \psi'^{\ot r}}  &= \norm{\E_{\substack{\alpha \gets \mathcal{D}_0, \ket{\psi_1}, \ket{\psi_2} \gets \mu_{2^{m-1}} \,, \\ \ket{\psi} = \alpha \ket{0} \ket{\psi_1} + \sqrt{1-\alpha^2}\ket{1} \ket{\psi_2}}}\psi^{\ot r}  - \E_{\substack{\ket{\psi_1}, \ket{\psi_2} \gets \mu_{2^{m-1}} \\  \ket{\psi'} = \frac{1}{\sqrt{2}} \ket{0} \ket{\psi_1} + \frac{1}{\sqrt{2}} \ket{1} \ket{\psi_2}}} \psi'^{\ot r}} \nonumber\\
        &= \norm{\E_{\substack{\alpha \gets \mathcal{D}_0, \ket{\psi_1}, \ket{\psi_2} \gets \mu_{2^{m-1}} \,, \\ \ket{\psi} = \alpha \ket{0} \ket{\psi_1} + \sqrt{1-\alpha^2}\ket{1} \ket{\psi_2}}} \big(\psi^{\ot r} - F(\psi)^{\ot r} \big)} \nonumber\\
        & = \norm{\E_{\ket{\psi}\gets \mu_{2^m} } \big(\psi^{\ot r} - F(\psi)^{\ot r}\big)} \nonumber\\
        &\leq \E_{\ket{\psi}\gets \mu_{2^m} }  \norm{\psi^{\ot r} - F(\psi)^{\ot r}} \nonumber \\
        &\leq r \E_{\ket{\psi}\gets \mu_{2^m} }  \norm{\psi - F(\psi)} \,, \label{eq:ineq-psi-psi-prime}
    \end{align}
    \else
    \begin{align}
        &\norm{\E_{\ket{\psi} \gets \mu_{2^m}} \psi^{\ot r} - \E_{\substack{\ket{\psi_1}, \ket{\psi_2} \gets \mu_{2^{m-1}} \\  \ket{\psi'} = \frac{1}{\sqrt{2}} \ket{0} \ket{\psi_1} + \frac{1}{\sqrt{2}} \ket{1} \ket{\psi_2}}} \psi'^{\ot r}}  \nonumber\\&= \norm{\E_{\substack{\alpha \gets \mathcal{D}_0, \ket{\psi_1}, \ket{\psi_2} \gets \mu_{2^{m-1}} \,, \\ \ket{\psi} = \alpha \ket{0} \ket{\psi_1} + \sqrt{1-\alpha^2}\ket{1} \ket{\psi_2}}}\psi^{\ot r}  - \E_{\substack{\ket{\psi_1}, \ket{\psi_2} \gets \mu_{2^{m-1}} \\  \ket{\psi'} = \frac{1}{\sqrt{2}} \ket{0} \ket{\psi_1} + \frac{1}{\sqrt{2}} \ket{1} \ket{\psi_2}}} \psi'^{\ot r}} \nonumber\\
        &= \norm{\E_{\substack{\alpha \gets \mathcal{D}_0, \ket{\psi_1}, \ket{\psi_2} \gets \mu_{2^{m-1}} \,, \\ \ket{\psi} = \alpha \ket{0} \ket{\psi_1} + \sqrt{1-\alpha^2}\ket{1} \ket{\psi_2}}} \big(\psi^{\ot r} - F(\psi)^{\ot r} \big)} \nonumber\\
        & = \norm{\E_{\ket{\psi}\gets \mu_{2^m} } \big(\psi^{\ot r} - F(\psi)^{\ot r}\big)} \nonumber\\
        &\leq \E_{\ket{\psi}\gets \mu_{2^m} }  \norm{\psi^{\ot r} - F(\psi)^{\ot r}} \nonumber \\
        &\leq r \E_{\ket{\psi}\gets \mu_{2^m} }  \norm{\psi - F(\psi)} \,, \label{eq:ineq-psi-psi-prime}
    \end{align}
    \fi
    where the last line holds due to the triangle inequality and properties of the trace distance. So, to prove the lemma, it is enough to prove that
    \begin{equation*}
    \begin{split}
        \E_{\ket{\psi}\gets \mu_{2^m} }  \norm{\psi - F(\psi)} \leq \frac{80 \sqrt{m}}{2^{m/2}}
    \end{split}
    \end{equation*}
    Notice that, letting $\ket{\psi} = \alpha \ket{0} \ket{\psi_1} + \sqrt{1-\alpha^2}\ket{1} \ket{\psi_2}$, for $\alpha \geq 0$, and denoting $\beta = \sqrt{1-\alpha^2}$, we have
    \if\widemargin=0
    \begin{align}
        \norm{\psi - F(\psi)} &\leq \abs{\alpha^2 - \frac{1}{2}}\norm{\ket{\psi_1}\bra{\psi_1}} + \abs{\beta^2 - \frac{1}{2}}\norm{\ket{\psi_2}\bra{\psi_2}} + \abs{\alpha\beta - \frac{1}{2}} \norm{\ket{\psi_1}\bra{\psi_2}} + \abs{\alpha\beta - \frac{1}{2}} \norm{\ket{\psi_2}\bra{\psi_1}} \nonumber\\
        &= \abs{\alpha^2 - \frac{1}{2}} +\abs{\beta^2 - \frac{1}{2}} + 2\abs{\alpha\beta - \frac{1}{2}} \nonumber\\
        &\leq 4 \abs{\alpha^2 - \frac{1}{2}} \label{eq:alpha-norm-bound}
    \end{align}
    \else
    \begin{align}
        \norm{\psi - F(\psi)} &\leq \abs{\alpha^2 - \frac{1}{2}}\norm{\ket{\psi_1}\bra{\psi_1}} + \abs{\beta^2 - \frac{1}{2}}\norm{\ket{\psi_2}\bra{\psi_2}}  \nonumber\\&\quad+\abs{\alpha\beta - \frac{1}{2}} \norm{\ket{\psi_1}\bra{\psi_2}} + \abs{\alpha\beta - \frac{1}{2}} \norm{\ket{\psi_2}\bra{\psi_1}} \nonumber\\
        &= \abs{\alpha^2 - \frac{1}{2}} +\abs{\beta^2 - \frac{1}{2}} + 2\abs{\alpha\beta - \frac{1}{2}} \nonumber\\
        &\leq 4 \abs{\alpha^2 - \frac{1}{2}} \label{eq:alpha-norm-bound}
    \end{align}
    \fi
    So it is enough of us to bound $\E_{\alpha \gets \mathcal{D}_0} \abs{\alpha^2 - \frac{1}{2}}$.
    Consider the function $f: U(d) \to \R$ such that $f(\ket{\psi}) = \Vert\braket{0_1 | \psi}\Vert^2$, 
    where recall that we denote $\braket{0_1 | \psi} =(\bra{0}\otimes I )\ket{\psi}$.  
    $f$ is 2-Lipschitz, because for any two states $\ket{\psi_1}$ and $\ket{\psi_2}$, we have
    \ifnum\widemargin=0
    \begin{align*}
        |f(\ket{\psi_1})-f(\ket{\psi_2})| &= \Big| \Vert\braket{0_1|\psi_1}\Vert^2 - \Vert\braket{0_1|\psi_2}\Vert^2 \Big| \\
        &\leq \norm{\braket{0_1|\psi_1}} \cdot \Big| \norm{\braket{0_1|\psi_1}} - \norm{\braket{0_1|\psi_2}}\Big| + \norm{\braket{0_1|\psi_2}}  \cdot \Big| \norm{\braket{0_1|\psi_1}} - \norm{\braket{0_1|\psi_2}} \Big| \\
        &\leq 2 \Big| \norm{\braket{0_1|\psi_1}} - \norm{\braket{0_1|\psi_2}} \Big| \leq 2 \norm{\ket{\psi_1}-\ket{\psi_2}}.
    \end{align*}
    \else    
    \begin{align*}
        |f(\ket{\psi_1})-f(\ket{\psi_2})| &= \Big| \Vert\braket{0_1|\psi_1}\Vert^2 - \Vert\braket{0_1|\psi_2}\Vert^2 \Big| \\
        &\leq \norm{\braket{0_1|\psi_1}} \cdot \Big| \norm{\braket{0_1|\psi_1}} - \norm{\braket{0_1|\psi_2}}\Big| \\ &\quad+\norm{\braket{0_1|\psi_2}}  \cdot \Big| \norm{\braket{0_1|\psi_1}} - \norm{\braket{0_1|\psi_2}} \Big| \\
        &\leq 2 \Big| \norm{\braket{0_1|\psi_1}} - \norm{\braket{0_1|\psi_2}} \Big| \leq 2 \norm{\ket{\psi_1}-\ket{\psi_2}}.
    \end{align*}
    \fi
    Thus, using L\'evy's lemma (\cref{lem:levy-lemma}), we have
    \begin{equation*}
           \Pr_{\ket{\psi}\gets \mu_{2^m}} \Big[|f(\psi) - \E_{\ket{\psi} \gets \mu_{2^m}} f(\ket{\psi})| \geq \delta\Big] \leq 4 \exp \parens{-\frac{2^m \delta^2}{18\pi^3}}
    \end{equation*}
    Let $\delta = 18 \frac{\sqrt{m}}{2^{m/2}}$. Then, since $\E_{\ket{\psi}\gets \mu_{2^m}} f(\ket{\psi}) = 1/2$, we have
    \begin{align*}
        \E_{\ket{\psi} \gets \mu_{2^m}} \Big| f(\ket{\psi}) - 1/2\Big| &\leq \frac{1}{2} \Pr_{\ket{\psi} \gets \mu_{2^m}} \bigg(\Big|f(\ket{\psi}) - \E_{\ket{\psi} \gets \mu_{2^m}} f(\ket{\psi})\Big|\geq \delta\bigg) + \delta \\
        &= \frac{1}{2} \cdot 4 \cdot \exp\parens{-\frac{2^m\delta^2}{18\pi^3}} + \delta \\
        &\leq 2 \exp(-m/2) + \frac{18\sqrt{m}}{2^{m/2}} \\
        &\leq \frac{2}{2^{m/2}} + \frac{18\sqrt{m}}{2^{m/2}} \\
        &\leq \frac{20\sqrt{m}}{2^{m/2}}.
    \end{align*}
Combining this with \cref{eq:ineq-psi-psi-prime} and \eqref{eq:alpha-norm-bound} gives the desired conclusion. 

\end{proof}

\begin{proof}[Proof of~\cref{lem:qubit-amplify-ineq}]
    Let $A = \ket{0}\bra{0}_1 \otimes A_{00} + \ket{0}\bra{1}_1 \otimes A_{01} + \ket{1}\bra{0}_1 \otimes A_{10} + \ket{1}\bra{1}_1 \otimes A_{11}$, for some $A_{00}, A_{01}, A_{10}, A_{11}$, then the hypothesis of the lemma is equivalent to
    \begin{equation}\label{eq:density-matrices-bound}
    \begin{split}
        &\norm{A_{00}} \leq \eps \\
        &\norm{A_{11}} \leq \eps \\
        &\frac{1}{2}\norm{A_{00} + A_{01} + A_{10} + A_{11}} \leq \eps \\
        &\frac{1}{2}\norm{A_{00} - iA_{01} +i A_{10} + A_{11}} \leq \eps \,.
    \end{split}
    \end{equation}
    From \cref{eq:density-matrices-bound}, we can deduce that 
    \if\widemargin=0
    \begin{equation*}
    \begin{split}
        \norm{A_{01}} & = \norm{ \frac{1}{2}(A_{00} + A_{01} + A_{10} + A_{11}) + \frac{i}{2}(A_{00} - iA_{01} +i A_{10} + A_{11}) - \frac{1+i}{2}A_{00} - \frac{1+i}{2}A_{11}}\\ &\leq \eps + \eps + \frac{\sqrt{2}}{2} \eps + \frac{\sqrt{2}}{2} \eps \\
        &\leq (2+\sqrt{2}) \eps
    \end{split}
    \end{equation*}
    \else
    \begin{equation*}
    \begin{split}
        \norm{A_{01}} & = \bigg\| \frac{1}{2}(A_{00} + A_{01} + A_{10} + A_{11}) + \frac{i}{2}(A_{00} - iA_{01} +i A_{10} + A_{11}) \\&\quad- \frac{1+i}{2}A_{00} - \frac{1+i}{2}A_{11}\bigg\|\\ &\leq \eps + \eps + \frac{\sqrt{2}}{2} \eps + \frac{\sqrt{2}}{2} \eps \\
        &\leq (2+\sqrt{2}) \eps
    \end{split}
    \end{equation*}
    \fi
    Similarly we have $\norm{A_{10}} \leq \parens{2+\sqrt{2}}\eps$, so
    \begin{equation*}
    \begin{split}
        \norm{A} &\leq \norm{A_{00}} + \norm{A_{01}} + \norm{A_{10}} + \norm{A_{11}} \leq (6+2\sqrt{2}) \eps < 10 \eps \,.
    \end{split}
    \end{equation*}
\end{proof}
\fi
\end{document}